\documentclass[10pt, letterpaper]{article}
\pdfoutput=1
\title{Learning Large Electrical Loads via Flexible Contracts with Commitment}
\usepackage{graphicx}
\usepackage{subfigure}
\usepackage{amsmath}
\usepackage{amsthm}
\usepackage{algpseudocode}
\usepackage{algorithm}
\usepackage{amsfonts}
\usepackage{subfigure}
\usepackage{url}
\usepackage{setspace}
\usepackage{authblk}
\usepackage{xcolor}

\date{}
\newtheorem{theorem}{Theorem}[section]
\newtheorem{lemma}[theorem]{Lemma}
\newtheorem{claim}[theorem]{Claim}

\newtheorem{definition}[theorem]{Definition}
\newtheorem{assumption}[theorem]{Assumption}
\newtheorem{proposition}[theorem]{Proposition}
\newtheorem{corollary}[theorem]{Corollary}
\DeclareMathOperator{\E}{\mathbb{E}} 
\DeclareMathOperator*{\argmin}{arg\,min}
\algnewcommand{\LineComment}[1]{\State \(\triangleright\) #1}
\usepackage[left=3cm, right=3cm, top=2cm]{geometry}
\begin{document}
\author[1]{Pan Lai, Lingjie Duan, Xiaojun Lin \footnote{Pan Lai and Lingjie Duan are with the Engineering Systems and Design Pillar, Singapore University of Technology and Design, 487372, Singapore (e-mail: llaipann@gmail.com, lingjie\_duan@sutd.edu.sg), Xiaojun Lin is with School of Electrical and Computer Engineering, Purdue University, 465 Northwestern Ave, West Lafayette, IN 47907-2035, U.S.A. (linx@ece.purdue.edu).}}

\maketitle
\begin{abstract}
Large electricity customers (e.g., large data centers) can exhibit huge and variable electricity demands, which poses significant challenges for the electricity suppliers to plan for sufficient capacity. Thus, it is desirable to design incentive and coordination mechanisms between the customers and the supplier to lower the capacity cost. This paper proposes a novel scheme based on flexible contracts. Unlike existing demand-side management schemes in the literature, a flexible contract leads to information revelation. That is, a customer committing to a flexible contract reveals valuable information about its future demand to the supplier. Such information revelation allows the customers and the supplier to share the risk of future demand uncertainty. On the other hand, the customer will still retain its autonomy in operation. We address two key challenges for the design of optimal flexible contracts: i) the contract design is a non-convex optimization problem and is intractable for a large number of customer types, and ii) the design should be robust to unexpected or adverse responses of the customers, i.e., a customer facing more than one contract yielding the same benefit may choose the contract less favorable to the supplier. We address these challenges by proposing sub-optimal contracts of low computational complexity that can achieve a provable fraction of the performance gain under the global optimum. 




\end{abstract}
\section{Introduction}

Large electricity customers can exhibit huge and uncertain electricity demand, which poses new challenges to the electricity supplier. Some commercial and industrial entities consume a huge amount of electricity. For example, in 2014, Google data centers consumed 4.4 billion KWh of electricity, which is enough to power 366,903 US households \cite{Google-Data-Center}. Moreover, these large customers' demand can change dramatically over time due to the internal scheduling of their own business operations and workload. For example, the power usage of an IBM data center in the same hour can change by $50\%$ over different days of a week \cite{IBM-Data-Center}. Their local generation (through fossil-fuel and renewable sources) further adds variability to the net-demand seen by the supplier. The electricity supplier has to provision enough capacity so that it has adequate generating resources to meet the demand at all times. However, capacity is costly \cite{Kuser}. For example, in areas where ISOs (Independent System Operator) run capacity markets (e.g., NYISO, PJM) \cite{Capacity-market}, the suppliers have to shoulder significant capacity costs based on their peak demand. When a large fraction of the customer's load is highly variable and uncertain, the supplier has to incur a much higher cost to over-provision capacity in order to ensure that demand and supply can be balanced at all time. 

Generally, the customers have better knowledge of their workload schedule and electricity demands than the supplier. Thus, one promising way for the supplier to reduce the capacity cost is to overcome this information asymmetry. That is, the supplier can better predict the future capacity needs if it can learn useful information about the future demand of the customers. However, there is no systematic study in the literature to investigate how to motivate the customers to reveal their future demand information to the supplier. Without proper incentives, customers are unwilling to reveal their private information as they may lose the freedom to adjust their demands in the future. Thus, proper design of incentive mechanisms is crucial. In this work, we propose a new approach using \emph{flexible contracts}. A flexible contract sets a lower price of electricity and a commitment range for future demand ahead of time. We use the term ``flexible contract" to contrast with ``strict contract": the latter sets a discounted price for a \emph{specific} amount of future demand. Indeed, ``flexible contract" is more flexible than ``strict contract". For a customer that commits to a flexible contract, it will enjoy the price discount as long as its future net-demand is \emph{within the committed demand range}. At the same time, the supplier can better estimate the future demand from the information revealed by the flexible contract and better prepare the capacity. In this way, the contract leads to a win-win situation to both the supplier and the customers. 

Although there exist many demand-side management schemes in the literature, the above information revelation capability is unique to flexible contracts and its has the distinctive advantage to enable the customers and the supplier to share the risk of future demand uncertainty, while still allows the customers to retain their autonomy. In contrast, existing schemes often expose the risks of uncertainty mainly towards either the supplier or the customers, or forcibly intervene in the customers' internal operations. For example, at one extreme, the supplier may pass the risks to the customers by charging them with highly-dynamic prices (e.g., real-time pricing or critical peak pricing \cite{Roozbehani, Wang}) or very high price on peak demand (e.g., the peak-based pricing \cite{Zhao}). While these pricing schemes could reshape demand, it causes significant financial uncertainty to the customers' normal operations, especially when their own demand is uncertain. In comparison, our proposed flexible contracts allow the customers to accept an controlled amount of future risk in exchange for a better price. As we will show in Section \ref{sec-numerical-result}, our flexible contract can significantly lower capacity cost compared to the peak-based pricing. At the other extreme, static pricing schemes (e.g., the time-of-use pricing) do not provide a feedback loop for the supplier to learn the private operation decisions of the customers, and thus expose significant risks to the supplier. In comparison, our flexible contracts enable the supplier to learn and exploit the private workload information revealed by the customers. Another line of work argues for direct-load control (i.e., the supplier directly controls the energy consumption of the customers \cite{Hertzog}) or enforcing quotas on the energy consumption of the customers \cite{Duan}. These schemes are more intrusive to the customers' privacy and limit operation flexibility. In contrast, our flexible contracts still provide the customers with the freedom to reveal different levels of their private information as well as choose different demand variation ranges. Our idea of flexible contracts is also related to swing contracts \cite{LiW}. However, \cite{LiW} studies the transmission level, and the focus is on how an ISO can clear the market with offers specified by swing contracts. In contrast, we study the distribution level, and our focus is on how to motivate customers to truthfully pick the specified contract that reveals their true capacity need.

We summarize the key novelty and main contributions as follows. First, we propose the flexible contract as a novel approach to learn and exploit customers' private demand information by providing them with electricity price discounts as incentives. Our model is comprehensive because it captures customers' demand diversity both in the mean and variation, and further incorporates their demand elasticity. 

Second, we formulate the flexible contract design problem as a bi-level optimization problem to maximize the supplier's profit by taking into account the customers' preferred contract choices. The optimization problem, however, is non-convex and intractable for a large number of customer types. To address this challenge, we propose an approximate contract design to achieve at least $\frac{1}{2}$ of the maximum performance gain achieved by the optimal contract. 

Third, we take into account the situation where a customer may face more than one contract yielding the same benefit. For such situation, existing literature often assumes that the customer will always choose, from those with equal benefit, the contract most favorable to the supplier. Such an \emph{optimistic} assumption may be unrealistic. In contrast, we further study the optimal contract design in the \emph{pessimistic} setting where the customer will choose the contract least favorable to the supplier. Even in this setting, we design a robust contract to still achieve at least $\frac{1}{3}$ of the maximum possible profit gain.


Finally, we evaluate the empirical performances of our proposed flexible contracts and compare it to a typical pricing scheme, i.e., the peak-based pricing. The flexible contracts show significant performance gains thanks to information revelation. Specifically, our contract 
scheme not only increases the supplier's profit, but also reduces each customer's cost, thus achieves a win-win situation. 

We note that our proposed flexible contract scheme is a way to design pricing mechanism so that customers have an incentive to reveal their expected consumption range. There may be other ways to design such pricing mechanisms for information revelation. However, for any such pricing mechanisms (including our proposed flexible contract), there is the risk that the consumers may untruthly report their information. Thus, we will face the same technical challenge of how to design the pricing mechanisms under such a risk. Our key message is that, while the optimal contract can be difficult to design, approximate solutions can be found with low complexity and can yield good results. Our study thus reveals important and useful insights for other pricing mechanisms in the future, which will face similar difficulty.

\section{System Model and Problem Formulation}\label{sec-model}
We next present the models of customers and the supplier as well as the problem formulation for the flexible contract.  
\subsection{Customers' Demands and Costs}\label{subsec-variability-flexibility}
We first present the model for the customers with the crucial feature of information uncertainty. There are $N$ customers connecting to the supplier. We are interested in large customers (e.g., Google data center, IBM data center), the number of which is small in practice \cite{Smolaks}. These customers purchase electricity from the supplier to meet their net-demands. For a future time period (e.g., peak hours) of interest, we assume that each customer's mean usage can take values $m_1, m_2,\ldots, m_n$ \footnote{Our flexible contract design can be extended to a continuous distribution for the mean usage and still achieve a reasonable performance (See Section \ref{sec-extension}).} with probability $h(m_1),\ldots, h(m_n)$, respectively, where $\sum_{i=1}^{n}h(m_i)=1$. We refer to a customer as type-$m_i$ if its mean usage is $m_i, i=1,\ldots, n$. Without loss of generality, we assume that $m_1<m_2<\ldots<m_n$. To model information asymmetry, we assume that a customer knows its own mean demand and thus its type according to its scheduled workload (e.g., due to a data center's received computing jobs from its regular subscribers). However, \emph{the supplier only knows the distribution $h(\cdot)$, but not the exact type of a customer}. The actual net-demand of a type-$m_i$ customer can still deviate from its mean $m_i$ (e.g., unexpected task arrivals or cancellation in a data center). We assume that each customer has a maximum variation degree $\Delta\in [0,1]$, which is random according to the distribution $f(\Delta)$. We assume that, given $\Delta$, the (realized) demand $x$ of a type-$m_i$ customer is a random variable with probability density function $\rho(x)$ in the range $[m_i(1-\Delta), m_i(1+\Delta)]$. Again to model information asymmetry, we assume that the customer knows $\Delta$ but the supplier only knows the probability density function $f(\Delta)$. 

Our customer model can also incorporate demand elasticity. Facing the realized demand $x$, a customer can adjust the demand to $x'$ (e.g., it can reschedule or reject some computing tasks at some cost). If $x'< x$, i.e., the actual demand is cut down to $x'$, the customer incurs an elasticity cost $k(x-x')$ at unit cost (penalty) $k$. On the other hand, if $x'> x$, the cost can be trivial since the customer can just turn on more machines or servers.  

A customer's goal is to minimize its overall cost, including both the demand elasticity cost (only if $x'<x$) and the cost to purchase energy from the supplier. The latter depends on the pricing schemes chosen by the customer, as elaborated below.

\subsection{Flexible Contracts for Information Revelation}\label{subsec-contract}
Before proposing the flexible contracts, we first introduce the existing baseline pricing employed by the supplier. In the baseline pricing scheme, note that the supplier represents the utility, or LSE (Load Serving Entity), which buys electricity from generators and resells electricity to its customers. The supplier announces a fixed electricity price $p_0>0$ to all customers and a customer's cost $p_0x'$ is proportional to its actual usage $x'$. To avoid the trivial case that a customer does not purchase any electricity, we assume $p_0<k$, which implies that the customer will always decide $x'=x$ under the baseline scheme because using demand elasticity only incurs a higher cost.

With only the baseline pricing scheme, the supplier will have to face significant risk in capacity needs if the future demand of the customers is high. To address the issue, in addition to the baseline pricing scheme, the supplier also introduces flexible contracts to encourage customers to indirectly reveal their private information through their contract choices. Given that there are $n$ types of customers, it is enough for the supplier to set $n$ contract options, that is one for each type. The contract option designed for type-$m_i$ customers is denoted as $(p_i,\delta_i, \bar{p}_i)$, and it is centered at $m_i$. If a customer's actual demand is $x'$, its cost for buying $x'$ amount of electricity is given by 
\begin{equation*}\small
\vspace{-0.2cm} 
\bar{c}(x')=
  \begin{cases}
    m_i(1-\delta_i)p_i,\: \quad\quad\quad\quad\quad\quad\quad\quad\quad\text{if}\: x'<m_i(1-\delta_i)\\
    x'p_i,\:\quad\quad\quad\quad\quad\quad\quad\text{if}\: x' \in [m_i(1-\delta_i), m_i(1+\delta_i)] \\
	  x'\bar{p}_i+m_i(1+\delta_i)\big(p_i-\bar{p}_i\big),\:\quad\quad\quad\text{if}\: x' >m_i(1+\delta_i),\\
  \end{cases}
\end{equation*}
which is illustrated in Figure \ref{price-function-contract}. Here, $p_i<p_0$ is the discounted energy price if the real demand of a customer is in the commitment range $[m_i(1-\delta_i),m_i(1+\delta_i)]$, where $\delta_i\in [0,1]$, and $\bar{p}_i$ is the unit penalty (price) for extra demand beyond $m_i(1+\delta_i)$. If the realized demand is below the lower bound $m_i(1-\delta_i)$, it also incurs a cost that is equivalent to the cost of consuming $m_i(1-\delta_i)$ amount of electricity. 
\begin{figure}[]
\vspace{-0.1in}
\centering
\includegraphics[height=6cm,width=8cm]{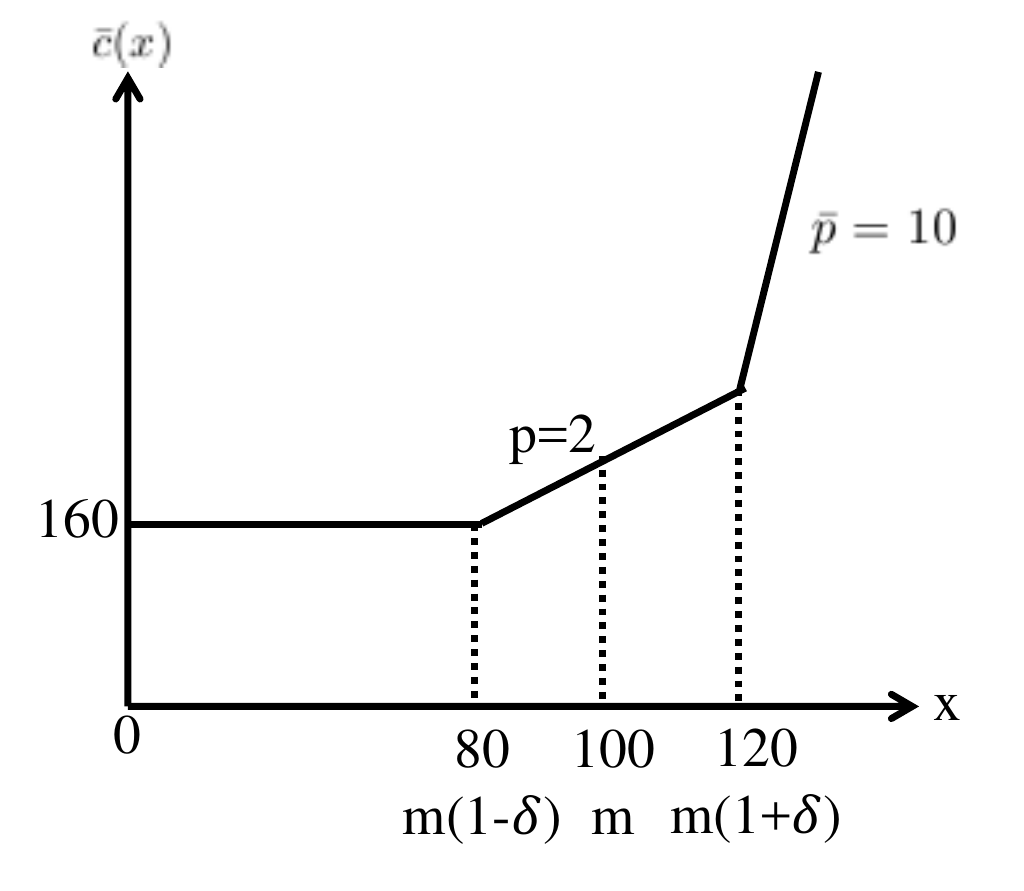}
\caption{Illustration of a customer's cost $\bar{c}(x')$ versus actual demand $x'$ under type-$m_i$ flexible contract $(p_i,\delta_i,\bar{p}_i)$.}
\label{price-function-contract}
\end{figure}

A customer has a lower cost if its actual demand is in $[m_i(1-\delta_i), m_i(1+\delta_i)]$, and otherwise the cost may be higher than the baseline cost. Thus, a type-$m_i$ customer will carefully choose the contract according to its mean usage $m_i$ and maximum variation $\Delta$. It should be noted that the customer may leverage its demand elasticity to cut down the overage demand from $x$ to $x'$, and this lowers its cost to $\bar{c}(x')+k\max(x-x',0)$. A customer's contract choice thus reveals its private information about internal demand, which is useful for the supplier to provision the capacity. 

\subsection{The Supplier's Profit}\label{subsec-supplier-cost}

The supplier's cost consists of two parts. First, there is an electricity generation cost to meet the actual total demand of the customers. We assume that the generation cost per unit electricity is $c_0$ as in \cite{Lu}, which should be smaller than $p_0$. (Otherwise, the supplier's profit is always negative.) Second, as the supplier must plan for enough capacity to meet the highest-possible future demand, there is a capacity cost reflecting the costs for capacity market payment. If the supplier needs $y$ units of capacity, in general the capacity cost $g(y)$ is increasing in $y$, since to provide more capacity resource, more investment on the generating stations is needed. 


In the baseline pricing scheme, the supplier does not know a customer's mean demand $m$ and variation $\Delta$. Thus, it has to prepare for the worst case when preparing capacity. As the maximum possible value of $m_i$ is $m_n$ and the maximum possible variation value of $\Delta$ is $100\%$, the capacity prepared for any type-$m_i$ customer is $2m_n$. In the proposed contract, not all type-$m_i$ customers will pick the contract option $i$. Indeed, a type-$m_i$ customer chooses the option $i$ only if the cost choosing the option is no larger than that choosing other options or the baseline pricing scheme (see Proposition \ref{pro-equilibrium-P1}). However, once a customer chooses the contract option $(p_i,\delta_i,\bar{p}_i)$ with a sufficiently high penalty price $\bar{p}_i>k$, the supplier only needs to provide capacity $m_i(1+\delta_i)$ instead of $2m_n$.


\subsubsection{A Motivating Example for Information Revelation}
We now use a motivating example to illustrate why the use of the above-proposed flexible contracts can facilitate information revelation and benefit both the supplier and customers. Assume that there are 10 customers. Each customer can fall into one of the two types: $m_1=b$ (low load), $m_2=3b$ (high load) for a future time period. Further, for simplicity assume that given a customer's type is $m_1$ or $m_2$, its additional variation $\Delta =0$. However, the supplier does not know the type of each customer nor its value of $\Delta$. Therefore, under the baseline scheme, the supplier has to provision the maximum capacity $2 m_2$ for every customer. The corresponding capacity cost is $g(20 m_2)$. At the same time, each type-$m_1$ customer pays $p_0 m_1$, and each type-$m_2$ customer pays $p_0 m_2$. Assume that each customer can be of type-$m_1$ (or type-$m_2$) with probability $h(m_1)$ (or $h(m_2)$). Thus, the total amount of energy consumed by the customers is $10h(m_1)\cdot m_1+10h(m_2)\cdot m_2$, and the total energy cost is $c_0\big(10h(m_1)\cdot m_1+10h(m_2)\cdot m_2\big)$. The profit of the supplier is 
\begin{equation}\label{eqn-profit-baseline}
(p_0-c_0)\big(10h(m_1)\cdot m_1+10h(m_2)\cdot m_2\big)-g(20m_2).
\end{equation}

Now, the supplier designs two contracts, one centered at $m_1$ and the other centered at $m_2$, both with parameters $p_i=0.9p_0$, $\delta_i = 0.1$, and $\bar{p}_i=1000 p_0$. Thus, each type-$m_1$ (or type-$m_2$) customer will choose the contract for $m_1$ (or $m_2$), and its payment is reduced to $0.9p_0m_1$ (or $0.9p_0 m_2$). Thus, all customers benefit from lower cost. Once the supplier sees the contract options that all customers pick, it only needs to provision a much lower capacity of $m_1(1+\delta_1)=1.1 m_1$ (or $1.1 m_2$) for each type-$m_1$ (or type-$m_2$) customer. In this way, its capacity cost is reduced to $g(10h(m_1)\cdot 1.1 m_1 + 10h(m_2)\cdot 1.1 m_2)$. In addition, the total energy cost remains at $c_0\big(10h(m_1)\cdot m_1+10h(m_2)\cdot m_2\big)$. Thus, the total profit of the supplier then becomes 
\begin{eqnarray}
&&(0.9p_0-c_0)\big(10h(m_1)\cdot m_1+10h(m_2)\cdot m_2\big)\nonumber\\
&&-g\big(10h(m_1)\cdot 1.1 m_1 + 10h(m_2)\cdot 1.1 m_2)\big).\label{eqn-profit-contract}
\end{eqnarray}

Figure \ref{profit} compares the supplier's profit under the flexible contracts in (\ref{eqn-profit-contract}) with that under the baseline pricing in (\ref{eqn-profit-baseline}) when $c_0 =0.2p_0, g(y)=\hat{c}y=0.1p_0$, as the fraction $h(m_1)$ of type-$m_1$ customer increases. Clearly, the profit of the supplier under the flexible contracts is higher than that under the baseline pricing. Indeed, as long as $\hat{c}>\frac{0.1\big(10h(m_1)m_1+10h(m_2)m_2\big)p_0}{20m_2-10h(m_1)\cdot 1.1m_1-10h(m_2)\cdot 1.1m_2}$, the supplier's profit under the flexible contracts is always higher than that under the baseline pricing (note that this is true independent of the value of unit energy-cost $c_0$). This condition is easy to satisfy, e.g., we only need $\hat{c}>0.0025p_0$ when $m_1=b$, $m_2=3b$, $h(m_1)=0.9$, $h(m_2)=0.1$. Overall, the above example shows that, thanks to information revelation, the flexible contracts scheme not only increases the supplier's profit, but also reduces the customer's cost compared to the baseline pricing.


\begin{figure}[]
\vspace{-0.1in}
\centering
\includegraphics[height=6cm,width=8cm]{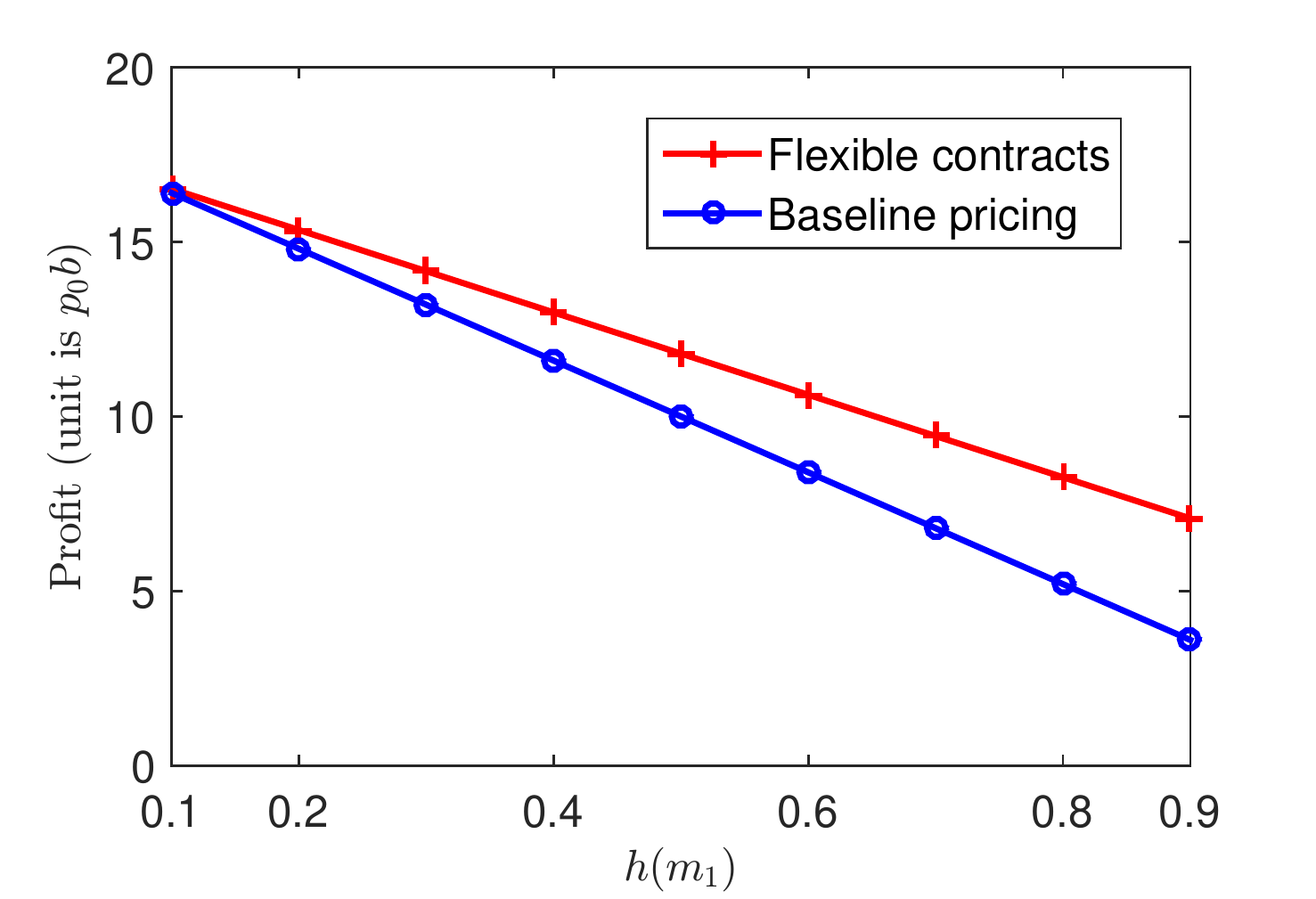}
\caption{The supplier's profit as a function of the probability $h(m_1)$ that a customer is type-$m_1$ customer under the flexible contracts and baseline pricing, respectively.}
\label{profit}
\end{figure}

\subsubsection{Mathematical Expression for the Supplier's Profit}
We now provide a mathematical expression for the supplier's profit (i.e., revenue minus cost) under the flexible contracts. For each customer $j=1,\ldots, N$, let $m(j)$ be its mean demand and $\Delta(j)$ be its variation level. Let $i(j)$ be the contract option picked by the \emph{j}-th customer (as in Section \ref{subsec-contract}), and we use $i(j)=0$ for the baseline pricing scheme. Let $e(j)$ be the customer's expected energy consumption depending on the chosen contract option, which is calculated by averaging all possible realized demand $x$. Let $\pi(j)$ be the capacity needed for customer $j$. With sufficiently high penalty price, i.e., $\bar{p}_{i(j)}>k$, we have $\pi(j)=m_{i(j)}(1+\delta_{i(j)})$ if the customer chooses contract option $i(j)\neq 0$; Otherwise, $\pi(j)=2m_n$ when choosing the baseline pricing scheme. Let $r(j)$ be the customer's expected energy payment to the supplier for buying energy, with the expectation taken over the real demand $x$. To sum up, the total profit to the supplier is given by
\begin{equation}\label{eqn-profit-complex}
P(\Phi)=\E\bigg(\sum_{j=1}^{N}r(j)-c_0\sum_{j=1}^{N}e(j)-g[\sum_{j=1}^{N}\pi(j)]\bigg)
\end{equation}
where the expectation is taken with respect to the distribution of each $m(j)$ and $\Delta(j)$. 

\subsection{Problem Formulation for Optimal Contract Design}\label{subsec-problem-formulation}
We now formulate the optimal contract design problem, where contract $(p_i,\delta_i,\bar{p}_i)$ is designed for type-$m_i$ customers. We use a Stackelberg game to study the interactions between the supplier and the $n$ types of customers. Let $\mathcal{I}=\{1,2,\ldots, n\}$ denote the set of all possible customer types. As the Stackelberg leader, the supplier first decides the contract options $\Phi=\{(p_i,\delta_i,\bar{p}_i), i\in\mathcal{I}\}$ in Stage I with the goal to maximize its expected profit. Note that the baseline price $p_0$ is not a decision variable for the supplier, as our focus is the contract design for any given $p_0$. After that, each customers chooses either the baseline pricing scheme or a particular flexible contract to minimize its cost in Stage II. 


We note that it is non-trivial to design the contract parameters $\delta_i$, $p_i$ and $\bar{p}_i$. If $\delta_i$ is too small, it imposes too much restriction on the customers' future demand, discouraging many customers to participate in the contracts and to reveal their demand information. If $\delta_i$ is too large, the supplier only learns coarse information on each customer's variation, and thus cannot benefit from capacity reduction. The choice of $p_i$ also has a tradeoff. While a low price stimulates high subscription rate from customers and helps information learning for saving capacity, it provides low energy revenue to the supplier. The choice of penalty price $\bar{p}_i$ has a similar tradeoff as $p_i$. Next, we introduce a contract formulation that captures the tradeoffs. According to the revelation principle \cite{Fudenberg}\cite{Mas-Colell}, it is sufficient to focus on incentive-compatible design of the contract options, which requires each customer to truthfully choose the contract option designed for its own type. Specifically, for a type-$m_i$ customer with variation $\Delta$, we denote its expected cost when choosing the baseline pricing as $\E[C_0(m_i,\Delta)]$ and denote its expected cost when choosing the type-$m_j$ contract as $\E[C_j(m_i,\Delta)]$. Then, the incentive compatibility requirement can be expressed as follows.
\begin{definition}
A contract $\Phi=\{(p_i,\delta_i,\bar{p_i}), i\in\mathcal{I}\}$ satisfies the incentive compatibility (IC) condition if 
\begin{equation}\label{eqn-mechanism-design}\small
\begin{gathered}
\text{for any two types}\: i\neq j,\: \text{and any}\: \Delta\in [0,1],\\
\min\{\E[C_i(m_i,\Delta)], m_ip_0\}\leq \min\{\E[C_j(m_i,\Delta)], m_ip_0\}.\\
\end{gathered}
\end{equation}
\end{definition}

In other words, if IC condition holds, a customer of type-$m_i$ will either pick its dedicated contract for type-$m_i$ or the baseline. Picking any other contract option $j\neq i$ will not yield lower cost for a type-$m_i$ customer.

Our goal is then 
\begin{equation}\label{problem-P1}
\max_{\Phi} P(\Phi)=\E\bigg(\sum_{j=1}^{N}r(j)-c_0\sum_{j=1}^{N}e(j)-g[\sum_{j=1}^{N}\pi(j)]\bigg),
\end{equation}
subject to (\ref{eqn-mechanism-design}) and
\begin{equation}\label{eqn-contract-constraint}
\vspace{-0.2cm} 
p_i\leq p_0, \text{and}\: \delta_i\in [0,1], \text{for any}\: i\in\mathcal{I}.
\end{equation}

We will refer to Problem (\ref{problem-P1}) as Problem $\mathbb{P}_1$ in the rest of the paper. To solve the optimal contract design problem, we face two key challenges. 
\begin{itemize}
\item \emph{Challenge I: Tractability.} The nonlinear IC condition in (\ref{eqn-mechanism-design}) is not a convex set and the objective is not concave (as will be shown in Section \ref{sec-detail}). Hence, the problem is intractable for a large number of customer types.  
\item \emph{Challenge II: Robust Contract Design.} Prior literature on contract design using a similar IC condition assumes the \emph{optimistic} setting: when a customer faces more than one contract options with the same cost, it will always pick its dedicated contract option, which is also the one most favorable to the supplier \cite{Mas-Colell}\cite{Gao}\cite{Navabi}. However, this assumption may not hold in practice considering customers' selfish or adverse behavior. 
\end{itemize}
We will present our solutions to deal with both challenges in Sections \ref{sec-challenge-I} and \ref{sec-challenge-II}.


For ease of reading, we first list the notations in Table I, which are used in this section as well as the next several sections.
\begin{table}\label{table-notation}
\caption{}
\centering
\begin{tabular}{|c|c|}
\hline
Notations & Definitions\\
\hline
$N$ & Number of customers \\
$n$ & Number of contract options\\
$m_i$ & Mean demand of a type-$m_i$ customer\\
$h(m_i)$&Probability that a customer is of type-$m_i$\\
$\Delta$ & Variation degree of a customer's demand\\
$f(\Delta)$ &Probability density function of $\Delta$\\
$\rho(x)$ &Probability density function of a customer's demand\\
$k$&Elasticity cost coefficient of customers\\
$p_0$&Baseline price\\
$p_i$ &Discounted price in the contract option $i$\\
$\delta_i$ &Contract option $i$'s variation range size\\
$\bar{p}_i$&Penalty in the contract option $i$\\
$c_0$&Unit electricity generation cost of supplier\\
$g(y)$ &Cost function of supplier for reserving capacity $y$\\
$\hat{c}$&Unit capacity cost of the supplier where $g(y)$ is linear\\
$\Phi$& A contract\\
$\Phi'$& Final contract in the optimistic scenario\\
$\Phi''$& Final contract in the pessimistic scenario\\
$\E[C_j(m_i,\Delta)]$&Type-$m_i$ customer's cost choosing contract option $j$\\
$P$& Supplier's profit under flexible contracts\\
$P_0$& Supplier's profit under baseline pricing scheme\\
$P^*$& Supplier's optimal profit under Problem $\mathbb{P}_1$\\
$P^H_i$& Supplier's profit under high penalty regime\\
$P^L_i$& Supplier's profit under low penalty regime\\
$\hat{P}^*$& Supplier's super-optimal profit\\
$\Delta_{th,i}$&Variation degree threshold \\ 
\hline
\end{tabular}
\end{table}%
\section{Detailed Formulation of Problem $\mathbb{P}_1$}\label{sec-detail}

In this section, we derive detailed expressions for Problem $\mathbb{P}_1$. Note that $\mathbb{P}_1$ is actually a bi-level optimization problem. We will use backward induction to first analyze a customer's option choice in Stage II, and then derive the supplier's profit in Stage I for optimal contract design.

For ease of exposition as well as analysis tractability, we further make the following Assumption \ref{assumption-variation} that each customer's demand and variation both follow a uniform distribution. Although this assumption may seem restrictive, it allows us to perform an in-depth study of the problem, and reveal new insights and non-trivial structure for the interaction between the customers and the supplier. Our results can be extended to other distributions for customer's demand and variation without major change of the revealed engineering insights (See Section \ref{sec-extension}). 
\begin{assumption}\label{assumption-variation}
Each customer's demand follows a uniform distribution, that is $\rho(x)=\frac{1}{2m_i\Delta}$ for $x\in [m_i(1-\Delta), m_i(1+\Delta)]$, and its variation $\Delta$ follows a uniform distribution in $[0,1]$.
\end{assumption}

For ease of analysis tractability, we also make the following assumption \ref{assumption-capacity-cost} that the supplier's capacity cost is linear with the amount of capacity prepared in advance, which enables us to reveal new insights. We believe that the insights under the restrictive assumption could potentially be generalized as well. 
\begin{assumption}\label{assumption-capacity-cost}
The supplier's capacity cost is linear with the amount of capacity prepared in advance, that is, $g(y)=\hat{c} y$.
\end{assumption}
Besides Assumption \ref{assumption-capacity-cost}, we assume $\hat{c}\leq\frac{1}{2}p_0$. Otherwise, the supplier's profit $P_0$ under the baseline pricing scheme is negative, which is impractical. To see this, note that 
\begin{equation}\label{eqn-P0}\small
\vspace{-0.1cm} 
P_0=\sum_{i=1}^{n}Nh(m_i)m_ip_0-2Nm_n\hat{c}-\sum_{i=1}^{n}Nh(m_i)c_0m_i,
\end{equation}
where we have used the assumption that each customer $j=1,\ldots, N$ will be of type $m_i$ with probability $h(m_i)$. If $m_i<m_n$ and $\hat{c}>\frac{1}{2}p_0$, we would have $P_0<Nm_np_0-2Nm_n\hat{c}-\sum_{i=1}^{n}Nh(m_i)c_0m_i<0$, which is impractical. Finally, since $k>p_0$ as discussed in Section \ref{subsec-contract}, we must have $k\geq 2\hat{c}$.

\subsection{Customers' Decisions in Stage II}\label{stage-II}

We first analyze a type-$m_i$ customer's expected cost $\E[C_i(m_i,\Delta)]$ when it chooses the dedicated contract option $i$. Its choice depends on its demand variation $\Delta$ and the contract's committed variation $\delta_i$. If $\Delta<\delta_i$, the customer's demand is always within the contract range and its expected cost is $\E[C_i(m_i,\Delta)]=m_ip_i$. On the other hand, if $\Delta>\delta_i$, its expected cost is larger because its random demand may exceed the commitment range. Recall that the future demand is $x\in [m_i(1-\Delta), m_i(1+\Delta))]$ with the probability density function $\rho(x)$. The customer's cost also depends on whether it employs the demand elasticity to change realized demand $x$ to actual demand $x'$. Specifically, the customer will decide $x'$ in the following way: 
\begin{itemize}
\item If $x\in [m_i(1+\delta_i), m_i(1+\Delta)]$, the customer will leverage elasticity to reduce demand to $m_i(1+\delta_i)$ only if $k<\bar{p}_i$. If $k\geq \bar{p}_i$, the customer will keep the original demand $x$ and undertake the contract penalty.
\item If $x\in [m_i(1-\delta_i), m_i(1+\delta_i)]$, the customer will not change realized demand but request $x$ from the supplier due to $p_i<k$.
\item When $x\in [m_i(1-\Delta), m_i(1-\delta_i)]$, the customer with insufficient demand will leverage elasticity to increase its demand to $x'=m_i(1-\delta_i)$ without incurring any additional cost\footnote{The customer (e.g., a data center) can easily activate more servers to keep demand at the promised lower bound $m_i(1-\delta_i)$. The incurred cost is trivial.}.
\end{itemize}
By considering all possible $x$, we can compute the expected cost for a type-$m_i$ customer with variability degree $\Delta>\delta_i$ if it chooses the contract option $i$. This expression depends on the relationship between $k$ and $\bar{p}_i$. If $\bar{p}_i>k$, we have 
\begin{equation}\label{eqn-cost-medium-multi-option}\small
\vspace{-0.1cm}
\E[C_i(m_i,\Delta)]=m_ip_i+\frac{m_ik}{4\Delta}(\Delta-\delta_i)^2.
\end{equation}
We can see that, the cost increases with demand variation $\Delta$ and elasticity cost coefficient $k$. As a special case, $E[C_i(m_i, \Delta)]=m_i p_i$ when $\Delta=\delta_i$. Otherwise, if $\bar{p}_i \leq k$,  we similarly have $\E[C_i(m_i,\Delta)]=m_ip_i+\frac{m_i\bar{p}_i}{4\Delta}(\Delta-\delta_i)^2$ by replacing $k$ in (\ref{eqn-cost-medium-multi-option}) by $\bar{p}_i$. 
Similarly, we can analyze the expected cost $\E[C_j(m_i,\Delta)]$ of a type-$m_i$ customer by choosing another contract option $j$, by considering the relationship between ranges $[m_i(1-\Delta), m_i(1+\Delta)]$ and $[m_j(1-\delta_j),m_j(1+\delta_j)]$. As there are many combination cases, here we skip the detailed analysis. 

Under the IC condition in (\ref{eqn-mechanism-design}), a customer will either choose its own contract type or the baseline pricing whichever's cost is lower. Assuming that the IC condition holds, we can derive the customers' optimal behavior as follows. As described earlier, we only focus on $k>p_0$.
\begin{proposition}\label{pro-equilibrium-P1}
At Stage II, observing the supplier's contract option $(p_i, \delta_i, \bar{p}_i)$ and the baseline price $p_0$, type-$m_i$ customers are partitioned into the following two groups, depending on their variation distribution of $\Delta\in [0,1]$:
\hspace{-2cm}
\begin{itemize}
\item Customers of low variation (i.e., $\Delta\in[0, \Delta_{th,i}]$) will subscribe to contract option $i$ to take advantage of price discount, where variation degree threshold $\Delta_{th,i}$ depends on the relationship between $k$ and $\bar{p}_i$ as follows.
\begin{itemize}
\item High penalty regime ($\bar{p}_i>k$): 
\begin{eqnarray}\small\label{eqn-threshold-high}
\!\!\!\!\!\!\!\!\!\Delta_{th,i}\!\!&=&\!\!\min\big(1,[\sqrt{(k\delta_i+2(p_0-p_i))^2-k^2\delta^2_i}\nonumber\\
             \!\!& &\!\!+k\delta_i+2(p_0-p_i)]/k\big),
\end{eqnarray}
which is independent of $\bar{p}_i$. 
\item Low penalty regime ($\bar{p}_i\leq k$): 
\begin{eqnarray}\small\label{eqn-threshold-low}
\!\!\!\!\!\!\!\!\!\Delta_{th,i}\!\!&=&\!\!\min\big(1,[\sqrt{(\bar{p}_i\delta_i+2(p_0-p_i))^2-\bar{p}^2_i\delta^2_i}\nonumber\\
                           \!\!& &\!\!+\bar{p}_i\delta_i+2(p_0-p_i)]/\bar{p}_i\big),
\end{eqnarray}
which is independent of $k$.
\end{itemize}
\item Customers of high variation (i.e., $\Delta\in (\Delta_{th,i},1]$) will subscribe to the baseline scheme to avoid the high over-usage penalty if $\bar{p}_i>k$) or high elasticity cost if $\bar{p}_i\leq k$. 
\end{itemize}
Moreover, $\Delta_{th,i}$ is nonincreasing in $k$ for $k>p_0$.
\end{proposition}

The intuition behind Proposition \ref{pro-equilibrium-P1} is as follows. If a type-$m_i$ customer has a small variation $\Delta$, a large fraction of its demand range is within the contract option $i$'s discounted price range so that it takes advantage of discounted price. Thus, its total expected cost when choosing option $i$ is smaller than that choosing baseline pricing. Otherwise, if it has a large variation, a large fraction of its demand range exceeds the upper bound of the contract option $i$'s discounted price range and incurs high penalty or elasticity cost. Thus, its total expected cost when choosing option $i$ is larger than that choosing baseline pricing. In addition, it is interesting to see that $\Delta_{th,i}$ is nonincreasing with $k$. The underlying reason is as follows. When $k$ is small, the contract option $i$ is in the high penalty regime (i.e., $k<\bar{p}_i$). In this penalty regime, $\Delta_{th,i}$ is decreasing in $k$. To see this, note that customers with medium variation (e.g., $\delta_i\leq\Delta\leq\Delta_{th,i}$) will choose option $i$. As $k$ increases, the expected cost of these customers increases. As a result, some of them will switch to the baseline pricing scheme, and $\Delta_{th,i}$ will decrease. After $k$ exceeds $\bar{p}_i$, the contract option $i$ enters the low penalty regime (i.e., $k\geq\bar{p}_i$). Thus, $\Delta_{th,i}$ remains constant as $k$ grows, since customers do not exercise costly elasticity and the expected cost of the customers with medium variation is independent of $k$.

\subsection{The supplier's Profit in Stage I}\label{stage-I}
Based on customers' decision in Section \ref{stage-II}, we now show that the IC condition can simplify the supplier's profit in (\ref{eqn-profit-complex}).
For each customer $j=1,\ldots, N$, it will be of type $m_i$ with probability $h(m_i)$. Recall Assumption \ref{assumption-variation} that $\Delta$ follows a uniform distribution in $[0,1]$. According to Proposition \ref{pro-equilibrium-P1}, under the IC condition, it will pick contract $i$ with probability $\Delta_{th,i}$ and the baseline scheme with probability $1-\Delta_{th,i}$. The supplier's (expected) profit from type-$m_i$ customers depends on whether contract option $i$ has high or low penalty. 

We first consider the supplier's (expected) profit $P^H_i$ from type-$m_i$ customers if contract option $i$ has high penalty. Note that, regardless of a type-$m_i$ customer's choice between contract option $i$ and baseline pricing scheme, its expected energy consumption is always $e(j)=m_i$. Thus, the energy cost to the supplier is always $c_0m_i$. Further, if it picks the option $i$, its maximum consumption is $\pi(j)=m_i(1+\delta_i)$ and its energy consumption payment is $r(j)=m_ip_i$. Otherwise, if it picks the baseline scheme, its maximum possible consumption is $\pi(j)=2m_n$ (because the supplier does not its type), and its energy payment is $r(j)=m_ip_0$. Thus, we have
\begin{eqnarray}\label{eqn-profit-high-penalty}
P^H_i&=&Nh(m_i)\bigg(\big(m_ip_i\Delta_{th,i}+m_ip_0(1-\Delta_{th,i})\big)-c_0m_i\nonumber\\
     \!\!\!\!& &\!\!\!\!\!-\hat{c}\big(m_i(1+\delta_i)\Delta_{th,i}+2m_n(1-\Delta_{th,i})\big)\bigg).
\end{eqnarray}

Similarly, we can model the supplier's (expected) profit $P^L_i$ from type-$m_i$ customers if contract option $i$ has low penalty. Compared to the above high penalty cost, if a type-$m_i$ customer $j$ picks the option $i$, its maximum consumption is $\pi(j)=m_i(1+\Delta_{th,i})$ instead of $m_i(1+\delta_i)$. Also, if its variation $\Delta\in (\delta_i,\Delta_{th,i}]$, its energy consumption payment is $m_ip_i+\frac{m_i\bar{p}_i}{4\Delta}(\Delta-\delta_i)^2$ instead of $m_ip_i$, and its energy consumption is $m_i\big(\frac{1}{4}\Delta+\frac{\delta^2_i}{4\Delta}+1-\frac{1}{2}\delta_i\big)$ instead of $m_i$. Thus, we have
\begin{eqnarray*}
P^L_i\!\!\!\!&=&\!\!\!\!Nh(m_i)\bigg(\int_{0}^{\delta_i}m_ip_id\Delta+\int_{\delta_i}^{\Delta_{th,i}}\big (m_ip_i\nonumber\\
  & &\!\!\!\!+\frac{m_i\bar{p}_i}{4\Delta}(\Delta-\delta_i)^2\big)d\Delta+(1-\Delta_{th,i})m_ip_0\nonumber\\
  & &\!\!\!\!-\hat{c} \big (m_i(1+\Delta_{th,i})\Delta_{th,i}+2m_n(1-\Delta_{th,i}))-C_e(m_i)\bigg),
\end{eqnarray*}
where 
\begin{eqnarray}
C_e(m_i) \!\!\!\!&=&\!\!\!\! m_ic_0\big (\delta_i+1-\Delta_{th,i}+\frac{1}{8}(\Delta^2_{th,i}-\delta^2_i)\nonumber\\
         & &\!\!\!\! +\frac{\delta^2_i}{4}\ln\frac{\Delta_{th,i}}{\delta_i}+(1-\frac{1}{2}\delta_i)(\Delta_{th,i}-\delta_i)\big).
\end{eqnarray}

Thus, $P(\Phi)$ in (\ref{eqn-profit-complex}) can be rewritten as  

\begin{equation}\label{eqn-profit-P1}\small
P(\Phi)=\sum_{i\in \mathcal{I}^H}P^H_i+\sum_{i\in \mathcal{I}^L}P^L_i,
\end{equation}
where $\mathcal{I}^H$ is the set of contract options with high penalty, and $\mathcal{I}^L$ is the set of contract options with low penalty.

Thus, Problem $\mathbb{P}_1$ is to maximize $P(\Phi)$ in (\ref{eqn-profit-P1}) subject to (\ref{eqn-mechanism-design}) and (\ref{eqn-contract-constraint}).


Proposition \ref{pro-equilibrium-P1} narrows the supplier's attentions to only those customers with small $\Delta\in [0,\Delta_{th,i}]$. By substituting (\ref{eqn-cost-medium-multi-option}) and the expressions of $\E[C_j(m_i,\Delta)]$ to the IC constraints, we can find that the constraints are non-convex. To see this, consider $\E[C_j(m_i,\Delta)]=m_j(1-\delta_j)p_j$ when $m_i(1+\Delta)<m_j(1-\delta_j)$ as an illustrative example. It involves the product term $(1-\delta_j) p_j$, which is neither convex nor concave in decision variables $\delta_j$ and $p_j$. In addition, the objective is not concave since it involves the product term $p_i\Delta_{th,i}$ in (\ref{eqn-profit-high-penalty}), which is not jointly concave in decision variables $p_i$ and $\Delta_{th,i}$. As a result, the optimal contract-design problem (\ref{problem-P1}) becomes intractable especially when the number of customer types $n$ is large. In the next section, we will show how to overcome this analysis difficulty and develop approximate solution with provable performance guarantees. 
   
 
\section{Approach for Solving $\mathbb{P}_1$'s Challenge I}\label{sec-challenge-I}

To address Challenge I for solving Problem $\mathbb{P}_1$, we define a new Problem $\mathbb{P}_2$. Problem $\mathbb{P}_2$ has the same objective as Problem $\mathbb{P}_1$. Its constraint on the optimization variables $(p_i,\delta_i,\bar{p}_i, \Delta_{th,i})$ in (\ref{eqn-mechanism-design}) is reduced to the set 

\begin{eqnarray}\label{eqn-R-hat}\small
\hat{R}&=&\!\!\{\forall i\in\mathcal{I}, (p_i,\delta_i, \bar{p}_i,\Delta_{th,i})|\forall i\in\mathcal{I}, p_i=p_0,\nonumber\\
       & & 0\leq\delta_i=\Delta_{th,i}\leq 1, \bar{p}_i>k\}.
\end{eqnarray}

The intuition behind choosing the smaller set $\hat{R}$ is as follows. First, it sets a high penalty, i.e., $\bar{p}_i>k$ to motivate customers to use their flexibility to reduce load, which saves both their own costs and the supplier's capacity cost. Further, we hypothesize that the contract price should not be significantly lower than the baseline price. Otherwise, the supplier will lose a significant amount of revenue in these flexible contracts. Assuming $p_i \approx p_0$, we have $\delta_i \approx \Delta_{th,i}$ by the equation (\ref{eqn-threshold-high}). The set $\hat{R}$ essentially looks at the case when the above two approximations exactly hold. This scenario of restricting to $\hat{R}$ is useful for the following reasons. To see this, first, under $\hat{R}$, each type-$m_i$ customer's payment is always $m_ip_0$, no matter it chooses the baseline or the flexible contract. As a result, the supplier's total revenue is also at the maximum. Second, under this restricted scenario, the supplier can still significantly save the capacity cost, because those customers with $\Delta\leq\delta_i$ will choose the contract option. Due to the above reasons, we expect that the solution to Problem $\mathbb{P}_2$ will produce a reasonable approximate solution to Problem $\mathbb{P}_1$. 

However, readers will immediately notice that, in set $\hat{R}$, for customer with small $\Delta\leq\delta_i$, choosing its dedicated contract option will produce exactly the same cost as choosing the baseline. Furthermore, it is also possible that the costs to a customer are the same across multiple contract options (i.e., all equal to $m_ip_0$), especially when the mean of these options are close to each other. An issue that immediately arises is why the customer would choose its dedicated contract option in the first place. This question will be the key issue for the next Section \ref{sec-challenge-II}. For this current section, we focus on the easier case that, as long as the IC condition holds, the customer with small $\Delta\leq\delta_i$ will choose its dedicated contract option even when its cost is the same as the baseline or under other contract options. We will then uncover some important structures of the solution that will also be useful later on. 

\begin{proposition}\label{pro-opt-P2}
In Problem $\mathbb{P}_2$, the supplier's optimal contract design $\Phi'=\{(p_i,\delta_i,\bar{p}_i), i\in\mathcal{I}|p_i=p_0, \delta_i=\Delta_{th,i},i\in\mathcal{I}\}$ is of a simple form as follows, depending on the diversity of customers' types and is incentive compatible:
\begin{itemize}
\item If type $m_i$ is close to $m_n$ (i.e., $\frac{m_n}{m_i}\leq\frac{3}{2}$), the optimal contract option $i$ for this type is $(p_0, \frac{m_n}{m_i}-\frac{1}{2}, \bar{p}_i>k)$;
\item If type $m_i$ is not close to type-$m_n$ (i.e., $\frac{m_n}{m_i}>\frac{3}{2}$), the optimal contract option $i$ is $(p_0, 1,\bar{p}_i>k)$. 
\end{itemize}
\end{proposition}

The intuition behind the optimal contract design in Proposition \ref{pro-opt-P2} is as follows. Recall that in the set $\hat{R}$, the payment of each customer to the supplier is always $m_i p_i = m_i p_0$, regardless of whether it chooses the contract or the baseline. Hence, the benefit to the supplier mainly comes from the reduced capacity cost, which can be studied separately for each customer. Under the baseline pricing scheme, the capacity that the supplier has to provision for each customer is always $2m_n$. Thus, if $m_i$ is much lower than $m_n$, enticing the customers to use the contract will lead to much lower capacity requirement, which is $2 m_i$ if $\delta_i = 1$. As a result, it is beneficial to set $\delta_i = 1$ so that all type-$m_i$ customers are willing to pick the contract. However, if $m_i$ is relatively close to $m_n$, the supplier faces the following tradeoff: if $\delta_i$ is high, the capacity reduction for each customer is small; if $\delta_i$ is low, very few customers will choose the contract option due to strict commitment. Neither extreme is good for the supplier to reduce the provisioned capacity. Hence, the value of $\delta_{i}$ needs to be optimized. This optimal $\delta_i$ turns out to be $\frac{m_n}{m_i} - \frac{1}{2}$, which balances the above tradeoff. 

\emph{Sketch of Proof:} To prove Proposition \ref{pro-opt-P2}, note that the profit from each type of customers can be calculated separately in Problem $\mathbb{P}_2$. To solve the optimal contract in Problem $\mathbb{P}_2$, it suffices to maximize $P^H_i$ in (\ref{eqn-profit-high-penalty}) subject to $p_i=p_0, 0\leq \Delta_{th,i}=\delta_i\leq 1$. To do this, we replace $p_i$ by $p_0$, and $\Delta_{th,i}$ by $\delta_i$ in (\ref{eqn-profit-high-penalty}). Through this simplification, $P^H_i$ becomes a quadratic function with only one variable $\delta_i$. We can then optimize the choice of $\delta_i$ easily. 

\emph{Remark:} We see that the approximate contract in Proposition \ref{pro-opt-P2} is independent of the distribution of $h(m_i)$. This independence is due to the restriction of the approximate contract to $\hat{R}$ in Problem $\mathbb{P}_2$. To see this, note that for any contract under $\hat{R}$, each contract option is independent with the other options. Thus, to solve the approximate contract, it suffices to solve each option separately in the optimal contract for the new Problem $\mathbb{P}_2$ where each option is independent of the distribution of $h(m_i)$. In contrast, note that for any contract satisfying the IC condition (\ref{eqn-mechanism-design}) of the optimal contract design problem $\mathbb{P}_1$, each contract option is restricted by other options. Thus, the optimal contract depends on the distribution of $h(m_i)$.
To evaluate the performance of approximate contract $\Phi'$ in Proposition \ref{pro-opt-P2}, we define the following gain ratio $\frac{P(X)-P_0}{P^*-P_0}$ for any contract $X$. Here, $P_0$ is the supplier's traditional profit under the baseline pricing scheme in (\ref{eqn-P0}), $P(X)$ is the supplier's profit under contract design $X$, and $P^*$ is the supplier's optimal profit in the original Problem $\mathbb{P}_1$. The gain ratio tells us how closely the approximate solution can approach the performance gain of the optimal contract over the baseline scheme. 

\begin{proposition}\label{pro-bound-P1}\label{pro-bound-P1}
The gain ratio of contract $\Phi'$ in Proposition \ref{pro-opt-P2} is at least $\frac{1}{2}$. 
\end{proposition}
Proposition \ref{pro-bound-P1} illustrates the good performance of $\Phi'$. The intuition behind this proposition is that, at solution $\Phi'$, the supplier collects the maximum amount of revenue. Although the supplier's profit at solution $\Phi'$ is lower than the optimal, the difference can be bounded. The proof uses similar ideas as the proof of Theorem \ref{thm-bound-P1-pess} in Section \ref{sec-challenge-II}. Since these ideas will be presented in Section \ref{sec-challenge-II}, we omit the details here.
\section{Approach for solving $\mathbb{P}_1$'s Challenge II}\label{sec-challenge-II}
We now turn to Challenge II, which is against the IC condition itself. Note that the IC condition implicitly assumes an optimistic scenario. That is, when the costs of more than one contract option (including baseline pricing option) are equal, the customer will pick the dedicated option designed by the supplier, which is usually most favorable to the supplier. Although this optimistic scenario is widely assumed in the mechanism design literature \cite{Mas-Colell}\cite{Gao}\cite{Navabi}, it may fail in practice. Indeed, the customer may pick the option least favorable to the supplier, which we refer to as the \emph{pessimistic} setting.

To illustrate this issue, consider the approximate contract $\Phi'$ in Proposition IV.1 designed for the optimistic scenario. For type-$m_i$ customer with small $\Delta\leq\delta_i$, choosing its dedicated contract option $i$ will produce exactly the same cost $m_ip_0$ as choosing the baseline. Under the optimistic scenario, it is assumed that the customer chooses option $i$, and the resulting capacity need is $m_i(1+\delta_i)$. However, in practice, the customer may choose the baseline, which results in a larger capacity need $2m_n$. Furthermore, it is also possible that the costs to a customer are the same across multiple contract options (i.e., all equal to $m_ip_0$), especially when the mean of these options are close to each other. Under the pessimistic setting, a type-$m_i$ customer may choose option $j>i$, which is dedicated for type-$m_j$ customers with larger mean demand. As a result, the capacity need is larger than that if the customer chooses its dedicated option $i$ under the optimistic scenario. Hence, in practice, the approximate contract $\Phi'$ may result in poor performance due to customers' selfish behaviors that are unfavorable for the supplier.

In this section, we will relax this strict IC condition and quantify the pessimistic or worst-case performance when a customer (facing more than one contract yielding the same benefit) may not pick the option that is the most favorable to the supplier.

Now we will present our approach to quantify the worst-case performance. We start from the contract $\Phi'$ to Problem $\mathbb{P}_2$ as described in Proposition \ref{pro-opt-P2}, but reduce all the contract prices from baseline price $p_0$ by $\epsilon>0$. This is to ensure that, for customers with $\Delta<\Delta_{th,i}$ given in (\ref{eqn-threshold-high}), choosing contracts are strictly better off than choosing the baseline scheme. For those customers with $\Delta=\Delta_{th,i}$, there is still ambiguity whether they will choose the contracts. However, they are of a probability measure of $0$ and hence do not affect the supplier's cost.
\begin{definition}\label{def-near-Shat}
In the pessimistic scenario, we define contract $\Phi''=\{(p_i,\delta_i,\bar{p}_i),i\in\mathcal{I}\}$ below, which depends on the diversity of customers' types:
\begin{itemize}
\item If type $m_i$ is close to type-$m_n$ (i.e., $\frac{m_n}{m_i}\leq\frac{3}{2}$), contract option $i$ is $(p_0-\epsilon, \frac{m_n}{m_i}-\frac{1}{2}, \bar{p}_i>k)$.
\item If type $m_i$ is much smaller than type-$m_n$ (i.e., $\frac{m_n}{m_i}>\frac{3}{2}$), contract option $i$ is $(p_0-\epsilon, 1,\bar{p}_i>k)$.
\end{itemize}
\end{definition}
Still, among the contract options $\Phi''$ with the identical price $p_0-\epsilon$, a customer may face the same cost when choosing between two different options and it may not choose the one preferred by the supplier. Different from the mechanism design literature, our approach allows such untruthful option selection that allows a customer to choose a contract option different from its type.

To evaluate the performance of $\Phi''$, we analyze the profit of the supplier under $\Phi''$ in the pessimistic setting. Consider a type-$m_i$ customer with variation $\Delta$. If $\Delta>\Delta_{th,i}$ with $\Delta_{th,i}$ given in (\ref{eqn-threshold-high}), it will choose the baseline pricing and the supplier's profit from such a customer is $s_0(m_i, \Delta)=m_ip_0-2m_n\hat{c}-c_0m_i$. On the other hand, if $\Delta<\Delta_{th,i}$, it will choose one of the contract options. If it chooses contract option $j$, let $s_j(m_i, \Delta)$ be the supplier's profit from a type-$m_i$ customer with variation $\Delta$. For example, if a type-$m_i$ customer with $\Delta\leq\Delta_{th,i}$ subscribes to option $i$, its contribution to the supplier's profit is $s_i(m_i, \Delta)=m_ip_i-\hat{c}m_i(1+\delta_i)-c_0m_i$. However, it is still possible that $\E[C_j(m_i,\Delta)]=\E[C_i(m_i,\Delta)]$ if the customer chooses another option $j$. Under the pessimistic setting, the customer chooses the contract option $i^*(m_i,\Delta)$ that results in the lowest profit to the supplier in the worst case. Thus, the supplier's (minimum) expected profit from all customers is
\begin{eqnarray}\small\label{eqn-profit-pess}
\vspace{-0.1cm}
\sum_{i=1}^{n}Nh(m_i)\big(\int_{0}^{\Delta_{th,i}}f(\Delta)s_{i^*(m_i,\Delta)}(m_i, \Delta)d\Delta\nonumber\\
+\int_{\Delta_{th,i}}^{1}f(\Delta)s_0(m_i, \Delta)d\Delta\big).
\end{eqnarray}

\vspace{-0.1cm}
In the following, we will use the objective (\ref{eqn-profit-pess}) for the supplier in the pessimistic setting instead of (\ref{eqn-profit-P1}).

To quantify the performance guarantee, we divide our analysis into three steps.
\subsection{Step 1}
We first determine an upper bound of the optimal profit $P^*$. This upper bound is needed as the optimal profit of Problem $\mathbb{P}_1$ is difficult to solve directly. We derive this upper bound by removing constraint (\ref{eqn-mechanism-design}) of Problem $\mathbb{P}_1$. We refer to the resulting optimal contract design as the ``super-optimal".
\begin{lemma}\label{lem-super}
By removing constraint (\ref{eqn-mechanism-design}) of Problem $\mathbb{P}_1$, the super-optimal contract $\hat{\Phi}$ is given as follows.
\begin{itemize}
\item When the type $m_i$ does not differ greatly from type-$m_n$ (i.e., $\frac{m_n}{m_i}\leq \frac{k-\hat{c}}{k}+\frac{1}{2}$), the contract option $(p_i, \delta_i, \bar{p}_i)$ is $p_i=p_0-\frac{\hat{c}^2}{2(k-\hat{c})}(\frac{2m_n}{m_i}-1)$, $\delta_i=\frac{k-2\hat{c}}{2(k-\hat{c})}(\frac{2m_n}{m_i}-1)$,
and with arbitrarily high penalty $\bar{p}_i>k$. In this case, only type-$m_i$ customers of low variation ($\Delta\leq\Delta_{th,i}=\frac{k}{2(k-\hat{c})}(\frac{2m_n}{m_i}-1))<1$) will choose the contract. The optimal profit collected from type-$m_i$ customers is
\begin{equation*}\small
\vspace{-0.1cm}
\E(P_i)=Nh(m_i)(m_ip_0-m_ic_0-2m_n\hat{c}+\frac{k\hat{c}(2m_n-m_i)^2}{4m_i(k-\hat{c})}).
\end{equation*}
\item When the type $m_i$ differs greatly from type-$m_n$ (i.e., $\frac{m_n}{m_i}>\frac{k-\hat{c}}{k}+\frac{1}{2}$), the contract option is $(p_i, \delta_i, \bar{p}_i)$ with $p_i=p_0-\frac{\hat{c}^2}{k}$, $\delta_i=1-\frac{2\hat{c}}{k}$, $\bar{p}_i>k$. All type-$m_i$ customers will choose the contract ($\Delta_{th,i}=1$). The optimal profit collected from type-$m_i$ customers is
\begin{equation*}\small
\vspace{-0.1cm}
\E(P_i)=Nh(m_i)(m_ip_0-m_ic_0-2m_i\hat{c}+\frac{m_i\hat{c}^2}{k}).
\end{equation*}
\end{itemize}
Moreover, the supplier's super-optimal profit $\hat{P}^*$, which is the sum of the super-optimal profit from all customers, is decreasing in $k$.
\end{lemma}


Lemma \ref{lem-super} has a similar structure as Proposition \ref{pro-opt-P2} yet it has another dimension of freedom for deciding the contract price $p_i$. The underlying intuition is also similar: without the IC condition, the profit contributed by each customer can again be separately optimized. The results can be interpreted as follows. If $m_i$ is much lower than $m_n$, enticing the customers to use the contract will lead to much lower capacity requirement ($m_i(1+\delta_i)\leq 2m_i\leq 2m_n$). Thus, it is beneficial to set $\delta_i=1-\frac{2\hat{c}}{k}$ so that $\Delta_{th,i} = 1$. However, if $m_i$ is relatively close to $m_n$, the supplier again faces the following tradeoff: if $\delta_i$ is high, the capacity reduction for each customer is small; if $\delta_i$ is low, very few customers choose contract options. It turns out the best $\delta_i$ is $\frac{k-2\hat{c}}{2(k-\hat{c})}(\frac{2m_n}{m_i}-1)$, which balances the above tradeoff. The last part of the lemma states that the super-optimal profit is decreasing in $k$. This property is intuitive because as $k$ increases, the cost of a customer with medium variation $\Delta\in [\delta_i,\Delta_{th,i}]$ increases, which may push them out of the contract. As a result, fewer customers will choose contract options and the supplier saves less capacity cost. This increasing property turns out to be quite crucial later on.

\emph{Sketch of Proof:} To prove Lemma \ref{lem-super}, we will prove that at the super-optimality, the supplier only determines a high penalty , i.e., $\bar{p}_i>k$ for each contract option $i$. This allows us to only focus on the high penalty regime of each contract option to solve the super-optimum.
Note that the profit from each type of customers can be calculated separately. To solve the super-optimum, it suffices to maximize $P^H_i$ in (\ref{eqn-profit-high-penalty}) subject to (\ref{eqn-contract-constraint}) and the equation of high penalty regime in (\ref{eqn-threshold-high}). To do this, we replace $p_i$ in (\ref{eqn-profit-high-penalty}) by $\delta_i, \Delta_{th,i}$ according to the equation of high penalty regime in (\ref{eqn-threshold-high}). Through this simplification, $P^H_i$ is now concave with either $\Delta_{th,i}$ or $\delta_i$, but not both jointly. Fortunately, we can sequentially optimize the choice of $\delta_i$ for a fixed $\Delta_{th,i}$, then optimize the one-variable $\Delta_{th,i}$.

\subsection{Step 2}
Next, we focus on the limiting regime when $\epsilon\to 0^+$, and show that the approximation ratio of the solution $\Phi''$ in Definition \ref{def-near-Shat} is no smaller than $\frac{1}{3}$. Note that when $\epsilon\to 0^+$, the solution $\Phi''$ in Definition \ref{def-near-Shat} is virtually the same as solution $\Phi'$ in Proposition \ref{pro-opt-P2} for the optimistic setting. The only difference is that the supplier's profit $P(\Phi')$ is calculated by (\ref{eqn-profit-pess}) in the pessimistic setting instead of by (\ref{eqn-profit-P1}) in the optimistic setting. As in Proposition \ref{pro-bound-P1}, we are interested in the lower bound of the performance gain ratio between $\Phi'$ and the optimal. Towards this end, it suffices to compare $P(\Phi')$ with the super-optimal. The term $\frac{P(\Phi')-P_0}{\hat{P}^*-P_0}$ provides a lower bound on the gain ratio between the solution $\Phi'$ and the super-optimal, which will be shown in Lemma \ref{lem-Shat-bound-P1-pess-support-0}.

\begin{lemma}\label{lem-Shat-bound-P1-pess-support-0}
The lower bound of the gain ratio between $\Phi'$ under pessimistic setting and the super optimal is at least $\frac{1}{3}$.
\end{lemma}
The intuition behind Lemma \ref{lem-Shat-bound-P1-pess-support-0} is as follows. On the one hand, if the customers' mean demands are close to each other, the profit loss due to pessimistic contract selection is not significant. On the other hand, if customers' mean demands differ a lot, it is more unlikely for customers to choose other contract options. As a result, the total profit loss is also limited. By Lemma \ref{lem-Shat-bound-P1-pess-support-0}, the contract solution of Definition \ref{def-near-Shat} when $\epsilon\to 0^+$ achieves an approximation ratio at least $\frac{1}{3}$.

\emph{Sketch of Proof:} To prove Lemma \ref{lem-Shat-bound-P1-pess-support-0}, it suffices to consider the case when $k$ is minimum at $k=2\hat{c}$, since the super-optimal profit decreases in $k$ as shown in Lemma \ref{lem-super}. In this extreme, the super-optimal solution simplifies to $\delta_i=0, \Delta_{th,i}=1, p_i=p_0-\frac{1}{2}\hat{c}, \forall i\in \mathcal I$ by Lemma \ref{lem-super}, and the expression for the gain of the super-optimal solution compared with the baseline can be calculated as a function of $h(m_i)$ and $m_i$. Also, we can calculate the expression for the gain of solution $\Phi'$ as a function of $h(m_i)$ and $m_i$ under pessimistic setting. Putting the above expressions together, we can show that the gain ratio between solution $\Phi'$ and the super-optimal solution is no smaller than $\frac{1}{3}$.



\subsection{Step 3}
Finally, we show that by increasing $\epsilon$ in a controllable manner, there exist feasible solutions $\Phi''$ near $\Phi'$ that only increases the supplier's profit in Lemma \ref{lem-better-solution}.
\begin{lemma}\label{lem-better-solution}
There exists an $\epsilon_0>0$ such that for all $0<\epsilon<\epsilon_0$, the contract $\Phi''$ (which is near $\Phi'$) is feasible in satisfying (\ref{eqn-mechanism-design}) and (\ref{eqn-contract-constraint}), and the supplier's profit under these solutions is no smaller than that at $\epsilon\to 0^+$.
\end{lemma}

\emph{Sketch of Proof:} To prove Lemma \ref{lem-better-solution}, we first show that there exist feasible solutions $\Phi''$ (as defined in Definition \ref{def-near-Shat}), that satisfy the IC condition. Then, to prove that the supplier's profit under these feasible solutions is no smaller than that at $\epsilon\to 0^+$, we only need to prove that the derivative of the supplier's profit with respect to $\epsilon$ is positive at $\epsilon\to 0^+$. The intuition behind this property is that lower contract price attracts more customers to subscribe to contract options. As a result, the supplier can estimate customers' total capacity more accurately and reduce its total capacity cost.

Combining the above three steps, we conclude the result in Theorem \ref{thm-bound-P1-pess}.
\begin{theorem}\label{thm-bound-P1-pess}
There exists an $\epsilon_0>0$ such that for all $0<\epsilon<\epsilon_0$, the contract $\Phi''$ is feasible in satisfying (\ref{eqn-mechanism-design}) and (\ref{eqn-contract-constraint}) and attains an gain ratio at least $\frac{1}{3}$ under the pessimistic setting.
\end{theorem}
\section{Numerical Results}\label{sec-numerical-result}
\begin{figure}[]
\vspace{-0.1in}
\centering
\includegraphics[height=6cm,width=8cm]{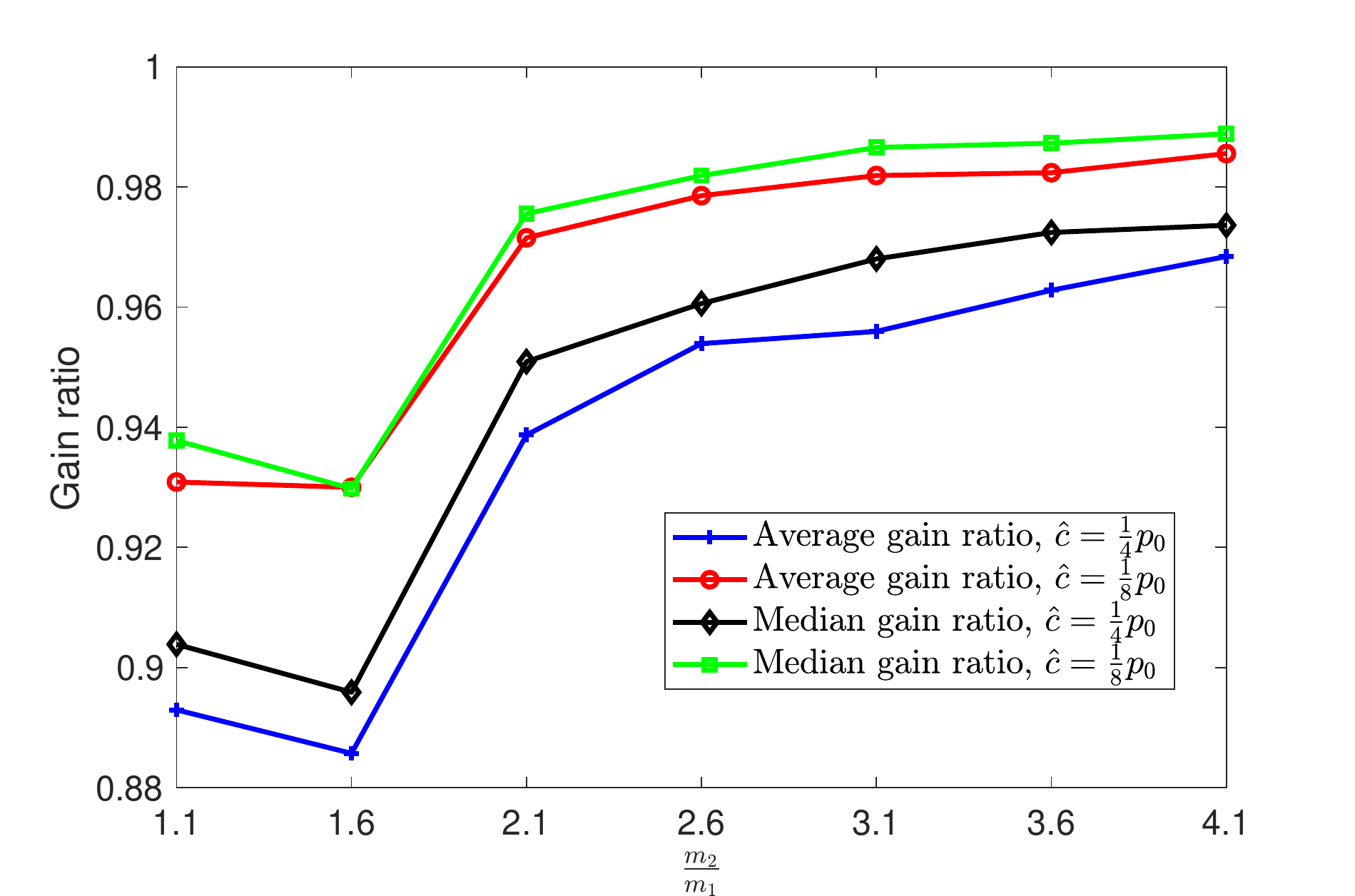}
\caption{Average gain ratio of the flexible contract as a function of $\frac{m_2}{m_1}$ and the unit capacity cost $\hat{c}$.}
\label{rate-gain-m2-hatc-average-median}
\end{figure}

In this section, we first evaluate the empirical performance of our proposed contract design $\Phi''$, and then compare the flexible contract to a typical pricing scheme: the peak-based pricing scheme.
\subsection{Empirical Gain Ratio of the Flexible Contract $\Phi''$}
As shown in Section \ref{sec-challenge-II}, the lower bound of the gain ratio of our proposed flexible contract $\Phi''$ is $\frac{1}{3}$, which represents its worst-case performance. However, in most settings, the gain ratio is higher than this lower bound. It is thus interesting to examine the average-case and median-case gain ratio of the flexible contract $\Phi''$, since they reflect the gain ratio in a common setting.

Consider an example where there are only two types of customers. Figure \ref{rate-gain-m2-hatc-average-median} shows how the average-case and median-case gain ratio changes as $\frac{m_2}{m_1}$ and the unit capacity cost $\hat{c}$ vary. To vary the ratio $\frac{m_2}{m_1}$, we fix $m_1=1$MWh and vary $m_2$. Given the values of $\frac{m_2}{m_1}$ and $\hat{c}$, we generate other parameters (i.e., $k, h(m_1), h(m_2)$) randomly and compute the average and median gain ratio. It is shown that the average and median gain ratio is much larger than the lower bound $\frac{1}{3}$. This is because the lower bound only occurs at some extreme case (e.g., $k=2\hat{c}$), which seldom occurs under the randomly generated parameters. In addition, for a given $\frac{m_2}{m_1}$, the average and median gain ratio decrease as $\hat{c}$ increases. The reason is that as $\hat{c}$ increases, the saving due to reduced capacity needs is greater. Thus, the flexible contract has a larger gain compared to the baseline pricing scheme.

\subsection{Comparison to the Peak-based Pricing Scheme}\label{sec-simulation-peak-based-pricing}
\begin{figure}[]
\vspace{-0.1in}
\centering
\includegraphics[height=6cm,width=8cm]{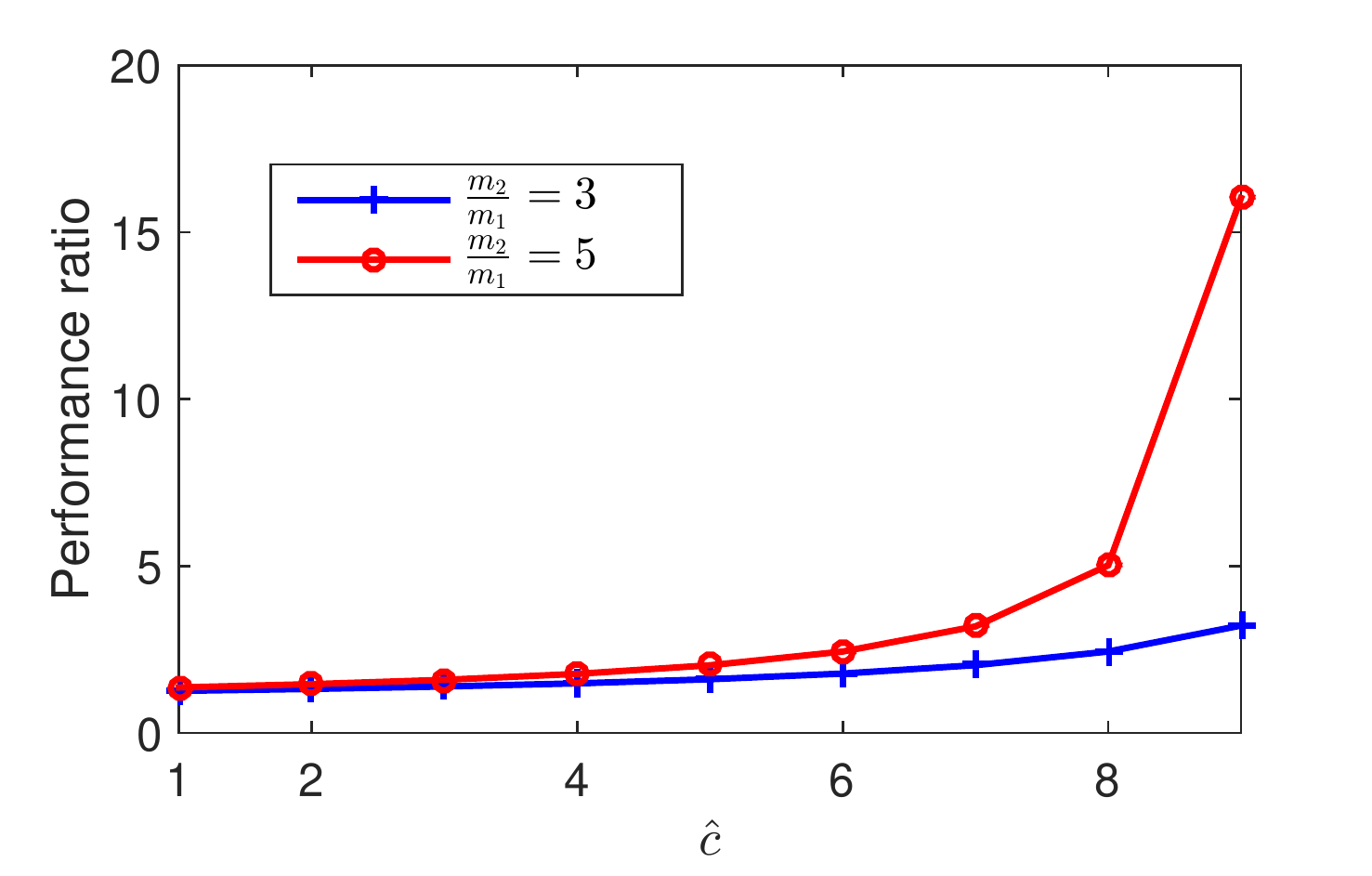}
\caption{The ratio between the supplier's profit under the flexible contract and that under the peak-based pricing as a function of unit capacity cost $\hat{c}$, and the ratio (i.e., $\frac{m_2}{m_1}$) of two mean demands at each time slot.}
\label{ratio-profit-slot4-medium-p0}
\end{figure}

Next, we compare the flexible contract with the peak-based pricing scheme in the literature. We consider one month since the supplier often pays capacity cost for an entire month each time. For simplicity, we consider four time-slots $t_1, t_2, t_3, t_4$, each of which can be viewed as representing a typical week of a particular type (e.g., busy period vs. idle period). Each time-slot thus contains $L=24\times 7=168$ hours. There are two types of customer at each time-slot. We assume that each customer's mean demand at time-slot $t_i$ can take values $m_1(t_i), m_2(t_i)$ with probability $h(t_i, m_1), h(t_i, m_2)$, respectively. In the peak-based pricing scheme, the supplier sets a base energy price $p^E$ and a peak demand price $p^D$ ($p^D>p^E$). Note that $p^E$ is the amount of price per energy consumption (e.g., in MWh) while $p^D$ is the amount of price per peak power (e.g., in MW). If a customer's real energy consumption at time-slot $t_i$ is $x_i$, then the total payment of the customer to the supplier is $p^E\sum_{i=1}^{4}x_i+p^D\frac{\max_i(x_i)}{L}$.

We first explain why we expect our flexible contract to outperform the peak-based pricing scheme. At a high level, although a high demand charge $p^D$ will force the customers to reduce its peak, the extent to which the peak can be reduced depends on the value of $k, p^E, p^D$. For example, if $k-p^E<\frac{p^D}{L}<2k-2p^E$, then each customer will only reduce its peak to the second highest demand. Further, the peak-based pricing does not provide a way for the supplier to learn the demand of the customers. In contrast, in our flexible contract, the supplier sets a baseline price $p_0\leq k$ for all four time-slots $t_1, t_2, t_3, t_4$, and additionally sets two contract options $(p_1(t_i),\delta_1(t_i), \bar{p}_1(t_i))$, $(p_2(t_i),\delta_2(t_i), \bar{p}_2(t_i))$ at each slot $t_i$ according to $\Phi''$ in Theorem \ref{thm-bound-P1-pess}. Thus, each customer can choose whether to reduce demand in \emph{every} time-slot, and the supplier will be able to learn the capacity needs by the contract selected by the customers. As a result, we expect cost-savings to both the supplier and the customers, provided that the flexible contracts are properly designed.


Consider an example where $m_1(t_1)=1$MWh, $m_1(t_2)=2$MWh, $m_1(t_3)=3$MWh, $m_1(t_4)=4$MWh, $h(t_1, m_1)=0.5$, $h(t_1, m_2)=0.5$, $h(t_2, m_1)=0.6$, $h(t_2, m_2)=0.4$, $h(t_3, m_1)=0.55$, $h(t_3, m_2)=0.45$, $h(t_4, m_1)=0.5$, $h(t_4, m_2)=0.5$, $k=75\$/$MWh, $N=10$. We assume that the base energy price $p^E$ and the peak demand price $p^D$ are $p^E=49\$/$MWh, $p^D=5258\$/$MW \footnote{Note that it is common for the demand charge of one hour of peak demand to be comparable with the energy charge of over $100$ hours at the same power level \cite{Peak-price}.}. We assume that the unit energy cost is $c_0=11\$/$MWh as in \cite{Energy-cost}. In the flexible contracts, we set a baseline price as $p_0=1.4p^E$ to balance the benefits of the supplier and each customer (Otherwise, if $p_0$ is too small, the supplier's profit is too low, and if $p_0$ is too large, each customer's cost is too high.) , and $\epsilon=0.1p_0$ in the flexible contract $\Phi''$. Figure \ref{ratio-profit-slot4-medium-p0} shows how the ratio between the supplier's profit under the flexible contract and that under the peak-based pricing changes with $\hat{c}$ and $\frac{m_2}{m_1}$. We can observe that the ratio is in general larger than $1$, especially when $\hat{c}$ and $\frac{m_2}{m_1}$ are both very large. The results thus suggest that the flexible contract enables the supplier to learn valuable demand information of customers and significantly save its capacity cost by slightly sacrificing its revenue. In addition, we see that the ratio increases with $\hat{c}$ and $\frac{m_2}{m_1}$. The reason is that, as $\hat{c}$ or $\frac{m_2}{m_1}$ increases, the capacity cost and the need to identify customer's type play a more important role in the supplier's profit. Hence, the relative advantage of the flexible contract over peak-based pricing becomes more significant. We find that the customer's mean cost reduces under the flexible contracts, achieving a win-win situation.

\section{Extension to More General Assumptions}\label{sec-extension}
In this section, we will relax the previous assumptions of our model to more general cases and show that our results are still useful under more general assumptions.
\subsection{Truncated Normal Distribution for Each Customer's Variation}

Now we show that our analysis and results also apply to other distribution for each customer's demand variation. Consider the case when each customer's variation $\Delta\in [0, 1]$ follows a normal distribution with mean $\mu$ and variance $\sigma^2$, respectively, truncated to $[0, 1]$. We choose truncated normal distribution since it can capture a wide range of distributions under different mean and variance values. Similar to the analysis in Section \ref{sec-challenge-I}, type-$m_i$ customers with low variation (i.e., $\Delta\leq \Delta_{th,i}$) will subscribe to contract option $i$, and type-$m_i$ customers of high variation will subscribe to the baseline scheme, where $\Delta_{th,i}$ is the same as (7) and (8). The supplier's expected profit $P^H_i$ from type-$m_i$ customers if option $i$ has high penalty (i.e., $\bar{p}_i>k$) is
\begin{equation}\label{eqn-profit-high-penalty-normal-distribution}
P^H_i=Nh(m_i)\bigg(\big(m_ip_iF+m_ip_0(1-F)\big)-c_0m_i-\hat{c}\big(m_i(1+\delta_i)F+2m_n(1-F)\big)\bigg),
\end{equation}
where $F=\frac{\text{erf}(\frac{\Delta_{th,i}-\mu}{\sqrt{2}\sigma})-\text{erf}(\frac{-\mu}{\sqrt{2}\sigma})}{\text{erf}(\frac{1-\mu}{\sqrt{2}\sigma})-\text{erf}(\frac{-\mu}{\sqrt{2}\sigma})}$ is the cumulative distribution function of the truncated normal distribution, representing the probability of a customer to have $\Delta\leq \Delta_{th,i}$. Note that the function $\text{erf}(x)=\frac{1}{\sqrt{2\pi}}\int_{0}^{x}e^{-t^2}dt$ is the error function of the standard normal distribution. We see that $P^H_i$ above in (\ref{eqn-profit-high-penalty-normal-distribution}) has a similar structure with that in (9). The difference is that the term $\Delta_{th,i}$ in (9) is replaced by $F$ in (\ref{eqn-profit-high-penalty-normal-distribution}).

Since it is challenging to find the optimal contract, similar to the analysis in Section \ref{sec-challenge-I}, we will find an approximate contract in the set $\hat{R}$ defined in (12). To find the approximate contract with good performance, we will maximize $P^H_i$ in (\ref{eqn-profit-high-penalty-normal-distribution}) subject to $p_i=p_0, 0\leq \Delta_{th,i}=\delta_i\leq 1$. To do this, we replace $p_i$ by $p_0$, and $\Delta_{th,i}$ by $\delta_i$ in (\ref{eqn-profit-high-penalty-normal-distribution}). Through this simplification, $P^H_i$ becomes a function with only one variable $\delta_i$ as follows.
\begin{eqnarray*}
P^H_i&=&Nh(m_i)\bigg(m_ip_0-c_0m_i-\hat{c}\big(m_i(1+\delta_i)F+2m_n(1-F)\big)\bigg)\\
     &=&Nh(m_i)\bigg(m_ip_0-c_0m_i-\hat{c}\big(2m_n+m_i(1+\delta_i-\frac{2m_n}{m_i})\frac{\text{erf}(\frac{\delta_i-\mu}{\sqrt{2}\sigma})-\text{erf}(\frac{-\mu}{\sqrt{2}\sigma})}{\text{erf}(\frac{1-\mu}{\sqrt{2}\sigma})-\text{erf}(\frac{-\mu}{\sqrt{2}\sigma})}\big)\bigg).
\end{eqnarray*}
To maximize $P^H_i$ over $\delta_i$, it suffices to minimize $(1+\delta-\frac{2m_n}{m_i})\frac{\text{erf}(\frac{\delta_i-\mu}{\sqrt{2}\sigma})-\text{erf}(\frac{-\mu}{\sqrt{2}\sigma})}{\text{erf}(\frac{1-\mu}{\sqrt{2}\sigma})-\text{erf}(\frac{-\mu}{\sqrt{2}\sigma})}$ over $\delta_i$. Thus, the final approximate contract is as follows. Recall that Problem $\mathbb{P}_2$ is defined in Section \ref{sec-challenge-I}.
\begin{proposition}\label{pro-opt-P2-normal-variation}
In Problem $\mathbb{P}_2$, the supplier's optimal contract design is $\Phi'=\{(p_i,\delta_i,\bar{p}_i), i\in\mathcal{I}|p_i=p_0, \delta_i=\argmin_{\delta\in [0,1]} (1+\delta-\frac{2m_n}{m_i})\frac{\text{erf}(\frac{\delta-\mu}{\sqrt{2}\sigma})-\text{erf}(\frac{-\mu}{\sqrt{2}\sigma})}{\text{erf}(\frac{1-\mu}{\sqrt{2}\sigma})-\text{erf}(\frac{-\mu}{\sqrt{2}\sigma})}, \bar{p}_i>k, i\in\mathcal{I}\}$, which is incentive compatible.
\end{proposition}
Since the function $(1+\delta-\frac{2m_n}{m_i})\frac{\text{erf}(\frac{\delta-\mu}{\sqrt{2}\sigma})-\text{erf}(\frac{-\mu}{\sqrt{2}\sigma})}{\text{erf}(\frac{1-\mu}{\sqrt{2}\sigma})-\text{erf}(\frac{-\mu}{\sqrt{2}\sigma})}$ has a complex form with respect to $\delta$, it is difficult to solve the closed-form expression for $\delta_i$. Instead, we optimize the choice of $\delta_i$ by one dimensional search.

We see that the value of $\delta_i$ in the approximate contract is not the same as that in Proposition \ref{pro-opt-P2}. We can study how the value of $\delta_i$ in the approximate contract is impacted by the diversity of customers' types, and the type of the truncated normal distribution in terms of the mean $\mu$ and the standard deviation $\sigma$. First, we fix $\sigma=0.5$, and vary the ratio $\frac{m_n}{m_i}$ and $\mu$, which reflects the diversity of customer's types. Figure \ref{appro-contract-RatioM-mu} shows how optimal $\delta_i$ in the approximate contract changes with $\frac{m_n}{m_i}$ and $\mu$. Afterwards, we fix $\mu=0.5$, and vary the ratio $\frac{m_n}{m_i}$ and $\sigma$.  Figure \ref{appro-contract-RatioM-sigma} shows how $\delta_i$ in the approximate contract changes with $\frac{m_n}{m_i}$ and $\sigma$. In both figures, we see that $\delta_i$ in the approximate contract is non-decreasing with $\frac{m_n}{m_i}$. Specifically, if $\frac{m_n}{m_i}$ is smaller than a threshold value, $\delta_i$ is smaller than 1. If $\frac{m_n}{m_i}$ is larger than the threshold value, $\delta_i$ is equal to $1$. This threshold-based insight is the same as that revealed in Proposition \ref{pro-opt-P2}. The difference is that the threshold value varies at different $\mu$ and $\sigma$ in the setting of normal distribution of $\Delta$ while the threshold value is a constant of $\frac{3}{2}$ in the setting of uniform distribution of $\Delta$.

\begin{figure*}[!htb]
\centering
\minipage{0.45\textwidth}
  \includegraphics[width=\linewidth]{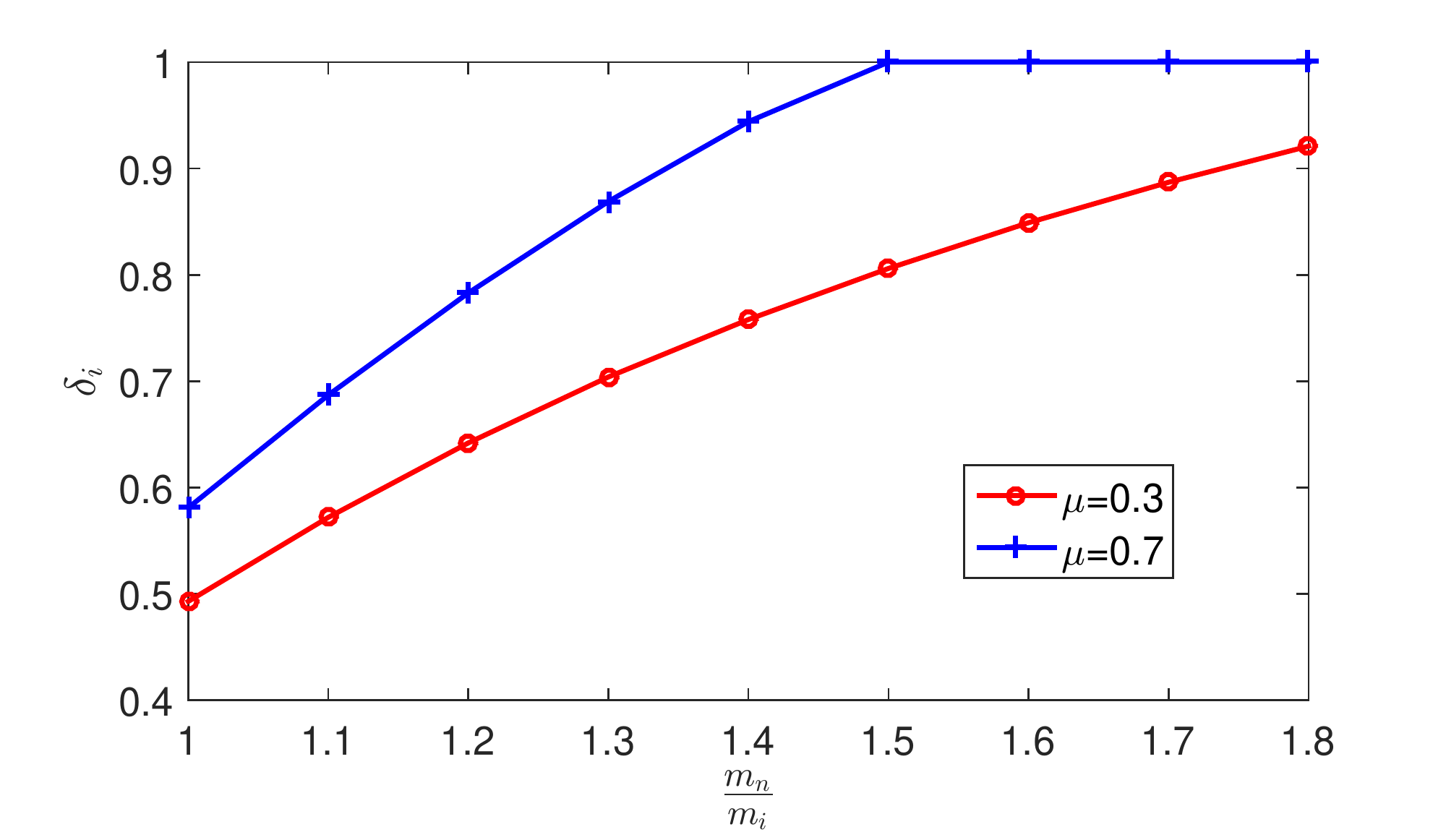}
	\vspace{-0.8cm}
  \caption{Optimal $\delta_i$ as a function of $\frac{m_n}{m_i}$ and $\mu$. The x-axis is $\frac{m_n}{m_i}$, and the y-axis is the $\delta_i$ in the approximate contract $\Phi'$}\label{appro-contract-RatioM-mu}
\endminipage\hfill
\minipage{0.45\textwidth}%
  \includegraphics[width=\linewidth]{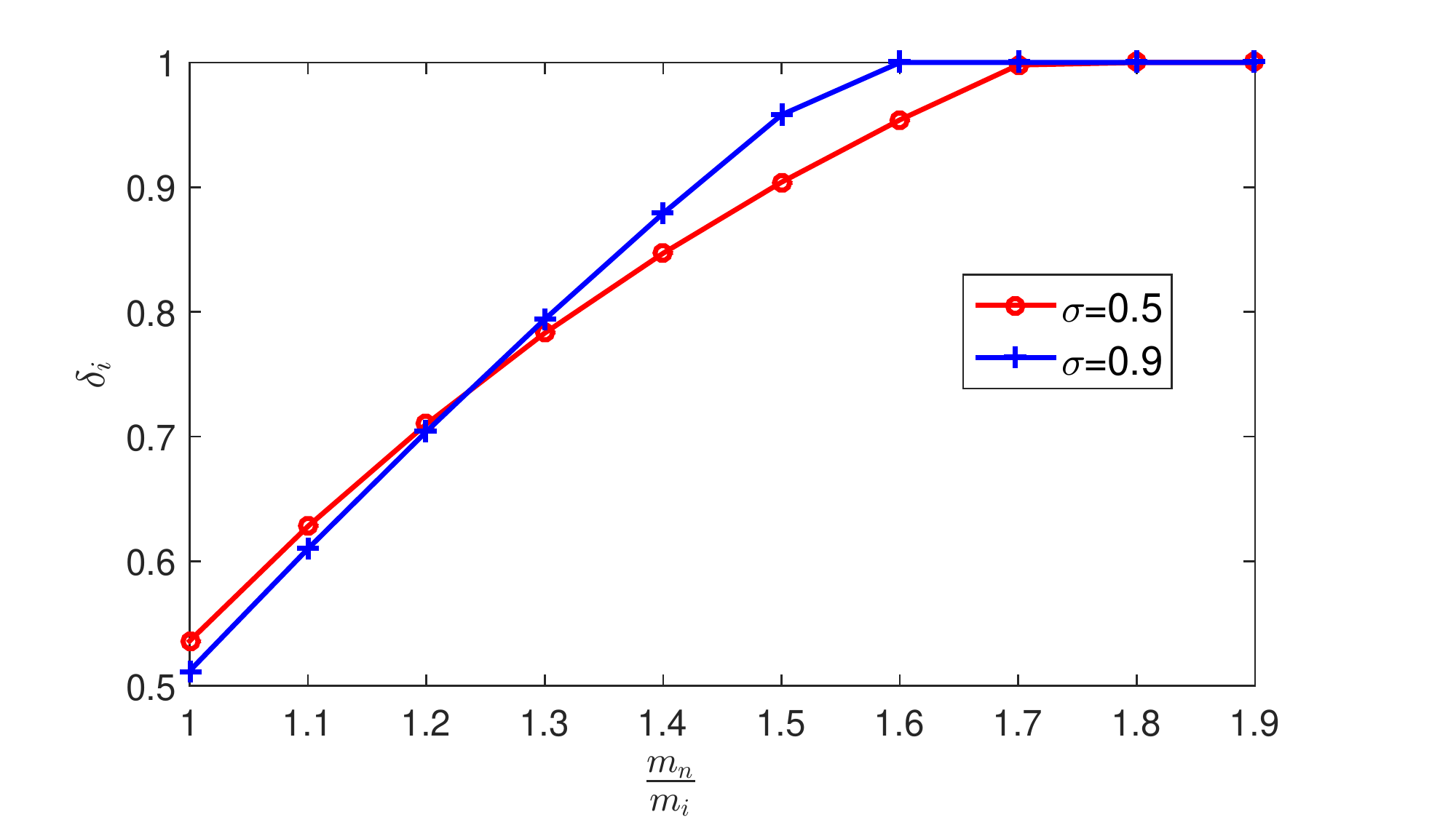}
	\vspace{-0.8cm}
  \caption{Optimal $\delta_i$ as a function of $\frac{m_n}{m_i}$ and $\sigma$}\label{appro-contract-RatioM-sigma}
\endminipage
\end{figure*}

To evaluate the performance of the approximate contract, we compare the gains of the approximate contract and the upper bound of the optimal contract. Similar to Section \ref{sec-challenge-II}, we obtain the upper bound by purposely removing constraint (3) of Problem $\mathbb{P}_1$. We refer to the resulting optimal contract design as the ``super-optimal". For ease of analysis, we first consider the optimistic setting similar to Section \ref{sec-challenge-I}, where a customer always chooses its dedicated option if it faces more than one contract yielding the same benefit. The ratio between the gain of the approximate contract under the optimistic setting and that of the super-optimal contract is difficult to analyze theoretically because the approximate contract and the super-optimal contract do not have closed-form solutions. Thus, we use extensive simulations to examine the ratio. To do this, we first consider a typical example of $n=3$ types of customers, and generate parameters $m_i, h(m_i), p_0, \hat{c}, c_0, \mu ,\sigma$ randomly. In particular, we vary $m_1\in [1,10]$, $m_2\in (m_1, 10m_1]$, $m_3\in (m_2, 10m_2]$, $p_0\in [1,100]$, $k\in (p_0, 10p_0]$, $\hat{c}=[0,\frac{1}{2}p_0], c_0\in [0, p_0], \mu\in [0,1], \sigma\in (0, 10]$, and we run $10000$ trials. We find that the average ratio is $0.9922$, and the median ratio is $0.9975$, which indicates that the approximate contract achieves a reasonable performance under the optimistic setting. We also examine the cases of $n>3$, and find the gain ratio similar (not becoming worse as $n$ increases).


Note that the above results assume an optimistic setting on the IC (Incentive Compatibility) condition. That is, when the costs of more than one contract option (including baseline pricing option) are equal, the customer will pick the dedicated option designed by the supplier, which is usually most favorable to the supplier. Although this optimistic scenario is widely assumed in the mechanism design literature, it may fail in practice. Similar to Section \ref{sec-challenge-II}, we will relax this strict IC condition and quantify the pessimistic or worst-case performance when a customer (facing more than one contract yielding the same benefit) always picks the option that is the least favorable to the supplier.

Now we will present our approach to quantify the worst-case performance similar to Section \ref{sec-challenge-II}. We start from the contract $\Phi'$ as described in Proposition \ref{pro-opt-P2-normal-variation}, but reduce all the contract prices from baseline price $p_0$ by $\epsilon>0$. This is to ensure that, for customers with small $\Delta$, choosing contracts are strictly better off than choosing the baseline scheme.
\begin{corollary}\label{corol-opt-P2-normal-variation-pessimistic}
In the pessimistic scenario, we define contract $\Phi''=\{(p_i,\delta_i,\bar{p}_i),i\in\mathcal{I}\}$ as follows. $\Phi''=\{(p_i,\delta_i,\bar{p}_i), i\in\mathcal{I}|p_i=p_0-\epsilon, \delta_i=\argmin_{\delta\in [0,1]} (1+\delta-\frac{2m_n}{m_i})\frac{\text{erf}(\frac{\delta-\mu}{\sqrt{2}\sigma})-\text{erf}(\frac{-\mu}{\sqrt{2}\sigma})}{\text{erf}(\frac{1-\mu}{\sqrt{2}\sigma})-\text{erf}(\frac{-\mu}{\sqrt{2}\sigma})}, \bar{p}_i>k, i\in\mathcal{I}\}$.
\end{corollary}
Now we evaluate the performance of the approximate contract under the pessimistic setting. To do this, we compare the gain under the approximate contract to the gain under the super-optimal contract, which is the upper bound of the optimal contract as described above. We use extensive simulations to examine the ratio. To do this, we consider a typical example of $n=2$ types of customers, and generate parameters $m_i, h(m_i), p_0, \hat{c}, c_0, \mu ,\sigma$ randomly. In particular, we vary $m_1\in [1,10]$, $m_2\in [1.1m_1, 10m_1]$, $p_0\in [1,100]$, $k\in (p_0, 10p_0]$, $\hat{c}=[0.001p_0,\frac{1}{2}p_0], c_0\in [0, p_0], \mu\in [0,1], \sigma\in (0, 10]$, $\epsilon=0.001p_0$, and we run $10000$ trials. We find that the average ratio is $0.974$, and the median ratio is $0.989$, which indicates that the approximate contract achieves a reasonable performance under the pessimistic setting as well.

\subsection{Truncated Normal Distribution for Each Customer's Demand}

Next we show that our results also apply to other distribution for each customer's demand. Consider the case when each customer's demand follows a normal distribution with mean $\mu$ and variance $\sigma^2$, respectively. This is
a truncated normal distribution version, as a type-$m_i$ customer's demand is bounded in $[m_i(1-\Delta),m_i(1+\Delta)]$ in practice. We assume $\mu=m_i$ for ease of exposition.

Similar to the analysis in Section \ref{sec-detail}, we first analyze the expected cost $\E[C_i(m_i,\Delta)]$ of a type-$m_i$ customer with maximum possible variation $\Delta$ if it chooses option $i$. If $\Delta<\delta_i$, the customer's demand is always within the contract range and its expected cost is $\E[C_i(m_i,\Delta)]=m_ip_i$. If $\Delta>\delta_i$,  its expected cost depends on the relationship between $k$ and $\bar{p}_i$. If $\bar{p}_i>k$, we have
\begin{eqnarray}\label{eqn-cost-medium-multi-option-normal-demand}\small
\vspace{-0.1cm}
\E[C_i(m_i,\Delta)]\!\!&=&\!\!(2m_ip_i-km_i\delta_i)\frac{\text{erf}(\frac{m_i\Delta}{\sqrt{2}\sigma})-\text{erf}(\frac{m_i\delta_i}{\sqrt{2}\sigma})}{2\text{erf}(\frac{m_i\Delta}{\sqrt{2}\sigma})}+\frac{k\sigma}{\sqrt{2\pi}\text{erf}(\frac{m_i\Delta}{\sqrt{2}\sigma})}(e^{-\frac{m_i\delta^2_i}{2\sigma^2}}-e^{-\frac{m_i\Delta^2}{2\sigma^2}})\nonumber\\
									& &+m_ip_i\frac{\text{erf}(\frac{m_i\delta_i}{\sqrt{2}\sigma})}{\text{erf}(\frac{m_i\Delta}{\sqrt{2}\sigma})}.
\end{eqnarray}

Similar to the analysis in Proposition \ref{pro-equilibrium-P1}, under Incentive Compatibility (IC) condition, type-$m_i$ customers with low variation (i.e., $\Delta\leq \Delta_{th,i}$) will subscribe to contract option $i$, and type-$m_i$ customers of high variation $i$ will subscribe to the baseline scheme. The only difference is that $\Delta_{th,i}$ is not the same as that in (7) and (8). Still, $\Delta_{th,i}$ ensures that $\E[C_i(m_i,\Delta)]$ in (\ref{eqn-cost-medium-multi-option-normal-demand}) is equal to the customer's expected cost $m_ip_0$ under the baseline pricing. Thus, $\Delta_{th,i}$ is the unique solution to the following equation.
\begin{eqnarray}\label{eqn-threshold-normal-demand}
\vspace{-0.1cm}
m_ip_0\!\!&=&\!\!(2m_ip_i-km_i\delta_i)\frac{\text{erf}(\frac{m_i\Delta_{th,i}}{\sqrt{2}\sigma})-\text{erf}(\frac{m_i\delta_i}{\sqrt{2}\sigma})}{2\text{erf}(\frac{m_i\Delta_{th,i}}{\sqrt{2}\sigma})}+\frac{k\sigma}{\sqrt{2\pi}\text{erf}(\frac{m_i\Delta_{th,i}}{\sqrt{2}\sigma})}(e^{-\frac{m_i\delta^2_i}{2\sigma^2}}-e^{-\frac{m_i\Delta^2_{th,i}}{2\sigma^2}})\nonumber\\
									& &+m_ip_i\frac{\text{erf}(\frac{m_i\delta_i}{\sqrt{2}\sigma})}{\text{erf}(\frac{m_i\Delta_{th,i}}{\sqrt{2}\sigma})}.
\end{eqnarray}
The supplier's expected profit $P^H_i$ from type-$m_i$ customers if option $i$ has high penalty (i.e., $\bar{p}_i>k$) is
\begin{equation}\label{eqn-profit-high-penalty-normal-demand}
P^H_i=Nh(m_i)\bigg(\big(m_ip_i\Delta_{th,i}+m_ip_0(1-\Delta_{th,i})\big)-c_0m_i-\hat{c}\big(m_i(1+\delta_i)\Delta_{th,i}+2m_n(1-\Delta_{th,i})\big)\bigg),
\end{equation}
which has the same structure as (9). Note that the difference between (\ref{eqn-profit-high-penalty-normal-demand}) and (9) is that $\Delta_{th,i}$ is now given indirectly through the equation in (\ref{eqn-threshold-normal-demand}) above.

Since it is still challenging to find the optimal contract, similar to the analysis in Section \ref{sec-challenge-I}, we will find an approximate contract in the set $\hat{R}$ defined in (12). By following the similar analysis, we prove that the final approximate contract given in Proposition \ref{pro-opt-P2-normal-demand} is of the same structure as Proposition \ref{pro-opt-P2}. The reason is that the supplier's expected profit $P^H_i$ has the same structure in (9) under the uniform distribution. Further, in the final approximate contract, $\frac{m_n}{m_i}$ has the same threshold value of $\frac{3}{2}$. 
\begin{proposition}\label{pro-opt-P2-normal-demand}
In Problem $\mathbb{P}_2$, the supplier's optimal contract design $\Phi'=\{(p_i,\delta_i,\bar{p}_i), i\in\mathcal{I}|p_i=p_0, \delta_i=\Delta_{th,i},i\in\mathcal{I}\}$ is of a simple form as follows, depending on the diversity of customers' types and is incentive compatible:
\begin{itemize}
\item If type $m_i$ is close to $m_n$ (i.e., $\frac{m_n}{m_i}\leq\frac{3}{2}$), the optimal contract option $i$ for this type is $(p_0, \frac{m_n}{m_i}-\frac{1}{2}, \bar{p}_i>k)$;
\item If type $m_i$ is not close to type-$m_n$ (i.e., $\frac{m_n}{m_i}>\frac{3}{2}$), the optimal contract option $i$ is $(p_0, 1,\bar{p}_i>k)$.
\end{itemize}
\end{proposition}

To evaluate the performance of the approximate contract, we compare the gains of the approximate contract and the super-optimal optimal contract as in Section \ref{sec-challenge-II}. For ease of analysis, we consider the optimistic setting where a customer always chooses its dedicated option if it faces more than one contract yielding the same benefit similar to Section \ref{sec-challenge-I}. The ratio between the gain of the approximate contract under the optimistic setting and that of the super-optimal contract is difficult to analyze theoretically. Thus, we use extensive simulations to examine the ratio. To do this, we consider a typical example of $n=2$ types of customers, and generate parameters $m_i, h(m_i), p_0, \hat{c}, c_0, \sigma$ randomly. In particular, we vary $m_1\in [1,10]$, $m_2\in (m_1, 10m_1]$, $p_0\in [1,10]$, $k\in (p_0, 10p_0]$, $\hat{c}=[0,\frac{1}{2}p_0], c_0\in [0, p_0], \sigma\in (0, 10]$, and we run $1000$ trials. We find that the average ratio is $0.81$, and the median ratio is $0.87$, which indicates that the approximate contract achieves a reasonable performance under the optimistic setting. We also examine the cases of $n>2$, and find the gain ratio to be similar (not becoming worse as $n$ increases).


\subsection{Continuous Distribution for Each Customer's Mean Usage}

\begin{figure}[]
\vspace{-0.1in}
\centering
\includegraphics[height=3cm,width=6cm]{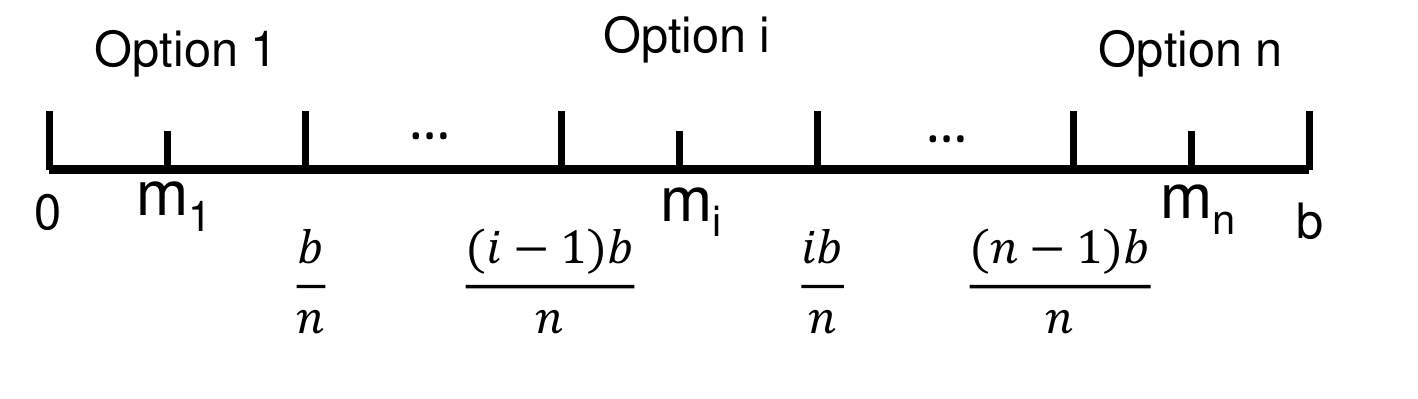}
\caption{Illustration of the relationship between the contract options and the subranges of mean usage. $m_i$ is the midpoint of the subrange $i$  with $[\frac{(i-1)b}{n},\frac{ib}{n}]$, and is set as the midpoint of the contract range of option $i$.}
\label{multiple-option-uniform-mean-usage}
\end{figure}
Our flexible contract design can be extended to a continuous distribution for the mean usage and still achieve a reasonable performance. To be concrete, we assume that each customer's mean usage $m$ is continuous random variable by following a uniform distribution in the range $[0, b]$. We believe that the insights revealed under the assumption can be generalized to other continuous distribution. We can equally divide this range into $n$ buckets \big(i.e., $[0,\frac{b}{n}],[\frac{b}{n},\frac{2b}{n}],\ldots, [\frac{(n-1)b}{n}, b]$\big), and approximate the setting as having $n$ classes of mean usage, where $n$ can be tuned to trade off efficiency and complexity of the flexible contract mechanism. We can accordingly design $n$ contract options, each of which is dedicated for one class of mean usage by the proposed approach in our paper. Figure \ref{multiple-option-uniform-mean-usage} illustrates the relationship between the $n$ contract options and the $n$ subranges of mean usage. To evaluate the performance of the approximate contract, we consider a case of ``perfect information" as an upper bound benchmark. In the case, the supplier charges the customers with the baseline price, and is assumed to know each customer's full information of its mean usage and variation. We compare the gain of the approximate contract to the gain under the case of perfect information. Figure \ref{ratio-gain-number-options} shows how the ratio between the gain of the approximate contract and the gain under the case of perfect information changes with the number of contract options. We see that the ratio increases concavely with the number of contract options, and it quickly reaches over $70\%$ as the number of options grows to $10$. Finally, it converges to around $80\%$ as the number of options reaches $30$. The reason behind the result is that as the number of contract options increases, the supplier can learn the mean usage information of each customer at better granularity, and thus predict the future capacity more accurately. We provide the details in the following.

\begin{figure}[]
\vspace{-0.1in}
\centering
\includegraphics[height=6cm,width=8cm]{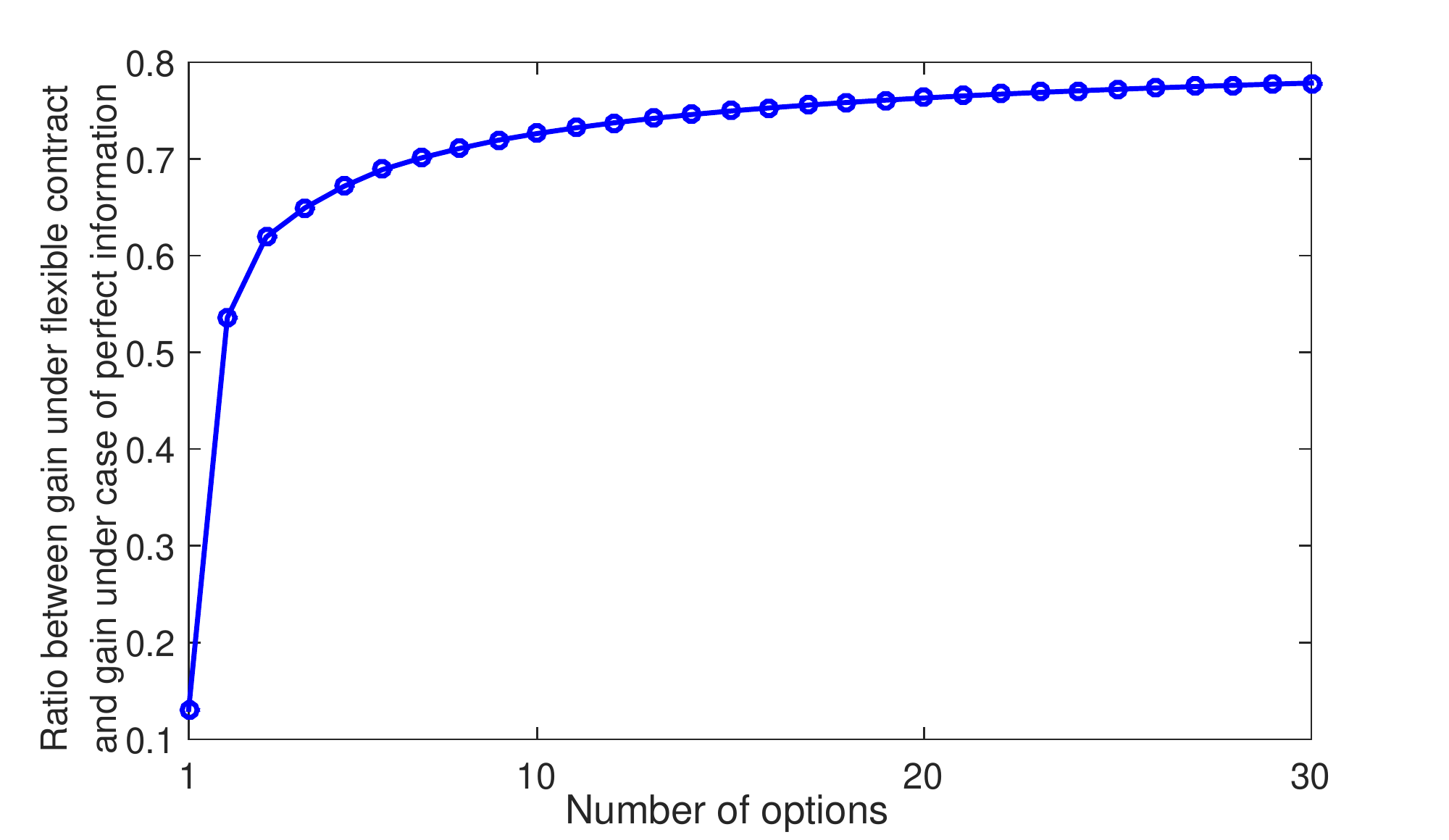}
\caption{The ratio between the gain under the flexible contract and that under the case of perfect information versus the number of contract options.}
\label{ratio-gain-number-options}
\end{figure} 


For ease of analysis, we consider an optimistic scenario similar to Section \ref{sec-challenge-I}. That is, when the costs of more than one contract option (including baseline pricing option) are equal, the customer will pick the dedicated option designed by the supplier, which is usually most favorable to the supplier. We equally divide the mean usage range $[0, b]$ into $n$ subranges: $[0,\frac{b}{n}], [\frac{b}{n},\frac{2b}{n}],\ldots, [\frac{(n-1)b}{n}, b]$ . We design $n$ contract options for the customers, where each option is dedicated for a subrange of mean usage. Specifically, a contract option $i$ is dedicated for the subrange $i$ with $[\frac{(i-1)b}{n},\frac{ib}{n}]$. Thus, we assume that the midpoint of the discounted range of each option $i$ is the midpoint $m_i$ of the subrange $i$, where 
\begin{equation}\label{eqn-mi}
m_i=\frac{(2i-1)b}{2n}.
\end{equation}
We refer to a customer as type-$m_i$ customer if its mean demand is in the subrange $i$. (Note that the definition is different from that under discrete distribution of the mean usage.) Intuitively, a type-$m_i$ customer with a small variation $\Delta$ is very likely to choose its dedicated option $i$. The question is how to design the contract parameters $p_i,\delta_i, \bar{p}_i$. To do this, we determine an approximate contract in Definition \ref{def-Phi-hat} below using the results in Proposition \ref{pro-opt-P2}. Note that in Definition \ref{def-Phi-hat}, the term $m_i$ represents the midpoint of the subrange $i$ and is defined in (\ref{eqn-mi}). This is somewhat different from Proposition \ref{pro-opt-P2} which is based on a discrete distribution of mean usage, and which uses $m_i$ to denote the particular mean usage of a type-$m_i$ customer. Under the continuous distribution, each option $i$ corresponds to an infinite number of mean usages in the subrange $i$. To apply the Proposition \ref{pro-opt-P2} to the contract design for the continuous distribution, we simply use the midpoint of the mean usage subrange $i$ to approximately represent the mean usage of all type-$m_i$ customers. 
\begin{definition}\label{def-Phi-hat}
We define contract $\hat{\Phi}=\{(p_i,\delta_i,\bar{p}_i),i\in\mathcal{I}\}$ below, which depends on the diversity of customers' types:  
\begin{itemize}
\item If $m_i$ is close to $m_n$ (i.e., $\frac{m_n}{m_i}\leq\frac{3}{2}$), contract option $i$ is $(p_0, \frac{m_n}{m_i}-\frac{1}{2}, \bar{p}_i>k)$.
\item If $m_i$ is much smaller than $m_n$ (i.e., $\frac{m_n}{m_i}>\frac{3}{2}$), contract option $i$ is $(p_0, 1,\bar{p}_i>k)$. 
\end{itemize}
\end{definition}

Now we evaluate the performance of the approximate contract for the uniform distribution of mean usage. To do this, we compare the gain under the approximate contract to the gain under the case of perfect information. To achieve this, we need to compute the supplier's (expected) profit under the baseline pricing, the approximate contract, and the case of perfect information. The derivation of these terms is similar with Section III. We first analyze the profit of a customer with mean usage $m$, and variation $\Delta$. To derive the expected profit, we average the profit over all possible $\Delta\in [0,1]$, and all possible $m\in [0,b]$ sequentially. We omit the details of the derivation and only report the results. The supplier's (expected) profit under the baseline pricing is 
\begin{equation}
P_0=\frac{1}{2}(p_0-c_0)b-2\hat{c}b.
\end{equation} 
Under the approximate contract, a type-$m_i$ customer with small variation $\Delta$ chooses the option $i$, and a type-$m_i$ customer with large variation $\Delta$ chooses the baseline pricing. The supplier's (expected) profit under the approximate contract is
\begin{eqnarray*}\small
\!\!\!P(\hat{\Phi})&=&\frac{1}{2}(p_0-c_0)b-2\hat{c}b+\frac{\hat{c}}{b}[\sum_{i=1}^{\lfloor\frac{4n+1}{6}\rfloor}2m_i(2b-2m_i)\ln\frac{2i}{2i-1}-\\
             & &\sum_{i=\lfloor\frac{4n+1}{6}\rfloor+1}^{n}\bigg((\frac{3}{2}m_i-m_n)(2b-\frac{1}{2}m_i-m_n)\ln\frac{2i-1}{2i-2}\\
						 & &+(\frac{1}{2}m_i+m_n)(\frac{1}{2}m_i+m_n-2b)\ln\frac{2i}{2i-1}\bigg)], n\geq 2.
\end{eqnarray*}
The supplier's profit under the case of perfect information is 
\begin{equation}
P^*=\frac{1}{2}(p_0-c_0)b-\frac{3}{4}\hat{c}b.
\end{equation}
Replacing $m_i=\frac{(2i-1)b}{2n}$, we can derive the ratio between the gain of the approximate contract and the gain under the case of perfect information as follows.
\begin{proposition}
The ratio of between the gain of approximate contract and the gain under the case of perfect information only depends on $n$ as follows.
\begin{itemize}
\item If $n\geq 2$, then the gain ratio is
\begin{eqnarray*}\small
             \!\!\!\frac{P(\hat{\Phi})-P_0}{P^*-P_0}\!\!\!&=&\frac{4}{5}[\sum_{i=1}^{\lfloor\frac{4n+1}{6}\rfloor}\frac{(2i-1)(2n-2i+1)}{n^2}\ln\frac{2i}{2i-1}-\\
						                                  & &\sum_{i=\lfloor\frac{4n+1}{6}\rfloor+1}^{n}\bigg(\frac{(6i-4n-1)(4n-2i-1)}{16n^2}\ln\frac{2i-1}{2i-2}+\frac{(2i+4n-3)(2i-4n-3)}{16n^2}\ln\frac{2i}{2i-1}\bigg)];
\end{eqnarray*}
\item If $n=1$, then the gain ratio is $\frac{P(\hat{\Phi})-P_0}{P^*-P_0}=\frac{4}{5}(-\frac{5}{16}\ln 2+\frac{15}{16}\ln\frac{3}{2})$.
\end{itemize}
\end{proposition}
As described earlier, Figure \ref{ratio-gain-number-options} shows how the ratio between the gains changes with the number of contract options $n$. We see that the ratio increases concavely with the number of contract options, and it quickly reaches over $70\%$ as the number of options grows to $10$. The reason behind the result is that, as the number of contract options increases, the supplier can learn the mean usage information of each customer at better granularity, and thus predict the future capacity more accurately. 
\section{Conclusion}
This paper studies how to incentivize the customers to reveal their private demand information to the supplier. We propose a novel scheme of flexible contracts to motivate different types of customers to truthfully reveal and even control their demand variation with commitment. We address two key challenges in designing the optimal contract: i) the contract optimization problem is non-convex and intractable for a large number of various customer types, and ii) the design should be robust against customers' uncertain responses. We propose provably effective solutions to address these challenges. 

For future work, we would like to investigate how the supplier should design the optimal contract in multiple time periods. 

\section*{Appendix}
We first define $\Delta_{i,j}$ as the maximum variability degree $\Delta$ of type-$m_i$ customer whose whole usage is included within the contract range of option $j$. We have
\begin{equation}
\Delta_{i,j}=
\begin{cases}
\frac{m_j(1+\delta_j)}{m_i}-1, j<i\\
1-\frac{m_j(1-\delta_j)}{m_i}, j>i
\end{cases}
\end{equation}
Note that the notation will be also used in the proof of Lemma \ref{lem-Shat-bound-P1-pess-support-0}, Claim \ref{claim-feasible}, Lemma \ref{lem-better-solution}.


\begin{figure}[]
\centering
\includegraphics[height=8cm, width=10cm]{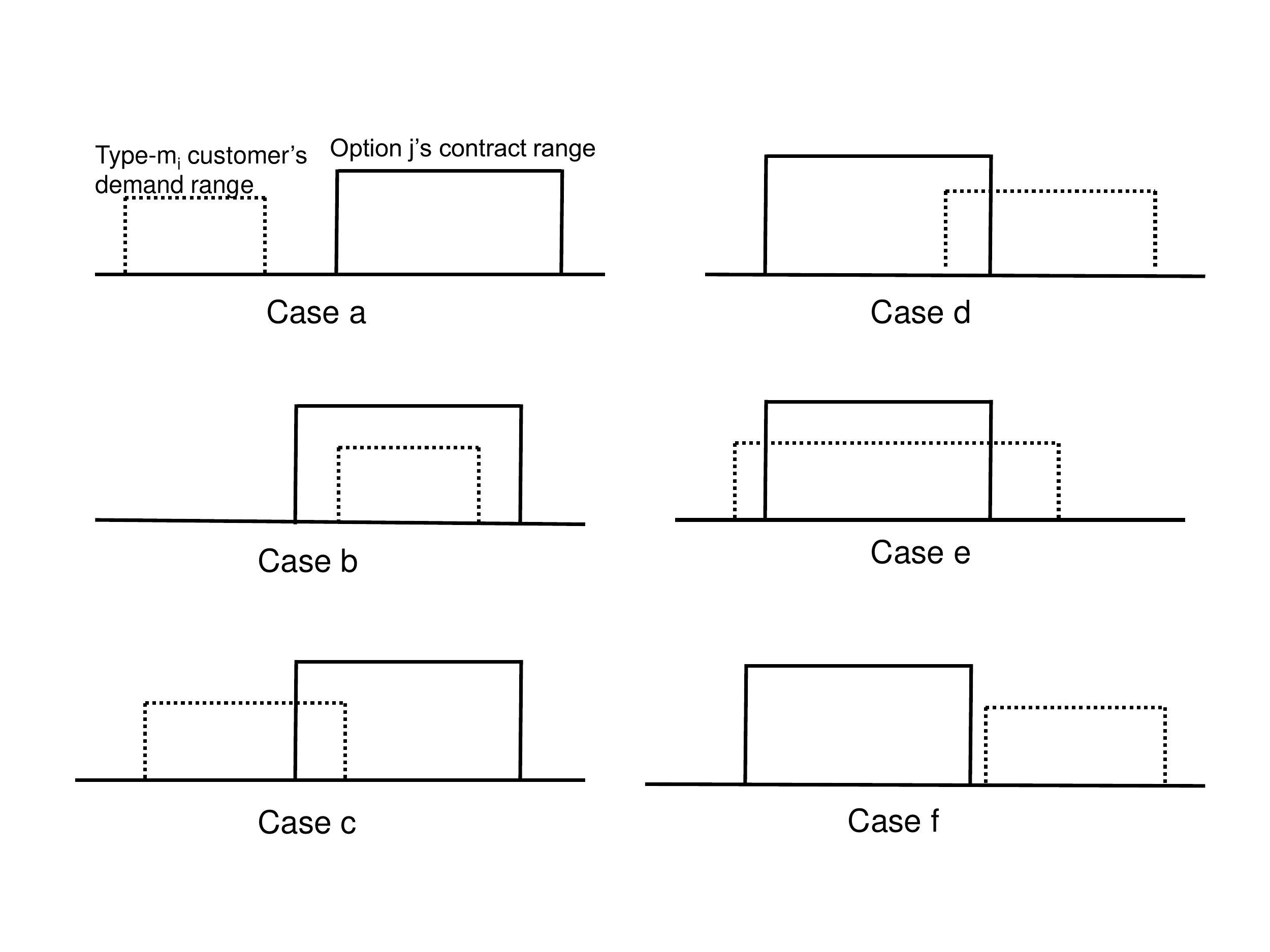}
\caption{Relationship between type-$m_i$ customer's demand range (i.e., $[m_i(1-\Delta), m_i(1+\Delta)]$) and the contract range (i.e., $[m_j(1-\delta_j), m_j(1+\delta_j)]$) of option $j$}
\label{fig-relationship-contract-range}
\end{figure}
\vspace{0.5cm}
\begin{lemma}\label{lem-cost-wrong}
$\E[C_j(m_i,\Delta)]$ depends on the relationship between type-$m_i$ customer's demand range (i.e., $[m_i(1-\Delta), m_i(1+\Delta)]$) and the contract range (i.e., $[m_j(1-\delta_j), m_j(1+\delta_j)]$) of option $j$. 
\begin{itemize}
\item Case a: If $m_i(1+\Delta)<m_j(1-\delta_j)$, then
\begin{equation*}\small
\E[C_j(m_i,\Delta)]=m_j(1-\delta_j)p_j.
\end{equation*} 

\item Case b: If $m_i(1-\Delta)>m_j(1-\delta_j)$ and $m_i(1+\Delta)\leq m_j(1+\delta_j)$, then 
\begin{equation*}\small
\E[C_j(m_i,\Delta)]=m_ip_j.
\end{equation*} 

\item Case c: If $m_i(1-\Delta)\leq m_j(1-\delta_j)$ and $m_j(1-\delta_j)\leq m_i(1+\Delta)\leq m_j(1+\delta_j)$, then 
\begin{equation*}\small
\E[C_j(m_i,\Delta)]=\frac{p_j}{4m_i}\big(m^2_i\Delta+\frac{(m_i-m_j(1-\delta_j))^2}{\Delta}+2m^2_i+2m_im_j(1-\delta_j)\big).
\end{equation*}

\item Case d: If $m_j(1-\delta_j)\leq m_i(1-\Delta)\leq m_j(1+\delta_j)$ and $m_i(1+\Delta)\geq m_j(1+\delta_j)$, then $\E[C_j(m_i,\Delta)]$ if $\bar{p}_j>k$ is
\begin{equation*}\small
 \E[C_j(m_i,\Delta)]=\frac{1}{4m_i}\big((k-p_j)m^2_i\Delta+\frac{(k-p_j)(m_j(1+\delta_j)-m_i)^2}{\Delta}+2km^2_i+2p_jm^2_i+2(p_j-k)m_im_j(1+\delta_j)\big);
\end{equation*}
Otherwise, if $\bar{p}_j\leq k$ 
\begin{equation*}\small
 \E[C_j(m_i,\Delta)]=\frac{1}{4m_i}\big((\bar{p}_j-p_j)m^2_i\Delta+\frac{(\bar{p}_j-p_j)(m_j(1+\delta_j)-m_i)^2}{\Delta}+2\bar{p}_jm^2_i+2p_jm^2_i+2(p_j-\bar{p}_j)m_im_j(1+\delta_j)\big).
\end{equation*} 

\item Case e: If $m_i(1-\Delta)\leq m_j(1-\delta_j)$ and $m_i(1+\Delta)\geq m_j(1+\delta_j)$, then if $\bar{p}_j>k$, 
\begin{eqnarray*}\small
 \E[C_j(m_i,\Delta)]&=&\frac{1}{4m_i}\bigg(km^2_i\Delta+\frac{(-4\delta_jp_j+k(1+\delta_j)^2)m^2_j-2(k(1+\delta_j)-2\delta_jp_j)m_im_j+km^2_i}{\Delta}\\
                    & &+2km^2_i+2(-k(1+\delta_j)+2p_j)m_im_j\bigg);
\end{eqnarray*} 
Otherwise, if $\bar{p}_j\leq k$, 
\begin{eqnarray*}\small
 \E[C_j(m_i,\Delta)]&=&\frac{1}{4m_i}\bigg(\bar{p}_jm^2_i\Delta+\frac{(-4\delta_jp_j+\bar{p}_j(1+\delta_j)^2)m^2_j-2(\bar{p}_j(1+\delta_j)-2\delta_jp_j)m_im_j+\bar{p}_jm^2_i}{\Delta}\\
                    & &+2\bar{p}_jm^2_i+2(-\bar{p}_j(1+\delta_j)+2p_j)m_im_j\bigg).
\end{eqnarray*} 
\item Case f: If $m_i(1-\Delta)>m_j(1+\delta_j)$, then if $\bar{p}_j>k$
\begin{equation*}\small
\E[C_j(m_i,\Delta)]=(p_j-k)m_j(1+\delta_j)+km_i;
\end{equation*} 
Otherwise, if $\bar{p}_j\leq k$,
\begin{equation*}\small
\E[C_j(m_i,\Delta)]=(p_j-\bar{p}_j)m_j(1+\delta_j)+\bar{p}_jm_i.
\end{equation*} 

\end{itemize}

\end{lemma}

\begin{proof}
There are a total of six possible cases for the relationship between type-$m_i$ customer's usage range and contract range of option $j$ as shown in Figure \ref{fig-relationship-contract-range}. Thus, there are six cases for the cost $\E[C_j(m_i,\Delta)]$ as shown in the proposition. We discuss the cost case by case. In the first case (a) that $m_i(1+\Delta)<m_{j}(1-\delta_{j})$ (i.e., $\Delta<\frac{m_{j}(1-\delta_{j})}{m_i}-1$). We have 
\begin{equation*}
\E[C_j(m_i,\Delta)]=\int_{m_i(1-\Delta)}^{m_i(1+\Delta)}\frac{1}{2m_i\Delta}p_{j}m_{j}(1-\delta_{j})dx=m_{j}(1-\delta_{j})p_{j}.
\end{equation*} 

In the second case (b) that $m_i(1-\Delta)<m_{j}(1-\delta_{j})$ and $m_i(1+\Delta)\leq m_{j}(1+\delta_{j})$ (i.e., $\Delta\leq 1-\frac{m_{j}(1-\delta_{j})}{m_i}$ and $\Delta\leq \frac{m_{j}(1+\delta_{j})}{m_i}-1$). We have
\begin{equation*}
\E[C_j(m_i,\Delta)]=\int_{m_i(1-\Delta)}^{m_i(1+\Delta)}\frac{1}{2m_i\Delta}p_{j}xdx=m_ip_{j}.
\end{equation*} 

In the third case (c) that $m_i(1-\Delta)\leq m_{j}(1-\delta_{j})$ and $m_{j}(1-\delta_{j})\leq m_i(1+\Delta)\leq m_{j}(1+\delta_{j})$ (i.e., $\Delta\geq 1-\frac{m_{j}(1-\delta_{j})}{m_i}$, $\Delta\geq\frac{m_{j}(1-\delta_{j})}{m_i}-1$, and $\Delta\leq \frac{m_{j}(1+\delta_{j})}{m_i}-1$). We have 
\begin{eqnarray*}
 \E[C_j(m_i,\Delta)] &=&\int_{m_i(1-\Delta)}^{m_{j}(1-\delta_{j})}\frac{1}{2m_i\Delta}p_{j}m_{j}(1-\delta_{j})dx+\int_{m_{j}(1-\delta_{j})}^{m_i(1+\Delta)}\frac{1}{2m_i\Delta}p_{j}xdx\\
                     &=&\frac{p_{j}}{2m_i\Delta}\big(m_{j}(1-\delta_{j})(m_{j}(1-\delta_{j})-m_i(1-\Delta))+\frac{1}{2}(m^2_i(1+\Delta)^2-m^2_{j}(1-\delta_{j})^2)\big)\\
                     &=&\frac{p_{j}}{4m_i}\big(m^2_i\Delta+\frac{(m_i-m_{j}(1-\delta_{j}))^2}{\Delta}+2m^2_i+2m_im_{j}(1-\delta_{j})\big).
\end{eqnarray*} 

In the fourth case (d) that $m_{j}(1-\delta_{j})\leq m_i(1-\Delta)\leq m_{j}(1+\delta_{j})$ and $m_i(1+\Delta)\geq m_{j}(1+\delta_{j})$ (i.e., $\Delta\leq 1-\frac{m_{j}(1-\delta_{j})}{m_i}$, $\Delta\geq 1-\frac{m_{j}(1+\delta_{j})}{m_i}$ and $\Delta\geq\frac{m_{j}(1+\delta_{j})}{m_i}-1$). If $\bar{p}_j>k$, we have 
\begin{eqnarray*}
 \E[C_j(m_i,\Delta)]  &=&\int_{m_i(1-\Delta)}^{m_{j}(1+\delta_{j})}\frac{1}{2m_i\Delta}p_{j}xdx+\int_{m_{j}(1+\delta_{j})}^{m_i(1+\Delta)}\frac{1}{2m_i\Delta}\big(p_{j}m_{j}(1+\delta_{j})+k(x-m_{j}(1+\delta_{j}))\big)dx\\
                   &=&\frac{1}{2m_i\Delta}\big(\frac{1}{2}p_{j}(m^2_{j}(1+\delta_{j})^2-m^2_i(1-\Delta)^2)+(p_{j}-k)m_{j}(1+\delta_{j})(m_i(1+\Delta)-m_{j}(1+\delta_{j}))\\
									 & &+\frac{1}{2}k(m^2_i(1+\Delta)^2-m^2_{j}(1+\delta_i)^2)\big)\\
				           &=&\frac{1}{4m_i}\big((k-p_{j})m^2_i\Delta+\frac{(k-p_{j})(m_{j}(1+\delta_{j})-m_i)^2}{\Delta}+2km^2_i+2p_{j}m^2_i+2(p_{j}-k)m_im_{j}(1+\delta_{j})\big).
\end{eqnarray*} 
Similarly, if $\bar{p}_j\leq k$, We can compute $\E[C_j(m_i,\Delta)]=\frac{1}{4m_i}\big((\bar{p}_j-p_{j})m^2_i\Delta+\frac{(\bar{p}_j-p_{j})(m_{j}(1+\delta_{j})-m_i)^2}{\Delta}+2\bar{p}_jm^2_i+2p_{j}m^2_i+2(p_{j}-\bar{p}_j)m_im_{j}(1+\delta_{j})\big)$ by replacing $k$ with $\bar{p}_j$.

In the fifth case (e) that $m_i(1-\Delta)\leq m_{j}(1-\delta_{j})$ and $m_i(1+\Delta)\geq m_{j}(1+\delta_{j})$ (i.e., $\Delta\geq 1-\frac{m_{j}(1-\delta_{j})}{m_i}$, and $\Delta\geq\frac{m_{j}(1+\delta_{j})}{m_i}-1$). If $\bar{p}_j> k$, we have 
\begin{eqnarray*}\small
 \E[C_j(m_i,\Delta)]&=&\int_{m_i(1-\Delta)}^{m_{j}(1-\delta_{j})}\frac{1}{2m_i\Delta}p_{j}m_{j}(1-\delta_{j})dx+\int_{m_{j}(1-\delta_{j})}^{m_{j}(1+\delta_{j})}\frac{1}{2m_i\Delta}p_{j}xdx\\
                   & &+\int_{m_{j}(1+\delta_{j})}^{m_i(1+\Delta)}\frac{1}{2m_i\Delta}\big(p_{j}m_{j}(1+\delta_{j})+k(x-m_{j}(1+\delta_{j}))\big)dx\\
                   &=&\frac{1}{2m_i\Delta}\big(\frac{1}{2}(-4\delta_{j}p_{j}+k(1+\delta_{j})^2)m^2_{j}-(k(1+\Delta)(1+\delta_{j})-2(\delta_{j}+\Delta)p_{j})m_im_{j}\\
									 & &+\frac{1}{2}km^2_i(1+\Delta)^2\big)\\
				           &=&\frac{1}{4m_i}\big(km^2_i\Delta+\frac{(-4\delta_{j}p_{j}+k(1+\delta_{j})^2)m^2_{j}-2(k(1+\delta_{j})-2\delta_{j}p_{j})m_im_{j}+km^2_i}{\Delta}\\
				           & &+2km^2_i+2(-k(1+\delta_{j})+2p_{j})m_im_{j}\big).
\end{eqnarray*} 
Similarly, if $\bar{p}_j\leq k$, We can compute $\E[C_j(m_i,\Delta)]=\frac{1}{4m_i}\bigg(\bar{p}_jm^2_i\Delta+\frac{(-4\delta_jp_j+\bar{p}_j(1+\delta_j)^2)m^2_j-2(\bar{p}_j(1+\delta_j)-2\delta_jp_j)m_im_j+\bar{p}_jm^2_i}{\Delta}+2\bar{p}_jm^2_i+2(-\bar{p}_j(1+\delta_j)+2p_j)m_im_j\bigg)$ by replacing $k$ with $\bar{p}_j$.

In the sixth case (f) that $m_i(1-\Delta)>m_{j}(1+\delta_{j})$ (i.e., $\Delta<1-\frac{m_{j}(1+\delta_{j})}{m_i}$). If $\bar{p}_j> k$, we have
\begin{eqnarray*}
 \E[C_j(m_i,\Delta)]&=&\int_{m_i(1-\Delta)}^{m_i(1+\Delta)}\frac{1}{2m_i\Delta}\big(p_{j}m_{j}(1+\delta_{j})+k(x-m_{j}(1+\delta_{j}))\big)dx\\
                    &=&(p_{j}-k)m_{j}(1+\delta_{j})+km_i.
\end{eqnarray*}
Similarly, if $\bar{p}_j\leq k$, we can compute $\E[C_j(m_i,\Delta)]=(p_j-\bar{p}_j)m_j(1+\delta_j)+\bar{p}_jm_i$ by replacing $k$ with $\bar{p}_j$.

Hence, the lemma holds.
\end{proof}

\vspace{0.5cm}
\hspace{-0.4cm}\textbf{Proof of Proposition \ref{pro-equilibrium-P1}}
\begin{proof}
We first consider the equilibrium behavior of a customer with variability degree $\Delta\in [0, \delta_i]$. Its expected cost is $m_ip_i$ if it chooses the contract option while its expected cost is $m_ip_0$ if it chooses the baseline pricing scheme. Since $p_i<p_0$, the customer's expected cost under the contract scheme is always smaller than that under the baseline pricing scheme. Thus, the customer will subscribe to the contract. 

We then consider a customer's equilibrium behavior with $\Delta>\delta_i$. In the high penalty regime ($\bar{p}_i>k$), its expected cost is $m_ip_i+\frac{m_ik}{4\Delta}(\Delta-\delta_i)^2$ if it chooses the contract option and its cost is $m_ip_0$ if it chooses the baseline pricing scheme. Thus, to compute $\Delta_{th,i}$ which separates customers in choosing between the two pricing schemes, we let 
\begin{equation}
m_ip_i+\frac{m_ik}{4\Delta_{th,i}}(\Delta_{th,i}-\delta_i)^2=m_ip_0.
\end{equation}
Since $p_i\leq p_0$, the equation has two real value solutions: 
\begin{eqnarray}
\frac{k\delta_i+2(p_0-p_i)-\sqrt{(k\delta+2(p_0-p_i))^2-k^2\delta^2_i}}{k}\leq\delta_i,\nonumber\\
\frac{k\delta_i+2(p_0-p_i)+\sqrt{(k\delta+2(p_0-p_i))^2-k^2\delta^2_i}}{k}\geq\delta_i. 
\end{eqnarray}
Note that customers with $\Delta\in [0, \delta_i]$ will subscribe to the contract, which indicates $\Delta_{th,i}>\delta_i$. Thus, the first solution is not a feasible threshold. Furthermore, since each customer's variation must not exceed $1$, we have
\begin{equation}
\Delta_{th,i}=\min(1,\frac{k\delta_i+2(p_0-p_i)+\sqrt{(k\delta+2(p_0-p_i))^2-k^2\delta^2_i}}{k}). 
\end{equation}
We then prove customers with $\Delta\in [\delta_i, \Delta_{th,i}]$ will subscribe to the contract, and customers with $\Delta\in (\Delta_{th,i}, 1]$ choose the baseline pricing scheme. To do this, let $q(\Delta)=m_ip_i+\frac{m_ik}{4\Delta}(\Delta-\delta_i)^2$, and it suffices to prove $q(\Delta)\leq m_ip_0$ if $\Delta\in [\delta_i, \Delta_{th,i}]$ and $q(\Delta)> m_ip_0$ if $\Delta\in (\Delta_{th,i},1]$. To do this, there are two cases depending on $\Delta_{th,i}$. In the first case that $\frac{k\delta_i+2(p_0-p_i)+\sqrt{(k\delta_i+2(p_0-p_i))^2-k^2\delta^2_i}}{k}\leq 1$, then $\Delta_{th,i}=\frac{k\delta_i+2(p_0-p_i)+\sqrt{(k\delta_i+2(p_0-p_i))^2-k^2\delta^2_i}}{k}\in [\delta_i, 1]$. Thus, we have $q(\Delta=\Delta_{th,i})=m_ip_i+\frac{m_ik}{4\Delta_{th,i}}(\Delta_{th,i}-\delta_i)^2=m_ip_0$, $\Delta_{th,i}\in [\delta_i, 1]$. Since
\begin{equation}
\frac{\partial q(\Delta)}{\partial \Delta}=\frac{m_ik}{4}(1-\frac{\delta^2_i}{\Delta^2})\geq 0, \Delta\geq\delta_i,
\end{equation}
we have $q(\Delta)$ is increasing in $\Delta, \Delta\geq\delta_i$, and thus $q(\Delta)\leq m_ip_0$ if $\Delta\in [\delta_i, \Delta_{th,i}]$ and $q(\Delta)> m_ip_0$ if $\Delta\in (\Delta_{th,i},1]$. In the second case that $\frac{k\delta_i+2(p_0-p_i)+\sqrt{(k\delta_i+2(p_0-p_i))^2-k^2\delta^2_i}}{k}> 1$, then $\Delta_{th,i}=1$, and $q(\Delta=\Delta_{th,i})=m_ip_i+\frac{m_ik}{4\Delta_{th,i}}(\Delta_{th,i}-\delta_i)^2\leq m_ip_0$. Since $q(\Delta)$ is increasing in $\Delta, \Delta\geq\delta_i$, we have $q(\Delta)\leq m_ip_0$ if $\Delta\in [\delta_i, \Delta_{th,i}]$. Note that $\Delta$ should not exceed $1$, and thus there is no need to prove $q(\Delta)> m_ip_0$ if $\Delta\in (\Delta_{th,i},1]$. 

\end{proof}

\vspace{0.5cm}
\hspace{-0.4cm}\textbf{Proof of Proposition \ref{pro-opt-P2}}
\begin{proof}

By the definition of Problem $\mathbb{P}_2$, its domain is $\hat{R}=\{(p_i, \delta_i, \bar{p}_i), i\in\mathcal{I}|p_i=p_0, \delta_i\in [0, 1], \bar{p}_i>k, i\in\mathcal{I}\}$ and its objective is the same as Problem $\mathbb{P}_1$. For any solution in $\hat{R}$, we have $\forall i\in \mathcal{I}, p_i=p_0, \Delta_{th,i}=\delta_i$. Besides, the objective of Problem $\mathbb{P}_1$ is shown in Equation (\ref{eqn-profit-P1}). Thus, for any solution of Problem $\mathbb{P}_2$, the supplier's profit is simplified to
\begin{equation}
P(\Phi)=\sum_{i=1}^{n}Nh(m_i)(m_ip_0-\hat{c}(m_i(1+\delta_i)\delta_i+2m_n(1-\delta_i))-c_0m_i).
\end{equation}

Let $P_i(\Phi_i)=Nh(m_i)(m_ip_0-\hat{c}(m_i(1+\delta_i)\delta_i+2m_n(1-\delta_i))-c_0m_i)$ where $\Phi_i=(p_i, \delta_i, \bar{p}_i)$. To solve Problem $\mathbb{P}_2$, it is equivalent to solve $n$ subproblems. In particular, each subproblem $i (i=1,\ldots, n)$ is
\begin{equation}
\max_{\Phi_i} P_i(\Phi_i).
\end{equation}
subject to 
\begin{equation*}
0\leq \Delta_{th,i}=\delta_i\leq 1.
\end{equation*}
We will solve each subproblem $i$, respectively. We have 
\begin{eqnarray*}
P_i(\Phi_i)&=&Nh\big(m_i)(m_ip_0-\hat{c}(m_i\delta^2_i+(m_i-2m_n)\delta_i+2m_n)-c_0m_i\big).
\end{eqnarray*}
which is one variable function in $\delta_i\in [0,1]$. The critical point of $P_i(\Phi_i)$ occurs at $\delta_1=\frac{m_n}{m_i}-\frac{1}{2}>0$ since $m_n>m_i$. There are two situations on the critical point. One situation is that the critical point occurs outside the domain of $\delta_i\in [0,1]$. In this situation, we have $\frac{m_n}{m_i}-\frac{1}{2}>1$ ((i.e., $\frac{m_n}{m_i}>\frac{3}{2}$)), and we have $P_i(\Phi_i)$ obtains the maximum value $Nh(m_i)(m_ip_0-2m_i\hat{c}-c_0m_i)$ at $\delta_i=1$. The other situation is that the critical point occurs within the domain of $\delta_i\in [0,1]$. In this situation, we have $\frac{m_n}{m_i}-\frac{1}{2}<1$ (i.e., $\frac{m_n}{m_i}\leq\frac{3}{2}$), and the maximum $P_i(\Phi_i)$ is $Nh(m_i)(m_ip_0-\hat{c}(2m_n-\frac{(2m_n-m_i)^2}{4m_i})-c_0m_i)$ at $\delta_i=\frac{m_n}{m_i}-\frac{1}{2}$. Thus, the second statement of the proposition follows. 


\end{proof}

\vspace{0.5cm}
\hspace{-0.4cm}\textbf{Proof of Proposition \ref{pro-bound-P1}}
\begin{proof}
To prove the proposition, it suffices to prove the gain ratio between solution $\Phi'$ and the super-optimal solution $\hat{\Phi}$ in Lemma \ref{lem-super} is no smaller than $\frac{1}{2}$, since the optimal profit is no larger than than super-optimal profit. Let $C_i$ be the capacity for a type-$m_i$ customer prepared by the supplier. At $\Phi'$, we have $C_i=m_i(1+\delta_i)\delta_i+2m_n(1-\delta_i)$. As shown in Lemma \ref{lem-Shat-bound-P1-pess-support-1}, the gain ratio between solution $\Phi'$ and the super-optimal solution is no smaller than $\min\limits_{i\in\mathcal{I}}(\frac{2m_n-C_i}{2m_n-\frac{3}{2}m_i})$. . To do this, it suffices to prove $\forall i\in \mathcal{I}, \frac{2m_n-C_i}{2m_n-\frac{3}{2}m_i}\geq \frac{1}{2}$. Since at $\Phi'$, $\delta_i$ has two cases depending on the difference between $m_i$ and $m_n$ by Proposition \ref{pro-opt-P2}, $C_i$ also has two cases depending on the difference between $m_i$ and $m_n$. 

If $\frac{m_n}{m_i}\leq\frac{3}{2}$, we have $\delta_i=\frac{m_n}{m_i}-\frac{1}{2}$ by Proposition \ref{pro-opt-P2}, and therefore $C_i=m_i(1+\delta_i)\delta_i+2m_n(1-\delta_i)=2m_n-\frac{(2m_n-m_i)^2}{4m_i}$. Thus, we have
\begin{eqnarray*}
\frac{2m_n-C_i}{2m_n-\frac{3}{2}m_i}&=&\frac{\frac{(2m_n-m_i)^2}{4m_i}}{2m_n-\frac{3}{2}m_i}\\
                                    &=&\frac{\frac{m^2_n}{m^2_i}-\frac{m_n}{m_i}+\frac{1}{4}}{\frac{2m_n}{m_i}-\frac{3}{2}}\\
																		&\geq&\frac{1}{2}.
\end{eqnarray*}
The last inequality follows since $\frac{m_n}{m_i}\in [1,\frac{3}{2}]$ and the term $\frac{\frac{m^2_n}{m^2_i}-\frac{m_n}{m_i}+\frac{1}{4}}{\frac{2m_n}{m_i}-\frac{3}{2}}$ is increasing in $\frac{m_n}{m_i}\in [1,\frac{3}{2}]$.

If $\frac{m_n}{m_i}>\frac{3}{2}$, we have $\delta_i=1$ by Proposition \ref{pro-opt-P2}, and therefore $C_i=m_i(1+\delta_i)\delta_i+2m_n(1-\delta_i)=2m_i$. Thus, we have
\begin{eqnarray*}
\frac{2m_n-C_i}{2m_n-\frac{3}{2}m_i}&=&\frac{2m_n-2m_i}{2m_n-\frac{3}{2}m_i}\\
                                    &=&\frac{\frac{2m_n}{m_i}-2}{\frac{2m_n}{m_i}-\frac{3}{2}}\\
																		&\geq&\frac{2}{3}. 
\end{eqnarray*}
The last inequality follows since $\frac{m_n}{m_i}\in [\frac{3}{2},+\infty]$ and the term $\frac{\frac{2m_n}{m_i}-2}{\frac{2m_n}{m_i}-\frac{3}{2}}$ is increasing in $\frac{m_n}{m_i}\in [\frac{3}{2},+\infty]$. Thus, the proposition holds.


\end{proof}

\vspace{0.5cm}
\hspace{-0.4cm}\textbf{Proof of Lemma \ref{lem-super}}

To prove Lemma \ref{lem-super}, we will first prove Claim \ref{claim-high-penalty-opt}.
\begin{claim}\label{claim-high-penalty-opt}
At super-optimality, the supplier will always choose a high penalty $\bar{p}_i$ for each option contract $i\in\mathcal{I}$ such that the customers will actively employ their elasticity. 
\end{claim}
\begin{proof}
By removing constraint (\ref{eqn-mechanism-design}) of Problem $\mathbb{P}_1$, the supplier's profit from all type-$m_i$ customers, which is denoted by $P_i$ is determined only by contract option $i$. Thus, to prove the claim, we will show that at the optimality of $P_i$, each contract option $i$ should be in high penalty regime (e.g., $\bar{p}_i>k$). In other words, the maximum expected profit $S^L_i$ in the low penalty regime is always not larger than the maximum expected profit $S^H_i$ in the high penalty regime. To do this, we will first compute $P_i$ under high penalty regime and low penalty regime, respectively. 

Under high penalty regime, we have 
\begin{equation}\small
\vspace{-0.1cm} 
P_i=Nh(m_i)\bigg(\big(m_ip_i\Delta_{th,i}+m_ip_0(1-\Delta_{th,i})\big)-c_0m_i-\hat{c}\big(m_i(1+\delta_i)\Delta_{th,i}+2m_n(1-\Delta_{th,i})\big)\bigg).
\end{equation}
Thus, the optimization problem for $P_i$ in high penalty regime is
\begin{equation}\label{eqn-max-high}
\max_{0\leq \delta_i\leq 1, p_i\leq p_0, \bar{p}_i\geq k, 0\leq \Delta_{th,i}\leq 1} P_i.
\end{equation}
subject to constraint 
\begin{equation*}
\Delta_{th,i}=\min(1,\frac{k\delta_i+2(p_0-p_i)+\sqrt{(k\delta_i+2(p_0-p_i))^2-k^2\delta^2_i}}{k}).
\end{equation*}

Under low penalty regime, a type-$m_i$ customer with medium variability degree $\Delta\in [\delta_i, \Delta_{th,i}]$ contributes $m_ip_i+\frac{m_i\bar{p}_i}{4\Delta}(\Delta-\delta_i)^2$ amount of revenue, which is different from the $m_ip_i$ term in the high penalty regime. Moreover, note that the upper bound of electricity bought by a customer choosing the contract is $m_i(1+\Delta_{th,i})$ instead of $m_i(1+\delta_i)$, since customers do not need to exercise their elasticity. Furthermore, we can derive the energy consumption cost of each type $m_i$ customer by the realized demand distribution of each customer with variability $\Delta$ and the distribution of $\Delta$ is
\begin{eqnarray}
C_e(m_i) &=&m_ic_0\big (\delta_i+1-\Delta_{th,i}+\frac{1}{8}(\Delta^2_{th,i}-\delta^2_i)+\frac{\delta^2_i}{4}\ln\frac{\Delta_{th,i}}{\delta_i}\nonumber\\
    & &+(1-\frac{1}{2}\delta_i)(\Delta_{th,i}-\delta_i)\big ),
\end{eqnarray}
which is different from the energy generation cost term $m_ic_0$ in the high penalty regime since in the low penalty regime no customers will leverage their elasticity to shift their demand exceeding the upper bound demand of the contract. Thus, we have
\begin{eqnarray}
\hspace{-0.3cm}P_i \hspace{-0.3cm}&=&\hspace{-0.3cm}Nh(m_i)\bigg (\Big (\int_{0}^{\delta_i}f(\Delta)m_ip_id\Delta+\int_{\delta_i}^{\Delta_{th,i}}f(\Delta)\big (m_ip_i+\frac{m_i\bar{p}_i}{4\Delta}(\Delta-\delta_i)^2\big)d\Delta\Big)+\int_{\Delta_{th,i}}^{1}m_ip_0d\Delta\nonumber\\
  & &-\hat{c} \big (m_i(1+\Delta_{th,i})+2m_n(1-\Delta_{th,i})\big)-C_e(m_i)\bigg).
\end{eqnarray}
Thus, the optimization problem for $P_i$ in low penalty regime is
\begin{equation}\label{eqn-max-low}
\max_{0\leq \delta_i\leq 1, p_i\leq p_0, \bar{p}_i\leq k, 0\leq \Delta_{th,i}\leq 1} P_i,
\end{equation}
subject to constraint 
\begin{equation*}
\Delta_{th,i}=\min(1,\frac{\bar{p}_i\delta_i+2(p_0-p_i)+\sqrt{(\bar{p}_i\delta_i+2(p_0-p_i))^2-\bar{p}^2_i\delta^2_i}}{\bar{p}_i}).
\end{equation*}
from table 1.

We first replace the $C_e(m_i)$ term in Problem (\ref{eqn-max-low}) by $c_0m_i$ and construct another optimization problem. The optimization goal of the new problem is 
\begin{eqnarray}
\hspace{-0.3cm}P_i \hspace{-0.3cm}&=&\hspace{-0.3cm}Nh(m_i)\bigg (\Big (\int_{0}^{\delta_i}f(\Delta)m_ip_id\Delta+\int_{\delta_i}^{\Delta_{th,i}}f(\Delta)\big (m_ip_i+\frac{m_i\bar{p}_i}{4\Delta}(\Delta-\delta_i)^2\big)d\Delta\Big)+\int_{\Delta_{th,i}}^{1}m_ip_0d\Delta\nonumber\\
  & &-\hat{c} \big (m_i(1+\Delta_{th,i})+2m_n(1-\Delta_{th,i})\big)-c_0m_i\bigg).
\end{eqnarray}
while the constraints of the new problem remain the same. Specifically, the new optimization problem is
\begin{equation}\label{eqn-max-low-new}
\max_{0\leq \delta_i\leq 1, p_i\leq p_0, \bar{p}_i\leq k, 0\leq \Delta_{th,i}\leq 1} P_i,
\end{equation}
subject to constraint 
\begin{equation*}
\Delta_{th,i}=\min(1,\frac{\bar{p}_i\delta_i+2(p_0-p_i)+\sqrt{(\bar{p}_i\delta_i+2(p_0-p_i))^2-\bar{p}^2_i\delta^2_i}}{\bar{p}_i}).
\end{equation*}
We can show $C_e(m_i)\geq c_0m_i$. Since 
\begin{equation}
\frac{\partial C_e(m_i)}{\partial \Delta_{th,i}}=c_0m_i(\frac{\Delta_{th,i}}{4}+\frac{\delta^2_i}{4\Delta_{th,i}}-\frac{\delta_i}{2})\geq 0,
\end{equation}
$C_e(m_i)$ is increasing in $\Delta_{th,i}$. Note that $\Delta_{th,i}\geq\delta_i$. Thus, $C_e(m_i,\Delta_{th,i})\geq C_e(m_i,\delta_i)\geq c_0m_i$. Thus, $S^L_i$ is not larger than the maximum value $\hat{S}^L_i$ under the new problem (\ref{eqn-max-low-new}). Thus, to prove the optimal $P_i$ always occur in high penalty regime, it suffices to prove $\hat{S}^L_i\leq S^H_i$. 

To prove $\hat{S}^L_i\leq S^H_i$, we can show the optimal solution of the new problem is $p_i=p_0, \Delta_{th,i}=\delta_i=\min(1,\frac{m_n}{m_i}-\frac{1}{2})=\delta, \bar{p}_i\leq k, \hat{S}^L_i=N\big(m_ip_0-(m_i\delta^2+(m_i-2m_n)\delta+2m_n)-c_0m_i\big)$. In the high penalty regime, the supplier's expected profit at $p_i=p_0, \Delta_{th,i}=\delta_i=\min(1,\frac{m_n}{m_i}-\frac{1}{2})=\delta, \bar{p}_i\geq k$ is $\hat{S}^H_i=N\big(m_ip_0-(m_i\delta^2+(m_i-2m_n)\delta+2m_n)-c_0m_i\big)$, which is equal to $\hat{S}^L_i$. Thus, since $p_i=p_0, \Delta_{th,i}=\delta_i=\min(1,\frac{m_n}{m_i}-\frac{1}{2})=\delta, \bar{p}_i>k$ is a special case in the high penalty regime, we have $S^H_i\geq \hat{S}^H_i\geq \hat{S}^L_i$. Thus, the optimal solution of $P_i$ always occur in high penalty regime. Thus, the claim holds.
\end{proof}

Now, we will prove Lemma \ref{lem-super}.
\begin{proof}
As shown in the above Claim \ref{claim-high-penalty-opt}, to find the supplier's profit, it suffices to consider high penalty regime of contract option $i$.
By only removing constraint (\ref{eqn-mechanism-design}) of Problem $\mathbb{P}_1$, we have each type-$m_i$ customer either chooses the contract option $i$designed for it or chooses the baseline pricing scheme. Thus, the supplier's profit from all type-$m_i$ customers, which is denoted by $P_i$ is determined only by contract option $i$. Thus, to optimize the supplier's profit from all customers, it suffices to optimize $P_i$ in high penalty regime.  

We have the supplier's expected profit from all type-$m_i$ customers in high penalty regime is 
\begin{eqnarray*}
P_i\!\!\!&=&\!\!\!Nh(m_i)\big(m_i(p_i\Delta_{th,i}+p_0(1-\Delta_{th,i}))-\hat{c}(m_i(1+\delta_i)\Delta_{th,i}+2m_n(1-\Delta_{th,i}))-c_0m_i\big)\\
   &=&N(h(m_i)\big(m_1(p_0-\frac{k}{4}(\Delta_{th,i}-\delta_i)^2)-\hat{c}(m_i(1+\delta_i)\Delta_{th,i}+2m_n(1-\Delta_{th,i}))-c_0m_i\big).
\end{eqnarray*}
The second equality follows by Table 1. The problem to optimize the supplier's expected profit from all type-$m_i$ customers in high penalty regime becomes
\begin{equation}
\max_{0\leq \delta_i\leq \Delta_{th,i}\leq 1} P_i.
\end{equation}

To solve the problem, we first fix $\Delta_{th,i}$ at a certain value within $[0,1]$, and optimize $\delta_i$, $(\delta_i\in [0,\Delta_{th,i}])$ under a fixed $\Delta_{th,i}$. We have $P_i$ under a fixed $\Delta_{th,i}$ is 
\begin{equation*}
P_i(\delta_i)=Nh(m_i)\big (-\frac{k}{4}m_i\delta^2_i+m_i(\frac{k}{2}-\hat{c})\delta_i\Delta_{th,i}-\frac{1}{4}m_ik\Delta^2_{th,i}-\hat{c}(2m_n-2m_n\Delta_{th,i}+m_i\Delta_{th,i})+m_ip_0-m_ic_0\big ). 
\end{equation*}
where $\delta_i\in [0,\Delta_{th,i}]$.

Since $p_0\geq 2\hat{c}$ by the assumption in Section \ref{sec-detail} and $k\geq p_0$, we have $k\geq 2\hat{c}$. Thus, the critical point of $P_i(\delta_i), \delta_i\in [0, \Delta_{th,i}]$ (i.e., $\delta_i=(1-\frac{2\hat{c}}{k})\Delta_{th,i})$ is within $[0,\Delta_{th,i}]$ and thus $P_i(\delta_i)$ obtains the maximum value 
\begin{equation*}
P^*_i(\delta_i)=Nh(m_i)\Big (\frac{m_i\hat{c}}{k}(\hat{c}-k)\Delta^2_{th,i}-\hat{c}(m_i-2m_n)\Delta_{th,i}+m_ip_0-2m_n\hat{c}-m_ic_0\Big ).
\end{equation*}
at $\delta_i=(1-\frac{2\hat{c}}{k})\Delta_{th,i}$.

We then optimize $\delta_{th,i}$. The critical point of $w^*_i(\delta_i)$, $\Delta_{th,i}\in [0,1]$ occurs at $\Delta_{th,i}=\frac{k}{2(k-\hat{c})}(\frac{2m_n}{m_i}-1)$. Note that the domain of $\Delta_{th,i}$ is $[0,1]$. Thus, there are two cases. In the first case that $\frac{k}{2(k-\hat{c})}(\frac{2m_n}{m_i}-1)>1$ (i.e., $\frac{m_n}{m_i}>\frac{k-\hat{c}}{k}+\frac{1}{2}$), then $w^*_i(\delta_i)$ is increasing in $\Delta_{th,i}\in [0,1]$ and thus the maximum value of all $P^*_i(\delta_i)$ for $\Delta_{th,i}\in [0,1]$ is 
\begin{equation}
Nh(m_i)\big(m_ip_0-m_ic_0-2m_i\hat{c}+\frac{m_i\hat{c}^2}{k}\big),
\end{equation}
where $\Delta_{th,i}=1$, and $\delta_i=1-\frac{2\hat{c}}{k}$, and $p_i=p_0-\frac{\hat{c}^2}{k}$ by Table 1.

In the second case that $\frac{k}{2(k-\hat{c})}(\frac{2m_n}{m_i}-1)\leq 1$ (i.e., $\frac{m_n}{m_i}\leq \frac{k-\hat{c}}{k}+\frac{1}{2}$), then the extremum of $w^*_i(\delta_i)$ is within the domain of $\Delta_{th,i}\in [0,1]$ and thus the maximum value of all $w^*_i(\delta_i)$ for $\Delta_{th,i}\in [0,1]$ is 
\begin{equation}
Nh(m_i)\big (m_ip_0-m_ic_0-2m_n\hat{c}+\frac{\hat{c} k(2m_n-m_i)^2}{4m_i(k-\hat{c})}\big ),
\end{equation}
where $\Delta_{th,i}=\frac{k}{2(k-\hat{c})}(\frac{2m_n}{m_i}-1)$, $\delta_i=\frac{k-2\hat{c}}{2(k-\hat{c})}(\frac{2m_n}{m_i}-1)$, and $p_i=p_0-\frac{\hat{c}^2}{2(k-\hat{c})}(\frac{2m_n}{m_i}-1)$ by Table 1.

Note that in either of the two cases, the maximum $P_i$ is decreasing in $k$. Thus, the lemma holds.

\end{proof}

\vspace{0.5cm}
\hspace{-0.4cm}\textbf{Proof of Lemma \ref{lem-Shat-bound-P1-pess-support-0}}

To prove the lemma, let $C_i$ be the capacity requirement of a type-$m_i$ of customer under the pessimistic setting at solution $\Phi'$, if the supplier sets contract defined by $\Phi'$. For example, suppose that there are two types customers with $m_1=1, m_2=1.2$ and two contract options designed for them with $\delta_1=0.7, \delta_2=0.5, p_1=p_2=10, \bar{p}_1,\bar{p}_2>k$. For type-$m_1$ customer with $\Delta\in [0,0.4]$, choosing contract option $1$ incurs the same cost as option $2$. Since choosing option $2$ results into lower profit to the supplier, we assume that under the pessimistic setting it chooses option $2$. For type-$m_1$ customer with $\Delta\in [0.4,\Delta_{th,1}]$, it chooses its dedicated option $1$. Thus, we can calculate $C_1=m_2(1+\delta_2)\cdot0.4+m_1(1+\delta_1)(\Delta_{th,1}-0.4)+2m_2(1-\Delta_{th,1})$.
We will first prove Claim \ref{claim-monotonicity} and Claim \ref{claim-support}, Lemma \ref{lem-Shat-bound-P1-pess-support-1}, Claim \ref{claim-contract-increase} that will be used to prove Lemma \ref{lem-Shat-bound-P1-pess-support-0}. Claim \ref{claim-monotonicity} states the monotonicity of cost function $\E[C_j(m_i,\Delta)]$ at any contract, which is shown as follows.
\begin{claim}\label{claim-monotonicity}
For any solution at high penalty regime (i.e., $\forall i\mathcal{I}, \bar{p}_i>k$), $\E[C_j(m_i,\Delta)], i\neq j$ has the following property. 
\begin{itemize}
\item If $m_i\in [m_{j}(1-\delta_{j}), m_{j}(1+\delta_{j})]$, then $\E[C_j(m_i,\Delta)]$ remains constant with regard to $\Delta\in [0,\Delta_{i,j}]$ and is strictly increasing in $\Delta\geq\Delta_{i,j}$, and $\E[C_j(m_i,\Delta)]\geq m_ip_{j}$.
\item If $m_i\notin [m_{j}(1-\delta_{j}), m_{j}(1+\delta_{j})]$ (i.e. $j>i, m_i<m_{j}(1-\delta_{j})$ or $j<i, m_i>m_{j}(1+\delta_{j})$), then $\E[C_j(m_i,\Delta)]$ is non decreasing in $\Delta\in [0,1]$ and $\E[C_j(m_i,\Delta)]> m_ip_{j}$. 
\end{itemize}
\end{claim}
\begin{proof}
Recall that $\E[C_j(m_i,\Delta)]$ has six cases depends on the relationship between type-$m_i$ customer's demand range (i.e., $[m_i(1-\Delta), m_i(1+\Delta)]$) and the contract range (i.e., $[m_j(1-\delta_j), m_j(1+\delta_j)]$) of option $j$ by Lemma \ref{lem-cost-wrong}. 

In case (a), (b), (f), we have $\E[C_j(m_i,\Delta)]$ remains constant with regard to $\Delta$. Considering case (c), we have function $\E[C_j(m_i,\Delta)]$ is strictly increasing in $\Delta\geq |1-\frac{m_{j}(1-\delta_{j})}{m_i}|$. Since in case (c), we have $\Delta\geq\frac{m_{j}(1-\delta_{j})}{m_i}-1$ and $\Delta\geq 1-\frac{m_{j}(1-\delta_{j})}{m_i}$. Thus, $\E[C_j(m_i,\Delta)]$ is strictly increasing in $\Delta$ in case (f). 

Considering case (d), since $k>p_{j}$, we have function $\E[C_j(m_i,\Delta)]$ is strictly increasing in $\Delta\geq\sqrt{\frac{(k-p_{j})(m_{j}(1+\delta_{j})-m_i)^2}{(k-p_{j})m^2_i}}\geq |1-\frac{m_{j}(1+\delta_{j})}{m_i}|$. Since in case (d), we have $\Delta\geq\frac{m_{j}(1+\delta_{j})}{m_i}-1$ and $\Delta\geq 1-\frac{m_{j}(1+\delta_{j})}{m_i}$. Thus, $\E[C_j(m_i,\Delta)]$ is strictly increasing in $\Delta$ in case (d). 

We will prove $\E[C_j(m_i,\Delta)]$ is strictly increasing in $\Delta$ in case (e). If $(-4\delta_{j}p_{j}+k(1+\delta_{j})^2)m^2_{j}-2(k(1+\delta_{j})-2\delta_{j}p_{j})m_im_{j}+km^2_i\leq 0$, $\E[C_j(m_i,\Delta)]$ is strictly increasing in $\Delta$ in case (e). Otherwise, we have function $\E[C_j(m_i,\Delta)]$ is strictly increasing in $\Delta\geq\sqrt{\frac{(-4\delta_{j}p_{j}+k(1+\delta_{j})^2)m^2_{j}-2(k(1+\delta_{j})-2\delta_{j}p_{j})m_im_{j}+km^2_i}{km^2_i}}$. There are two subcases: $j>i$ or $j<i$. In the first subcase that $j>i$, we have 
\begin{eqnarray*}
& &(-4\delta_{j}p_{j}+k(1+\delta_{j})^2)m^2_{j}-2(k(1+\delta_{j})-2\delta_{j}p_{j})m_im_{j}+km^2_i\\
&=&k(m_{j}(1+\delta_{j})-m_i)^2+4\delta_{j}p_{j}m_{j}(m_i-m_{j})\\
&<& k(m_{j}(1+\delta_{j})-m_i)^2. 
\end{eqnarray*}
Thus, we have
\begin{eqnarray*}
& &\sqrt{\frac{(-4\delta_{j}p_{j}+k(1+\delta_{j})^2)m^2_{j}-2(k(1+\delta_{j})-2\delta_{j}p_{j})m_im_{j}+km^2_i}{km^2_i}}\\
&<&\sqrt{\frac{k(m_{j}(1+\delta_{j})-m_i)^2}{km^2_i}}\\
&<&\frac{m_{j}(1+\delta_{j})}{m_i}-1.
\end{eqnarray*}
The last inequality follows since $m_{j}>m_i$ and $\frac{m_{j}(1+\delta_{j})}{m_i}-1>0$. Thus, we have $\E[C_j(m_i,\Delta)]$ is increasing in $\Delta\geq \frac{m_{j}(1+\delta_{j})}{m_i}-1>0$. Thus, since $\Delta\geq \frac{m_{j}(1+\delta_{j})}{m_i}-1>0$ in case (e), we have $\E[C_j(m_i,\Delta)]$ is strictly increasing in $\Delta$ in case (e) if $(-4\delta_{j}p_{j}+k(1+\delta_{j})^2)m^2_{j}-2(k(1+\delta_{j})-2\delta_{j}p_{j})m_im_{j}+km^2_i> 0$ and $j>i$. In the second subcase that $j<i$, we have  
\begin{eqnarray*}
& &\frac{(-4\delta_{j}p_{j}+k(1+\delta_{j})^2)m^2_{j}-2(k(1+\delta_{j})-2\delta_{j}p_{j})m_im_{j}+km^2_i}{km^2_i}\\
&=&\frac{k(m_{j}(1+\delta_{j})-m_i)^2+4\delta_{j}p_{j}m_{j}(m_i-m_{j})}{km^2_i}\\
&=&\frac{((1+\delta_{j})m_{j}-m_i)^2}{m^2_i}+\frac{4\delta_{j}p_{j}m_{j}(m_i-m_{j})}{km^2_i}\\
&<& \frac{((1+\delta_{j})m_{j}-m_i)^2}{m^2_i}+\frac{4\delta_{j}p_{j}m_{j}(m_i-m_{j})}{p_{j}m^2_i}\\
&<& \frac{(m_i-m_{j}(1-\delta_{j}))^2}{m^2_i}.
\end{eqnarray*}
The first inequality follows since $k>p_{j}$ and $m_i>m_{j}$. 

Thus, we have 
\begin{eqnarray*}\small
\!\!\!\!\!\!\!\!\!& &\!\!\!\sqrt{\frac{(-4\delta_{j}p_{j}+k(1+\delta_{j})^2)m^2_{j}-2(k(1+\delta_{j})-2\delta_{j}p_{j})m_im_{j}+km^2_i}{km^2_i}}\\
&<& |1-\frac{m_{j}(1-\delta_{j})}{m_i}|.
\end{eqnarray*}
Thus, since $1-\frac{m_{j}(1-\delta_{j})}{m_i}>0$, we have $\E[C_j(m_i,\Delta)]$ is strictly increasing in $\Delta\geq 1-\frac{m_{j}(1-\delta_{j})}{m_i}$. Since in case (e) we have $\Delta\geq 1-\frac{m_{j}(1-\delta_{j})}{m_i}$, then $\E[C_j(m_i,\Delta)]$ is strictly increasing in $\Delta$ in case (e) if $(-4\delta_{j}p_{j}+k(1+\delta_{j})^2)m^2_{j}-2(k(1+\delta_{j})-2\delta_{j}p_{j})m_im_{j}+km^2_i> 0$ and $j<i$. 

Based on the above analysis of $\E[C_j(m_i,\Delta)]$ for each $\Delta$ in six possible cases, we will prove the statement of the lemma. We first prove the first statement. If $j>i, m_i\in [m_{j}(1-\delta_{j}), m_{j}(1+\delta_{j}])$, then case (a), (d), (f) does not occur for $\Delta\in [0,1]$. As $\Delta$ increases in $[0,1]$, case (b),(c) occurs sequentially, or (b),(c), (e) occur sequentially. Since $\E[C_j(m_i,\Delta)]$ is constant with regard to $\Delta$ in case (b) and is strictly increasing in case (c),(e), and continuous at borderline points among the cases, and the borderline point between case (b) and (c) is at $\Delta=\Delta_{i,j}$, we have $\E[C_j(m_i,\Delta)]$ is constant with regard to $\Delta\in [0,\Delta_{i,j}]$ and is strictly increasing in $\Delta\in [\Delta_{i,j},1]$. If $j<i, m_i\in [m_{j}(1-\delta_{j}), m_{j}(1+\delta_{j})]$, then case (a), (c), (f) does not occur for $\Delta\in [0,1]$. As $\Delta$ increases in $[0,1]$, case (b),(d) occur sequentially, or case (b),(d),(e) occur sequentially. Since $\E[C_j(m_i,\Delta)]$ is constant in case (b) and strictly increasing in $\Delta$ in case (d), (e), and continuous at borderline points among the cases, and the borderline point between case (b) and (d) is at $\Delta=\Delta_{i,j}$, we have $\E[C_j(m_i,\Delta)]$ is constant with regard to $\Delta\in [0,\Delta_{i,j}]$ and strictly increasing in $\Delta\in [\Delta_{i,j},1]$. Thus, if $m_i\in [m_{j}(1-\delta_{j}), m_{j}(1+\delta_{j})]$, then $\E[C_j(m_i,\Delta)]$ remains constant with regard to $\Delta\in [0,\Delta_{i,j}]$ and is strictly increasing in $\Delta\geq\Delta_{i,j}$. Thus, since $\E[C_j(m_i,\Delta)](\Delta=0)=m_ip_{j}$, we have $\E[C_j(m_i,\Delta)]\geq m_ip_{j}$ for all $\Delta\in [0,1]$. Thus, the first statement follows.

We now prove the second statement. If $m_i\notin [m_{j}(1-\delta_{j}), m_{j}(1+\delta_{j})]$. If $j>i, m_i\notin [m_{j}(1-\delta_{j}), m_{j}(1+\delta_{j})]$, then case (b), (d), (f) does not occur for $\Delta\in [0,1]$. As $\Delta$ increases in $[0,1]$, case (a) occur, or case (a), (c) sequentially, or case (a),(c),(e) occur sequentially. Since $\E[C_j(m_i,\Delta)]$ is constant with regard to $\Delta$ in case (a), and strictly increasing with regard to $\Delta$ in case (c),(e), and continuous at borderline points among the cases, we have $\E[C_j(m_i,\Delta)]$ is non-decreasing in $\Delta\in [0,1]$. If $j<i, m_i\notin [m_{j}(1-\delta_{j}), m_{j}(1+\delta_{j})]$, then case (a), (b), (c) does not occur for $\Delta\in [0,1]$. As $\Delta$ increases in $[0,1]$, case (f) occurs, or case (f), (d) occur sequentially, or case (f), (d),(e) occur sequentially. Since $\E[C_j(m_i,\Delta)]$ is constant with regard to $\Delta$ in case (f) and is strictly increasing in $\Delta$ in case (d), (e), and continuous at borderline points among the cases, we have $\E[C_j(m_i,\Delta)]$ is non-decreasing in $\Delta\in [0,1]$. In addition, we will prove if $m_i\notin [m_{j}(1-\delta_{j}), m_{j}(1+\delta_{j})]$, then $\E[C_j(m_i,\Delta)]> m_ip_{j}$. If $j>i$, since $m_i\notin [m_{j}(1-\delta_{j}), m_{j}(1+\delta_{j})])$, we have $m_i<m_{j}(1-\delta_{j})$. hus, we have $\E[C_j(m_i,\Delta)]\geq \E[C_j(m_i,\Delta)](\Delta=0)\geq m_{j}(1-\delta_{j})p_{j}>m_ip_{j}$. If $j<i$, we have $m_i>m_{j}(1+\delta_{j})$. Thus, we have $\E[C_j(m_i,\Delta)]\geq \E[C_j(m_i,\Delta)](\Delta=0)\geq (p_{j}-k)m_{j}(1+\delta_{j})+km_i>m_ip_{j}$ since $p_{j}<k$. Thus, the second statement follows.
\end{proof}
The claim is intuitively correct. Consider type-$m_i$ customer with $\Delta$ chooses option $j$. As $\Delta$ increases, the customer has more demand surpassing the contract range of option $j$ and therefore needs to undertake more penalty or flexibility cost. Note that the claim will be also used in Lemma \ref{lem-feasible-solution}. 

We then prove Claim \ref{claim-support}, which is stated below.
\begin{claim}\label{claim-support}
Let $z=\frac{A-\sum_{i=1}^{n}x_ia_i}{A-\sum_{i=1}^{n}x_ib_i}, A, a_i, b_i, x_i>0, A-\sum_{i=1}^{n}x_ib_i>0,\forall i\in [1,n], A-b_i>0, \sum_{i=1}^{n}x_i=1$. Then we have $z\geq\min(\frac{A-a_i}{A-b_i}), i=1, \ldots, n$.
\end{claim}
\begin{proof}
Without loss of generality, we assume $\min(\frac{A-a_1}{A-b_1},\ldots,\frac{A-a_n}{A-b_n})=\frac{A-a_j}{A-b_j}$. To prove the claim, we will prove
$z-\frac{A-a_j}{A-b_j}\geq 0$. We have
\begin{eqnarray*}\small
z-\frac{A-a_j}{A-b_j}&=&\frac{w}{(A-\sum_{i=1}^{n}b_ix_i)(A-b_j)}\\
                     &=&\frac{A(a_j-b_j)+\sum_{i=1}^{n}\big((A-a_j)b_i-(A-b_j)a_i\big)x_i}{(A-\sum_{i=1}^{n}b_ix_i)(A-b_j)}.\\
\end{eqnarray*}
where $w=(A-\sum_{i=1}^{n}a_ix_i)(A-b_j)-(A-\sum_{i=1}^{n}b_ix_i)(A-a_j)$.
To prove $z-\frac{A-a_j}{A-b_j}\geq 0$, since $A-\sum_{i=1}^{n}x_ib_i>0, A-b_j>0$, it suffice to prove $A(a_j-b_j)+\sum_{i=1}^{n}\big((A-a_j)b_i-(A-b_j)a_i\big)x_i\geq 0$. We have $\forall i\in [1,n],\frac{A-b_i}{A-a_i}\geq\frac{A-a_j}{A-b_j}$. Thus, since $A-b_i>0, A-b_j>0$, we have $(A-a_i)(A-b_j)\geq (A-b_i)(A-a_j)$. Thus, we have $A(A-b_j)-a_i(A-b_j)\geq A(A-a_j)-b_i(A-a_j)$. Thus, we have
\begin{eqnarray*}
b_i(A-a_j)-a_j(A-b_i)&\geq& A(A-a_j)-A(A-b_j)\\ 
                     &\geq& A(b_j-a_j).
\end{eqnarray*}
Thus, we have $\sum_{i=1}^{n}\big((A-a_j)b_i-(A-b_j)a_i\big)x_i\geq A(b_j-a_j)\sum_{i=1}^{n}x_i\geq A(b_j-a_j)$, since $\sum_{i=1}^{n}=1$. Thus, we have
\begin{equation}
A(a_j-b_j)+\sum_{i=1}^{n}\big((A-a_j)b_i-(A-b_j)a_i\big)x_i\geq 0.
\end{equation}
Thus, the claim holds.
\end{proof}
Now we will prove Lemma \ref{lem-Shat-bound-P1-pess-support-1} as follows.
\begin{lemma}\label{lem-Shat-bound-P1-pess-support-1}
The gain ratio between solution $\Phi'$ and the super-optimal solution is no smaller than $\min\limits_{i\in\mathcal{I}}(\frac{2m_n-C_i}{2m_n-\frac{3}{2}m_i})$.
\end{lemma}

\begin{proof}
As defined in Section \ref{sec-challenge-II}, the gain ratio between solution $\Phi'$ and the super-optimal solution is $\frac{P(\Phi')-P_0}{\hat{P}^*-p_0}$, where $P(\Phi')$ is the supplier's profit at solution $\Phi'$, and $\hat{P}^*$ is the super-optimal profit, and $P_0$ is the supplier's profit without contract options. 

By Lemma \ref{lem-super}, $\hat{P}^*$ is decreasing in $k$. Thus, since $k\geq 2\hat{c}$, we have $\hat{P}^*\leq \bar{P}^*(k=2\hat{c})$. Thus, since $\hat{P}^*(k=2\hat{c})= \sum_{i=1}^{n}Nh(m_i)m_i(p_0-c_0-\frac{3}{2}\hat{c})$, we have $\hat{P}^*\leq \sum_{i=1}^{n}Nh(m_i)m_i(p_0-c_0-\frac{3}{2}\hat{c})$. Combining this with Equation (\ref{eqn-P0}), we have 
\begin{equation}
\hat{P}^*-P_0\leq 2Nm_n\hat{c}-\frac{3}{2}Nh(m_1)m_1\hat{c}-\ldots-\frac{3}{2}Nh(m_n)m_n\hat{c}.
\end{equation}

Besides, we have $P(\Phi')=\sum_{i=1}^{n}Nh(m_i)(m_ip_0-\hat{c}C_i-c_0m_i)$. To prove this, we have each type-$m_i$ customer contributes $m_ip_0$ amount of revenue to the supplier even if it wrongly chooses other contract options. Also, we have each customer's mean energy consumption is $m_i$. Since $\Delta_{th,i}=\delta_i$, we divide type-$m_i$ customers into two classes by $\Delta\in [0,\delta_i]$ and $\Delta\in [\delta_i, 1]$. In particular, for each type-$m_i$customer with $\Delta\leq\delta_i$, it may choose option $i$ or other option $j$. If it chooses option $i$, its mean energy consumption is $m_i$. If it chooses other option $j$, it must have $m_i\in [m_{j}(1-\delta_{j}),m_{j}(1+\delta_{j})]$ and $\Delta\leq\Delta_{i,j}$. To show this, we have $\E[C_i(m_i,\Delta)]=m_ip_0$. If $m_i\notin [m_{j}(1-\delta_{j}),m_{j}(1+\delta_{j})]$, we have $\E[C_j(m_i,\Delta)]>m_ip_0$ by Claim \ref{claim-monotonicity} and $\E[C_j(m_i,\Delta)]>\E[C_i(m_i,\Delta)]$, and thus the customer will not choose option $j$. If $m_i\in [m_{j}(1-\delta_{j}),m_{j}(1+\delta_{j})], \Delta>\Delta_{i,j}$, we have $\E[C_j(m_i,\Delta)]>m_ip_0$ by Claim \ref{claim-monotonicity}, and $\E[C_j(m_i,\Delta)]>\E[C_i(m_i,\Delta)]$, and thus it will not choose option $j$. Thus, we have if type-$m_i$ customer with $\Delta\in [0,\delta_i]$ chooses other option $j$, it must have $m_i\in [m_{j}(1-\delta_{j}),m_{j}(1+\delta_{j})]$ and $\Delta\leq\Delta_{i,j}$. Thus, we have type-$m_i$ customer's whole usage range is contained in option $i$'s contract range. Thus, its demand is not changed by the customer and its mean energy consumption is $m_i$. For each type-$m_i$ customer with $\Delta>\delta_i$, it chooses baseline pricing scheme and its mean energy consumption is also $m_i$ since its demand within $[m_i(1-\Delta),m_i(1+\Delta)]$ is not changed by the customer in order not to incur high flexibility cost.

Thus, we have 
\begin{equation}
P(\Phi')-P_0=N\hat{c}(2m_n-h(m_1)C_1-\ldots-h(m_n)C_n).
\end{equation}
Note that $P(\Phi')-P_0>0$ since we can derive $\forall i\in [1,n], C_i<2m_n$.
Thus, we have 
\begin{eqnarray*}
\frac{P(\Phi')-P_0}{\hat{P}^*-P_0}&\geq& \frac{P(\Phi')-P_0}{\hat{P}^*(k=2\hat{c})-P_0}\\
                                  &\geq& \frac{2m_n-h(m_1)C_1-\ldots-h(m_n)C_n}{2m_n-\frac{3}{2}h(m_1)m_1-\frac{3}{2}h(m_2)m_2-\ldots-\frac{3}{2} h(m_n)m_n}\\
								                  &\geq& \min\limits_{i\in\mathcal{I}}(\frac{2m_n-C_i}{2m_n-\frac{3}{2}m_i}).
\end{eqnarray*}
The last inequality follows by Claim \ref{claim-support}. Hence, the lemma holds. 
\end{proof}


\begin{claim}\label{claim-contract-increase}
$\Phi'$ has the following property: 
\begin{itemize}
\item $m_1(1+\delta_1)<m_2(1+\delta_2)<\ldots<m_n(1+\delta_n)$; 
\item $m_1(1-\delta_1)\leq m_2(1-\delta_2)\leq\ldots\leq m_n(1-\delta_n)$.
\end{itemize}
\end{claim}

\begin{proof}
To prove the claim, it suffices to prove $m_i(1+\delta_i)<m_{i'}(1+\delta_{i'})$ and $m_i(1-\delta_i)\leq m_{i'}(1-\delta_{i'})$ for any two contract options $i,i',(i<i')$ at $\Phi'$. 

There are two cases. In the first case that $\frac{m_n}{m_1}\leq\frac{3}{2}$. We have $\delta_i=\frac{m_n}{m_i}-\frac{1}{2}$ and $\delta_{i'}=\frac{m_n}{m_{i'}}-\frac{1}{2}$ by Proposition \ref{pro-opt-P2}. Thus, we have $m_i(1+\delta_i)-m_{i'}(1+\delta_{i'})=\frac{1}{2}(m_i-m_{i'})<0$ since $m_i<m_{i'}$. Also, we have $m_i(1-\delta_i)-m_{i'}(1-\delta_{i'})=\frac{3}{2}(m_i-m_{i'})<0$ since $m_i<m_{i'}$.

In the second case that $\frac{m_n}{m_1}>\frac{3}{2}$. Assume $\frac{m_n}{m_j}>\frac{3}{2}$ and $\frac{m_n}{m_{j+1}}\leq \frac{3}{2}$, $j=1,,\ldots, n-1$. If $i,i'\leq j$, we have $\delta_i=\delta_{i'}=1$  by Proposition \ref{pro-opt-P2}. Thus, we have $m_i(1+\delta_i)-m_{i'}(1+\delta_{i'})=2(m_i-m_{i'})<0$, and $m_i(1-\delta_i)-m_{i'}(1-\delta_{i'})=0$. If $i,i'>j$, we have $\delta_i=\frac{m_n}{m_i}-\frac{1}{2}$ and $\delta_{i'}=\frac{m_n}{m_{i'}}-\frac{1}{2}$. Thus, we have $m_i(1+\delta_i)-m_{i'}(1+\delta_{i'})=\frac{1}{2}(m_i-m_{i'})<0$, and $m_i(1-\delta_i)-m_{i'}(1-\delta_{i'})=\frac{3}{2}(m_i-m_{i'})<0$. If $i\leq j$ and $i'>j$, we have $\delta_i=1, \delta_{i'}=\frac{m_n}{m_{i'}}-\frac{1}{2}$. Thus, we have $m_i(1+\delta_i)-m_{i'}(1+\delta_{i'})=\frac{1}{2}(4m_i-2m_n-m_{i'})<0$. Since $m_i<\frac{2}{3}m_n, m_{i'}\geq\frac{2}{3}m_n$, we have $4m_i<\frac{8}{3}m_n$ and $2m_n+m_{i'}\geq\frac{8}{3}m_n$, and $4m_i-2m_n-m_{i'}<0$. Thus, we have $m_i(1+\delta_i)-m_{i'}(1+\delta_{i'})<0$. Also, we have $m_i(1-\delta_i)-m_{i'}(1-\delta_{i'})=-m_{i'}(\frac{3}{2}-\frac{m_n}{m_{i'}})\leq 0$ since $\frac{m_n}{m_{i'}}\leq\frac{3}{2}$.

Overall, the claim follows.
\end{proof}

Next, we will prove Lemma \ref{lem-Shat-bound-P1-pess-support-0}.
\begin{proof}
As shown in Lemma \ref{lem-Shat-bound-P1-pess-support-1}, the gain ratio between solution $\Phi'$ and the super-optimal solution is no smaller than $\min\limits_{i\in\mathcal{I}}(\frac{2m_n-C_i}{2m_n-\frac{3}{2}m_i})$. For ease of presentation later, let $r_g$ represent the gain ratio between solution $\Phi'$ and the super-optimal solution. Thus, to prove $r_g$, it suffices to prove $\min\limits_{i\in\mathcal{I}}(\frac{2m_n-C_i}{2m_n-\frac{3}{2}m_i})\geq \frac{1}{3}$. Without loss of generality, we assume $\min\limits_{i=1,\ldots,n}(\frac{2m_n-C_i}{2m_n-\frac{3}{2}m_i})=\frac{2m_n-C_1}{2m_n-\frac{3}{2}m_1}$. Since $m_1(1+\delta_1)<m_2(1+\delta_2)<\ldots<m_n(1+\delta_n)$ by Claim \ref{claim-contract-increase}, if $\min\limits_{i=1,\ldots, n}(\frac{2m_n-C_i}{2m_n-\frac{3}{2}m_i})=\frac{2m_n-C_{i'}}{2m_n-\frac{3}{2}m_{i'}}$, it can be seen that $C_{i'}$ is $C_1$ and there are $n-i'+1$ options. 

Now we discuss all possible contract option choices of type-$m_1$ customer and find the lower bound of $\frac{2m_n-C_1}{2m_n-\frac{3}{2}m_1}$ in all cases. Generally, there are three cases. In the first case shown in Figure \ref{cost-function-case1}, we assume $\frac{m_n}{m_1}\leq\frac{3}{2}$. In the second case shown in figure \ref{cost-function-case2}, we assume $\frac{m_n}{m_1}\geq 2$. In the third case shown in figure \ref{cost-function-case3}, we assume $\frac{3}{2}\leq \frac{m_n}{m_1}\leq 2$.

\begin{figure}
	\centering 
		\subfigure[Case 1: $\frac{m_n}{m_1}\leq\frac{3}{2}$]{ 
		\label{cost-function-case1} 
		\includegraphics[height=5cm, width=6cm]{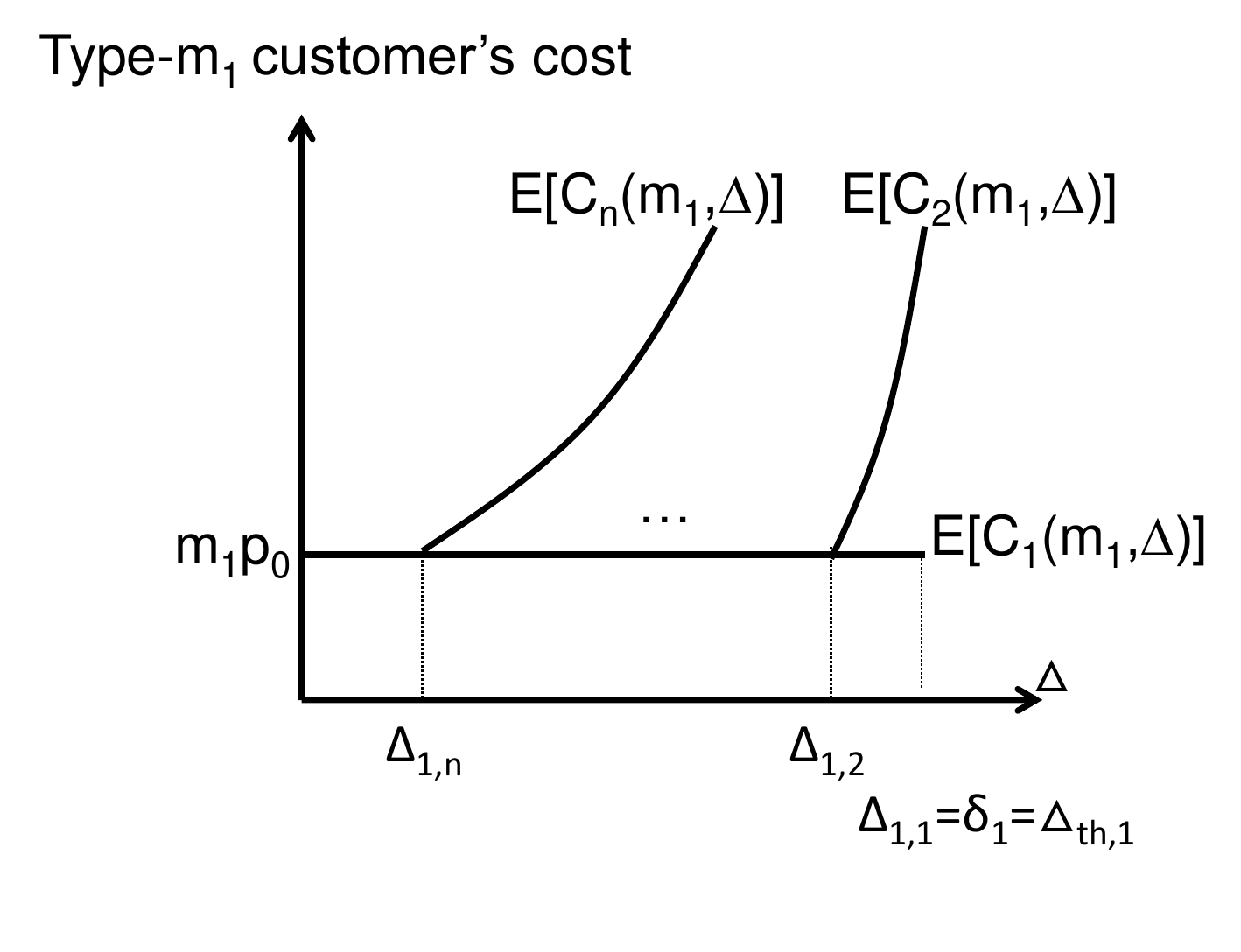}} 
	\subfigure[Case 2: $\frac{m_n}{m_1}\geq 2$]{ 
		\label{cost-function-case2} 
		\includegraphics[height=5cm, width=6cm]{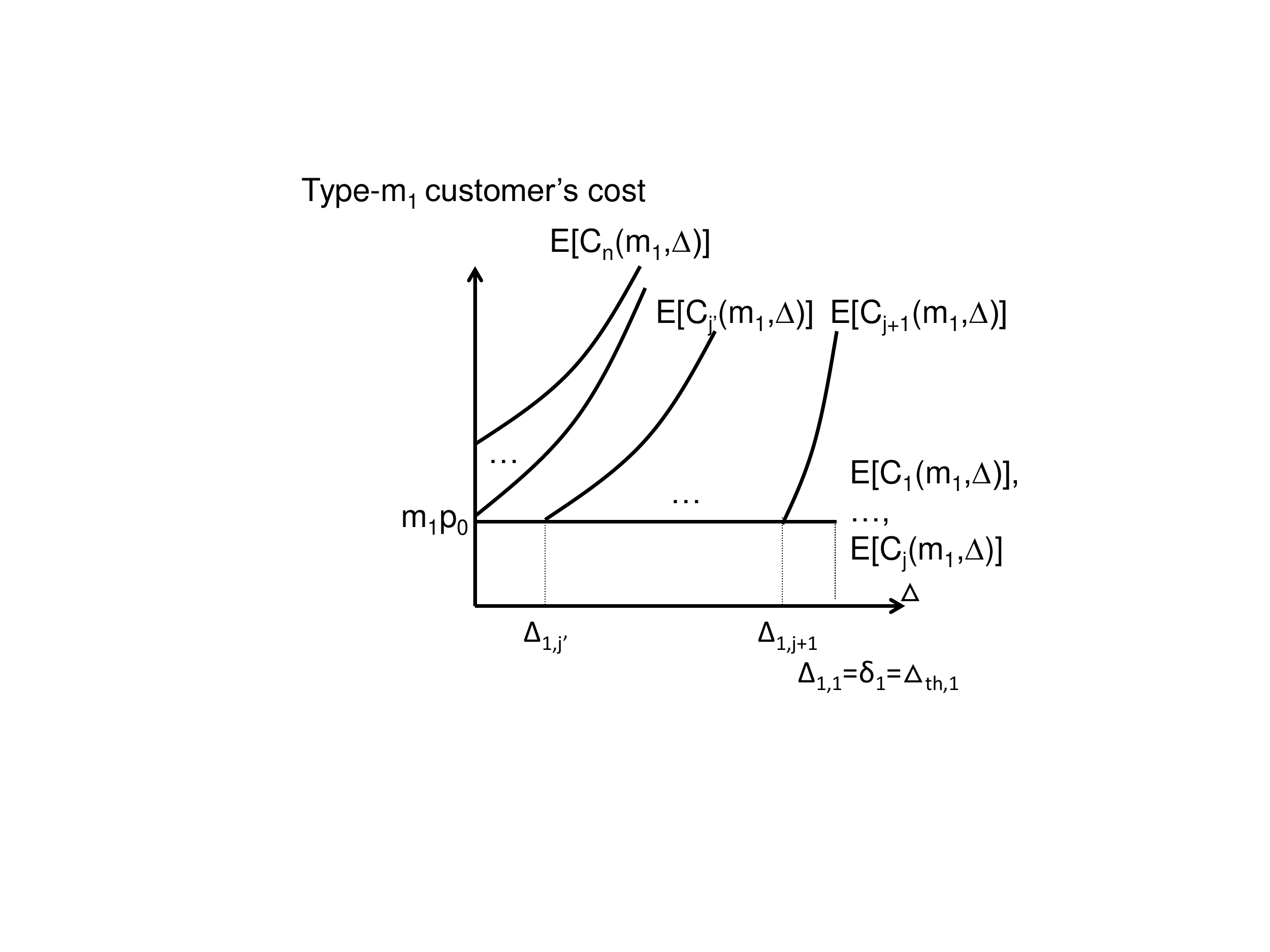}} 
	\subfigure[Case 3: $\frac{3}{2}<\frac{m_n}{m_1}<2$]{ 
		\label{cost-function-case3} 
		\includegraphics[height=5cm, width=6cm]{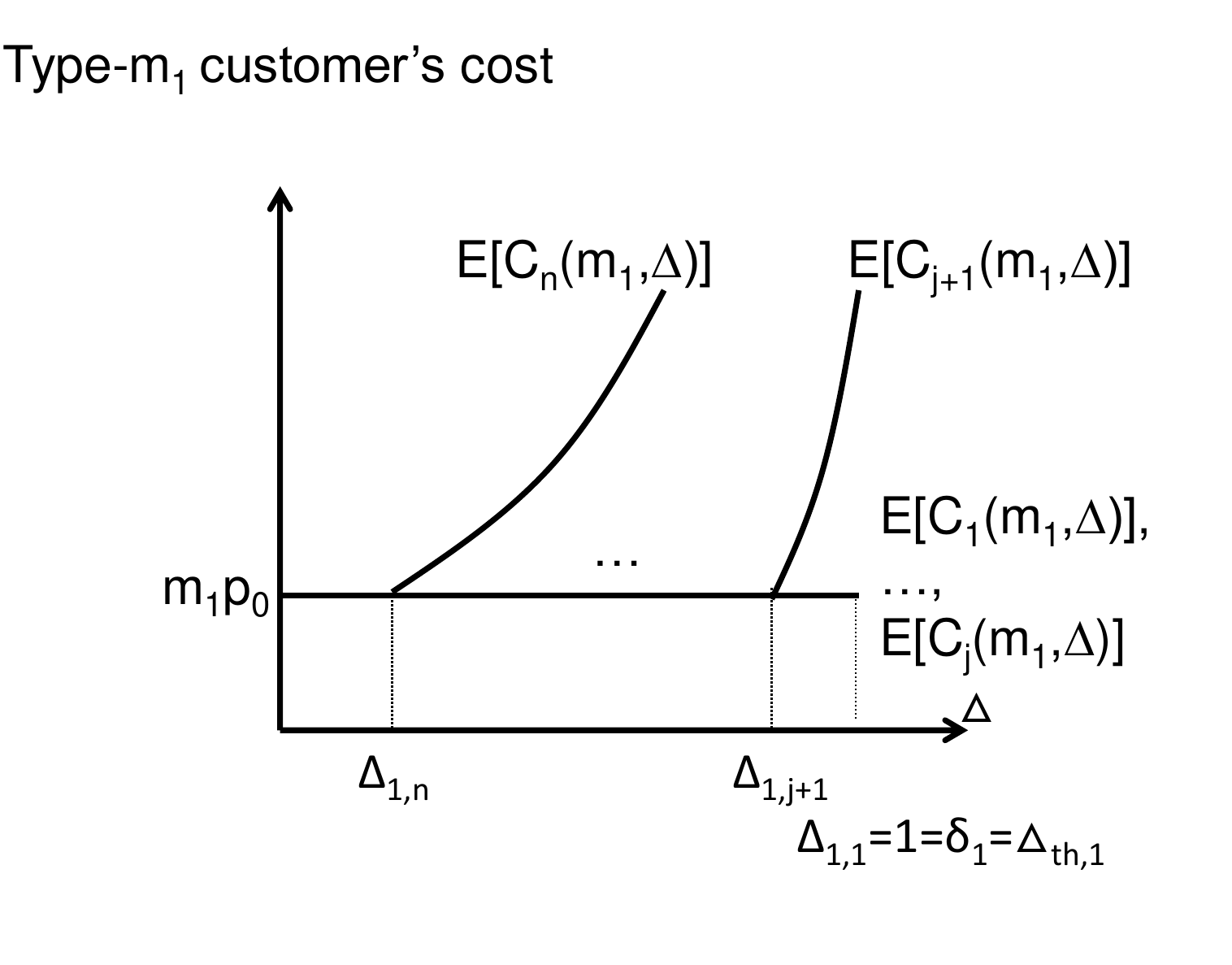}}  
	\caption{Type-$m_i$ customer's cost functions with $\Delta$ when choosing different options} 
\end{figure}
To prove the lemma, we will prove $r_g>0.476$ if $\frac{m_n}{m_1}\leq \frac{3}{2}$, $r_g>\frac{1}{3}$ if $\frac{m_n}{m_1}\geq 2$ and $r_g>0.466$ if $\frac{3}{2}<\frac{m_n}{m_1}<2$. 

\textbf{1)} In the first case, $\Phi'$ is $\delta_i=\frac{m_n}{m_i}-\frac{1}{2}, i=1,\ldots, n$ by Proposition \ref{pro-opt-P2}. We have $m_1(1+\delta_1)<m_2(1+\delta_2)<\ldots<m_n(1+\delta_n)$ by Claim \ref{claim-contract-increase}. We have $\Delta_{1,i}=1-\frac{m_i(1-\delta_i)}{m_1}$, $i=1,\ldots, n$. Thus, since $\delta_i=\frac{m_n}{m_i}-\frac{1}{2}, i=1,\ldots, n$, we have $\delta_1=\Delta_{1,1}>\Delta_{1,2}>\ldots>\Delta_{1,n}>0$ by Claim \ref{claim-contract-increase}. For a $m_1$ type of customer with $\Delta\in [\Delta_{1,i+1},\Delta_{1,i}], i=1,\ldots, n-1$, we have $\E[C_1(m_1,\Delta)]=\E[C_2(m_1,\Delta)]=\ldots=\E[C_i(m_1,\Delta)]=m_1p_0<\E[C_{i'}(m_1,\Delta)], i'>i$. The supplier can choose any option among options $1,\ldots, i$. Since $m_1(1+\delta_1)<m_2(1+\delta_2)<\ldots<m_i(1+\delta_i)$, we assume $m_1$ type of customer choose contract option $i$ since the estimated capacity of the customer is $m_i(1+\delta_i)$ which is the worst for the supplier, the revenue it contributes to the supplier is $m_ip_0$ whatever option it chooses, the mean energy consumption it results in is $m_i$ whatever option it chooses, and therfore the profit it chooses option $i$ is the worst. For a $m_1$ type of customer with $\Delta\in [0,\Delta_{1,n}]$, we have $\E[C_1(m_1,\Delta)]=C_{1,2}=\ldots=\E[C_n(m_1,\Delta)]=m_1p_0$ and we assume it chooses option $n$ for the worst-case profit of the supplier. Besides, type-$m_1$ customer with $\Delta\in [\delta_1, 1]$ subscribes to the baseline pricing scheme. Thus, we have 
\begin{eqnarray*}\small
   C_1&=&2m_n(1-\delta_1)+\sum_{i=1}^{n-1}m_i(1+\delta_i)(\Delta_{1,i}-\Delta_{1,i+1})+m_n(1+\delta_n)\Delta_{1,n}\\
   &=&2m_n(1-\delta_1)+\sum_{i=1}^{n-1}m_i(1+\delta_i)(\frac{m_{i+1}(1-\delta_{i+1})}{m_1}-\frac{m_i(1-\delta_i)}{m_1})+m_n(1+\delta_n)(1-\frac{m_n(1-\delta_n)}{m_1})\\
	 &=&2m_n(1-(\frac{m_n}{m_1}-\frac{1}{2}))+\sum_{i=1}^{n-1}\frac{3}{2}m_i(1+\frac{m_n}{m_i}-\frac{1}{2})\frac{m_{i+1}-m_i}{m_1}+\frac{3}{2}m_n(1-\frac{m_n}{2m_1})\\
   &=&3m_n-\frac{2m^2_n}{m_1}+\sum_{i=1}^{n-1}(\frac{3m_im_{i+1}}{4m_1}-\frac{3m^2_i}{4m_1}+\frac{3m_{i+1}m_n}{2m_1}-\frac{3m_im_n}{2m_1})+\frac{3m_n}{2}-\frac{3m^2_n}{4m_1}\\
	 &=&3m_n-\frac{2m^2_n}{m_1}+(\sum_{i=1}^{n-1}\frac{3m_im_{i+1}}{4m_1}-\sum_{i=1}^{n-1}\frac{3m^2_i}{4m_1}+\frac{3m^2_n}{2m_1}-\frac{3}{2}m_n)+\frac{3m_n}{2}-\frac{3m^2_n}{4m_1}\\
	 &=&3m_n-\frac{2m^2_n}{m_1}+\big(\frac{3}{4m_1}\sum_{i=1}^{n-1}(m_im_{i+1}-m^2_i)+\frac{3m^2_n}{2m_1}-\frac{3}{2}m_n\big)+\frac{3m_n}{2}-\frac{3m^2_n}{4m_1}.
\end{eqnarray*}
 
Thus, to find the lower bound of $\frac{2m_n-C_1}{2m_n-\frac{3}{2}m_1}$, we will find the maximum value of $C_1$, considering $m_1,m_n, n$ as fixed values. Note that $m_1<m_2<\ldots< m_n$. We have 
\begin{eqnarray*}
         C_1&=&3m_n-\frac{2m^2_n}{m_1}+\big(-\frac{3}{8m_1}\sum_{i=1}^{n-1}(2m^2_i-2m_im_{i+1})+\frac{3m^2_n}{2m_1}-\frac{3}{2}m_n\big)+\frac{3m_n}{2}-\frac{3m^2_n}{4m_1}\\
            &=&3m_n-\frac{2m^2_n}{m_1}+\big(-\frac{3}{8m_1}(m^2_1-m^2_n+\sum_{i=1}^{n-1}(m_{i+1}-m_i)^2)+\frac{3m^2_n}{2m_1}-\frac{3}{2}m_n\big)+\frac{3m_n}{2}-\frac{3m^2_n}{4m_1}\\
				    &\leq&3m_n-\frac{2m^2_n}{m_1}+\big(-\frac{3}{8m_1}(m^2_1-m^2_n+\frac{(m_n-m_1)^2}{n-1})+\frac{3m^2_n}{2m_1}-\frac{3}{2}m_n\big)+\frac{3m_n}{2}-\frac{3m^2_n}{4m_1}\\
					  &\leq& 3m_n-\frac{2m^2_n}{m_1}+\big(\frac{3(5n-6)}{8(n-1)}\frac{m^2_n}{m_1}+\frac{3(3-2n)}{4(n-1)}m_n-\frac{3n}{8(n-1)}m_1\big)+\frac{3m_n}{2}-\frac{3m^2_n}{4m_1}\\
					  &\leq&\big(\frac{9}{2}+\frac{3(3-2n)}{4(n-1)}\big)m_n+\big(\frac{3(5n-6)}{8(n-1)}-\frac{11}{4}\big)\frac{m^2_n}{m_1}-\frac{3n}{8(n-1)}m_1.
\end{eqnarray*}
To show the first inequality follows, it suffices to show $\sum_{i=1}^{n-1}(m_{i+1}-m_i)^2\geq\frac{(m_n-m_1)^2}{n-1}$. To do this, let $x_i=m_{i+1}-m_i, i\in [1,n-1]$. We have $x_1+x_2+\ldots+x_{n-1}=m_n-m_1$. Thus, it is equivalent to prove $\sum_{i=1}^{n-1}x^2_i\geq\frac{(\sum_{i=1}^{n-1}x_i)^2}{n-1}$
. We can prove the inequality by Cauchy-Schwarz inequality. We have $(n-1)\sum_{i=1}^{n-1}x^2_i=(1^2+1^2+\ldots+1^2)(x^2_1+x^2_2+\ldots+x^2_{n-1})\geq (x_1+x_2+\ldots+x_{n-1})^2$ by Cauchy-Schwarz inequality. Thus, we have $\sum_{i=1}^{n-1}x^2_i\geq\frac{(\sum_{i=1}^{n-1}x_i)^2}{n-1}$ and the first inequality follows.

Taking the upper bound of $C_1$ into $\frac{2m_n-C_1}{2m_n-\frac{3}{2}m_1}$, we have 
\begin{eqnarray*}\tiny
  \frac{2m_n-C_1}{2m_n-\frac{3}{2}m_1} &\geq&\frac{\frac{(-4n+1)m_n}{4(n-1)}+\frac{3nm_1}{8(n-1)}+\frac{(7n-4)m^2_n}{8(n-1)m_1}}{2m_n-\frac{3}{2}m_1}\\
   &\geq&\frac{(-8n+2)\frac{m_n}{m_1}+(7n-4)\frac{m^2_n}{m^2_1}+3n}{8(n-1)(\frac{2m_n}{m_1}-\frac{3}{2})}\\
   &\geq&\frac{-\frac{8m_n}{m_1}+\frac{7m^2_n}{m^2_1}+3}{8(\frac{2m_n}{m_1}-\frac{3}{2})}\\
   &\geq& 0.476.
\end{eqnarray*}
We first prove the third inequality follows. Let $y$ represent the term on the right of the second inequality. We have $\frac{\partial y}{\partial n}=\frac{-3(\frac{m_n}{m_1}-1)^2}{8(n-1)(\frac{2m_n}{m_1}-\frac{3}{2})}<0$ since $\frac{m_n}{m_1}>1$. Thus, $y$ is decreasing in $n$. Thus, since $n<+\infty$, we have the third inequality follows. We will prove the last inequality follows. Let $z$ represent the term on the right of the third inequality and $x=\frac{m_n}{m_1}\in [1,\frac{3}{2}]$. To prove the last inequality follows, we will prove $z=\frac{-8x+7x^2+3}{8(2x-\frac{3}{2})}\geq 0.476, x\in [1,\frac{3}{2}]$. We have $\frac{\partial z}{\partial x}=\frac{14x^2-21x+6}{8(2x-\frac{3}{2})^2}$. For equation $14x^2-21x+6=0, x\in [1,\frac{3}{2}]$, we have $x=\frac{21+\sqrt{105}}{28}\in (1,\frac{3}{2})$. Thus, since the function $14x^2-21x+6, x\in [1,\frac{3}{2}]$ is increasing in $x\in [1,\frac{3}{2}]$, we have $14x^2-21x+6\leq 0$ for $x\in [1,\frac{21+\sqrt{105}}{28}]$ and $14x^2-21x+6>0$ for $x\in [\frac{21+\sqrt{105}}{28},\frac{3}{2}]$. Thus, since $8(2x-\frac{3}{2})^2>0$, we have $\frac{\partial z}{\partial x}\leq 0$ for $x\in [1,\frac{21+\sqrt{105}}{28}]$ and $\frac{\partial z}{\partial x}>0$ for $x\in [\frac{21+\sqrt{105}}{28},\frac{3}{2}]$. Thus, we have $z(x), x\in [1,\frac{3}{2}]$ obtains the minimum value at $x=\frac{21+\sqrt{105}}{28}\in (1,\frac{3}{2})$. Thus, we have $z\geq z(x=\frac{21+\sqrt{105}}{28})\geq 0.476$ and the last inequality follows.

\textbf{2)} In the second case, we assume $\frac{m_n}{m_j}>\frac{3}{2}$ and $\frac{m_n}{m_{j+1}}\leq\frac{3}{2}, j=1,\ldots, n-1$. $\Phi'$ is $\delta_i=1, i=1,\ldots, j$, $\delta_i=\frac{m_n}{m_i}-\frac{1}{2}, i=j+1,\ldots, n$ by Proposition \ref{pro-opt-P2}. We also have $m_1(1+\delta_1)<m_2(1+\delta_2)<\ldots<m_n(1+\delta_n)$ and $m_1(1-\delta_1)>m_2(1-\delta_2)>\ldots<m_n(1-\delta_n)$. Assume $m_1\geq m_{j'}(1-\delta_{j'})$ and $m_1<m_{j'+1}(1-\delta_{j'+1})$, $j'=j+1, \ldots, n-1$. We have $\Delta_{1,1}=\Delta_{1,2}=\ldots=\Delta_{1,j}=1>\ldots>\Delta_{1,j'}\geq 0$. Since $\E[C_i(m_1,\Delta)]>m_1p_0, i=j'+1,\ldots, n$ by Claim \ref{claim-monotonicity} and $\E[C_1(m_1,\Delta)]=m_1p_0$ for $\Delta\in [0,1]$, type-$m_1$ customer will not choose options $j'+1,\ldots, n-1$. For a type-$m_1$ customer with $\Delta\in [\Delta_{1,i+1},\Delta_{1,i}], i=j+1,\ldots, j'-1$, we have $\E[C_1(m_1,\Delta)]=\E[C_2(m_1,\Delta)]=\ldots=\E[C_i(m_1,\Delta)]<\E[C_{i'}(m_1,\Delta)], i'>i$ and the customer may choose any contract option among options $1,\ldots, i$. Since $m_1(1+\delta_1)<m_2(1+\delta_2)<\ldots<m_i(1+\delta_i)$, we assume $m_1$ type of customer choose contract option $i$ since the estimated capacity of the customer is $m_i(1+\delta_i)$ which causes the worst-case profit for the supplier. For a $m_1$ type of customer with $\Delta\in [0,\Delta_{1,j'}]$, we have $\E[C_1(m_1,\Delta)]=\E[C_2(m_1,\Delta)]=\ldots=\E[C_{j'}(m_1,\Delta)]=m_1p_0$ and we assume it chooses option $j'$ for the worst-case profit of the supplier. Besides, type-$m_1$customer with $\Delta\in [\Delta_{1,j+1}, 1]$ subscribes to option $j$ for the worst-case profit of the supplier. Thus, we have 
\begin{eqnarray*}
 C_1&=&m_j(1+\delta_j)(1-\Delta_{1,j+1})+\sum_{i=j+1}^{j'-1}m_i(1+\delta_i)(\Delta_{1,i}-\Delta_{1,i+1})+m_{j'}(1+\delta_{j'})\Delta_{1,j'}\\
    &=&m_j(1+\delta_j)(1-(1-\frac{m_{j+1}(1-\delta_{j+1})}{m_1}))+\sum_{i=j+1}^{j'-1}m_i(1+\delta_i)(\frac{m_{i+1}(1-\delta_{i+1})}{m_1}-\frac{m_i(1-\delta_i)}{m_1})\\
    & &+m_{j'}(1+\delta_{j'})(1-\frac{m_{j'}(1-\delta_{j'})}{m_1})\\
    &=&\frac{2m_jm_{j+1}}{m_1}(1-(\frac{m_n}{m_{j+1}}-\frac{1}{2}))+\sum_{i=j+1}^{j'-1}\frac{3}{2}m_i(1+\frac{m_n}{m_i}-\frac{1}{2})\frac{m_{i+1}-m_i}{m_1}\\
    & &+m_{j'}(1+\frac{m_n}{m_{j'}}-\frac{1}{2})(1-\frac{m_{j'}}{m_1}(1-\frac{m_n}{m_{j'}}+\frac{1}{2}))\\
    &=&\frac{3m_j m_{j+1}}{m_1}-\frac{2m_jm_n}{m_1}+\sum_{i=j+1}^{j'-1}\frac{3}{2m_1}(\frac{m_im_{i+1}}{2}-\frac{m^2_i}{2}+m_nm_{i+1}-m_nm_i)\\
    & &+(\frac{1}{2}m_{j'}-\frac{3m^2_{j'}}{4m_1}-\frac{m_nm_{j'}}{m_1}+m_n+\frac{m^2_n}{m_1}).
\end{eqnarray*}

We will find the upper bound of $C_1$. To do this, we have 
\begin{eqnarray*}
 C_1&=&\frac{3m_j m_{j+1}}{m_1}-\frac{2m_jm_n}{m_1}+(-\frac{3}{8m_1})\sum_{i=j+1}^{j'-1}(2m^2_i-2m_im_{i+1}+4m_n(m_i-m_{i+1}))\\
    & &+(\frac{1}{2}m_{j'}-\frac{3m^2_{j'}}{4m_1}-\frac{m_nm_{j'}}{m_1}+m_n+\frac{m^2_n}{m_1})\\
		&=&\frac{3m_j m_{j+1}}{m_1}-\frac{2m_jm_n}{m_1}+(-\frac{3}{8m_1})\big((m^2_{j+1}-m^2_{j'})+\sum_{i=j+1}^{j'-1}(m_i-m_{i+1})^2+4m_n(m_{j+1}-m_{j'})\big)\\
    & &+(\frac{1}{2}m_{j'}-\frac{3m^2_{j'}}{4m_1}-\frac{m_nm_{j'}}{m_1}+m_n+\frac{m^2_n}{m_1})\\
		&\leq&\frac{3m_j m_{j+1}}{m_1}-\frac{2m_jm_n}{m_1}+(-\frac{3}{8m_1})\big((m^2_{j+1}-m^2_{j'})+\frac{(m_{j'}-m_{j+1})^2}{j'-j-1}+4m_n(m_{j+1}-m_{j'})\big)\\
    & &+(\frac{1}{2}m_{j'}-\frac{3m^2_{j'}}{4m_1}-\frac{m_nm_{j'}}{m_1}+m_n+\frac{m^2_n}{m_1})\\
		&\leq&\frac{3m_j m_{j+1}}{m_1}-\frac{2m_jm_n}{m_1}+\frac{3m^2_{j'}(j'-j-2)}{8(j'-j-1)m_1}-\frac{3(j'-j)m^2_{j+1}}{8(j'-j-1)m_1}+\frac{3m_{j'}m_{j+1}}{4(j'-j-1)m_1}\\
    & &+\frac{3m_n(m_{j'}-m_{j+1})}{2m_1}+(\frac{1}{2}m_{j'}-\frac{3m^2_{j'}}{4m_1}-\frac{m_nm_{j'}}{m_1}+m_n+\frac{m^2_n}{m_1}).
\end{eqnarray*}
The first inequality follows by the inequality $\sum_{i=1}^{n-1}(m_{i+1}-m_i)^2\geq\frac{(m_n-m_1)^2}{n-1}$, which has been proved in the first case. 
Since the term on the right of the last inequality is increasing in $m_j$ and $m_j\leq\frac{2}{3}m_n$, we have
\begin{eqnarray*}
C_1&\leq&\frac{m'_j}{2}-\frac{3(j'-j)}{8(j'-j-1)}\frac{m^2_{j'}}{m_1}+\frac{m_nm_{j'}}{2m_1}+m_n-\frac{m^2_n}{3m_1}-\frac{3(j'-j)}{8(j'-j-1)}\frac{m^2_{j+1}}{m_1}+\frac{3}{4(j'-j-1)}\frac{m'_jm_{j+1}}{m_1}+\frac{m_nm_{j+1}}{2m_1}\\
   &\leq&\frac{m_{j'}}{2}-\frac{3m^2_{j'}}{8m_1}+\frac{m_nm_{j'}}{2m_1}+m_n-\frac{m^2_n}{3m_1}-\frac{3m^2_{j+1}}{8m_1}+\frac{m_nm_{j+1}}{2m_1}\\
   &\leq&\frac{m_{j'}}{2}-\frac{3m^2_{j'}}{8m_1}+\frac{m_nm_{j'}}{2m_1}+m_n-\frac{m^2_n}{6m_1}\\	
	 &\leq&\frac{1}{6}m_1+\frac{4}{3}m_n.
\end{eqnarray*}
The second inequality follows since $j'-j<+\infty$ and the term on the right of the first inequality is increasing in $j'-j$. To see this, let $y$ represent the term on the right of the first inequality and $t=j'-j$. We have $\frac{\partial y}{\partial t}=\frac{3(m_{j'}-m_{j+1})^2}{8m_1(t-1)^2}\geq 0$. The third inequality follows since $m_{j+1}\geq\frac{2}{3}m_n$ and the term on the right of the second inequality is decreasing in $m_{j+1}$. To see this, let $z$ represent the term on the right of the second inequality. We have $\frac{\partial z}{\partial m_{j+1}}=-\frac{1}{4m_1}(3m_{j+1}-2m_n)\leq 0$  since $m_{j+1}\geq\frac{2}{3}m_n$. The last inequality follows since the term on the right of the third inequality is concave as a function of $m_{j'}$ and obtains the maximum at $m_{j'}=\frac{2}{3}(m_1+m_n)\in (\frac{2}{3}m_n, m_n)$.


Thus, we have
\begin{eqnarray*}\tiny
 \frac{2m_n-C_1}{2m_n-\frac{3}{2}m_1} &\geq&\frac{\frac{2}{3}m_n-\frac{1}{6}m_1}{2m_n-\frac{3}{2}m_1}\\
   &\geq&\frac{\frac{2}{3}\frac{m_n}{m_1}-\frac{1}{6}}{2\frac{m_n}{m_1}-\frac{3}{2}}\\
   &\geq&\frac{1}{3}.
\end{eqnarray*}
The third inequality follows since the term on the right of the second inequality is decreasing in $\frac{m_n}{m_1}$ and $\frac{m_n}{m_1}<+\infty$.

\textbf{3)} In the third case, $\Phi'$ is $\delta_i=1, i=1,\ldots, j$, $\delta_i=\frac{m_n}{m_i}-\frac{1}{2}, i=j+1,\ldots, n$ by Proposition \ref{pro-opt-P2}. We also have $m_1(1+\delta_1)<m_2(1+\delta_2)<\ldots<m_n(1+\delta_n)$ and $\Delta_{1,1}=\Delta_{1,2}=\ldots=\Delta_{1,j}=1>\ldots>\Delta_{1,n}>0$. For a $m_1$ type of customer with $\Delta\in [\Delta_{1,i+1},\Delta_{1,i}], i=j+1,\ldots, n-1$, we have $\E[C_1(m_1,\Delta)]=\E[C_2(m_1,\Delta)]=\ldots=\E[C_i(m_1,\Delta)]<\E[C_{i'}(m_1,\Delta)], i'>i$ and the customer may choose any contract option among options $j+1,\ldots, i$. Since $m_1(1+\delta_1)<m_2(1+\delta_2)<\ldots<m_i(1+\delta_i)$, we assume type-$m_1$ customer choose contract option $i$ since the estimated capacity of the customer is $m_i(1+\delta_i)$ which causes the worst-case profit for the supplier. For a type-$m_1$ customer with $\Delta\in [0,\Delta_{1,n}]$, we have $\E[C_1(m_1,\Delta)]=\E[C_2(m_1,\Delta)]=\ldots=\E[C_n(m_1,\Delta)]=m_1p_0$ and we assume it chooses option $n$ for the worst -case profit of the supplier. Besides, type-$m_1$ customer with $\Delta\in [\Delta_{1,j+1}, 1]$ subscribes to option $j$ for the worst-case profit of the supplier. Thus,
\begin{eqnarray*}\small
C_1&=&m_j(1+\delta_j)(1-\Delta_{1,j+1})+\sum_{i=j+1}^{n-1}m_i(1+\delta_i)(\Delta_{1,i}-\Delta_{1,i+1})+m_n(1+\delta_n)\Delta_{1,n}\\
   &=&m_j(1+\delta_j)(1-(1-\frac{m_{j+1}(1-\delta_{j+1})}{m_1}))+\sum_{i=j+1}^{n-1}m_i(1+\delta_i)(\frac{m_{i+1}(1-\delta_{i+1})}{m_1}-\frac{m_i(1-\delta_i)}{m_1})\\
   & &+m_n(1+\delta_n)(1-\frac{m_n(1-\delta_n)}{m_1})\\
	 &=&\frac{2m_jm_{j+1}}{m_1}(1-(\frac{m_n}{m_{j+1}}-\frac{1}{2}))+\sum_{i=j+1}^{n-1}\frac{3}{2}m_i(\frac{1}{2}+\frac{m_n}{m_i})\frac{m_{i+1}-m_i}{m_1}+\frac{3}{2}m_n(1-\frac{m_n}{2m_1})\\
	 &=&\frac{3m_jm_{j+1}}{m_1}-\frac{2m_jm_n}{m_1}+\sum_{i=j+1}^{n-1}(\frac{3m_im_{i+1}}{4m_1}-\frac{3m^2_i}{4m_1}+\frac{3m_nm_{i+1}}{2m_1}-\frac{3m_nm_i}{2m_1})+\frac{3}{2}m_n-\frac{3m^2_n}{4m_1}.
\end{eqnarray*}
Thus, to compute the lower bound of $\frac{2m_n-C_1}{2m_n-\frac{3}{2}m_1}$, it suffices to obtain the upper bound of $C_1$. To do this, we have
\begin{eqnarray*}\small
C_1&=&\frac{3m_jm_{j+1}}{m_1}-\frac{2m_jm_n}{m_1}+(-\frac{3}{8m_1})\sum_{i=j+1}^{n-1}(-2m_im_{i+1}+2m^2_i-4m_nm_{i+1}+4m_nm_i)+\frac{3}{2}m_n-\frac{3m^2_n}{4m_1}\\
   &=&\frac{3m_jm_{j+1}}{m_1}-\frac{2m_jm_n}{m_1}+(-\frac{3}{8m_1})\big((m^2_{j+1}-m^2_n)+\sum_{i=j+1}^{n-1}(m_i-m_{i+1})^2+\sum_{i=j+1}^{n-1}4m_n(m_i-m_{i+1})\big)+\frac{3}{2}m_n-\frac{3m^2_n}{4m_1}\\
	&\leq&\frac{3m_jm_{j+1}}{m_1}-\frac{2m_jm_n}{m_1}+(-\frac{3}{8m_1})\big((m^2_{j+1}-m^2_n)+\frac{(m_n-m_{j+1})^2}{n-j-1}+4m_n(m_{j+1}-m_n)\big)+\frac{3}{2}m_n-\frac{3m^2_n}{4m_1}\\
	&\leq&\frac{3}{2}m_n-\frac{3m^2_n}{4m_1}+\frac{3}{8(n-j-1)m_1}((5n-5j-6)m^2_n-(n-j)m^2_{j+1}\\
  & &+(-4n+4j+6)m_{j+1}m_n)+\frac{3m_j m_{j+1}}{m_1}-\frac{2m_jm_n}{m_1}.
\end{eqnarray*}
The first inequality follows by the inequality $\sum_{i=1}^{n-1}(m_{i+1}-m_i)^2\geq\frac{(m_n-m_1)^2}{n-1}$, which has been proved in the first case.

To compute the upper bound of $C_1$, we let $D=\frac{3}{2}m_n-\frac{3m^2_n}{4m_1}+\frac{3}{8(n-j-1)m_1}((5n-5j-6)m^2_n-(n-j)m^2_{j+1}+(-4n+4j+6)m_{j+1}m_n)+\frac{3m_j m_{j+1}}{m_1}-\frac{2m_jm_n}{m_1}$. We have $\frac{\partial D}{\partial m_j}=\frac{3m_{j+1}-2m_n}{m_1}$. Since $\frac{m_n}{m_{j+1}}\leq\frac{3}{2}$, we have $\frac{\partial D}{\partial m_j}\geq 0$. Thus, since $m_j\leq\frac{2}{3}m_n$ and $C_1\leq D$, we have
\begin{eqnarray*}
 C_1&\leq& \frac{3}{2}m_n+(-\frac{25}{12}+\frac{3(5n-5j-6)}{8(n-j-1)})\frac{m^2_n}{m_1}-\frac{3(n-j)}{8(n-j-1)}\frac{m^2_{j+1}}{m_1}+\frac{4n-4j+2}{8(n-j-1)}\frac{m_{j+1}m_n}{m_1}\\
    &\leq&\frac{3}{2}m_n+(-\frac{25}{12}+\frac{3(5n-5j-6)}{8(n-j-1)}+\frac{(2n-2j+1)^2}{24(n-j-1)(n-j)})\frac{m^2_n}{m_1}\\
    &\leq&\frac{3}{2}m_n-\frac{1}{24}(1+\frac{1}{n-j})\frac{m^2_n}{m_1}\\
    &\leq&\frac{3}{2}m_n-\frac{m^2_n}{24m_1}.
\end{eqnarray*}
The second inequality follows since the term on the right of the first inequality obtains the maximum value at $m_{j+1}=\frac{2n-2j+1}{3(n-j)}m_n\in (\frac{2}{3}m_n, m_n)$. The third inequality follows since $-\frac{25}{12}+\frac{3(5n-5j-6)}{8(n-j-1)}+\frac{(2n-2j+1)^2}{24(n-j-1)(n-j)}=-\frac{1}{24}(1+\frac{1}{n-j})$. The last inequality follows since the term on the right of the second inequality is increasing in $n-j$ and $n-j<+\infty$. Thus, we have 
\begin{eqnarray*}\tiny
  \frac{2m_n-C_1}{2m_n-\frac{3}{2}m_1} &\geq&\frac{\frac{m_n}{2}+\frac{m^2_n}{24m_1}}{2m_n-\frac{3}{2}m_1}\\
   &\geq&\frac{\frac{m_n}{2m_1}+\frac{m^2_n}{24m^2_1}}{\frac{2m_n}{m_1}-\frac{3}{2}}\\
   &\geq& 0.466.
\end{eqnarray*}
To prove the last inequality, let $y=\frac{\frac{x}{2}+\frac{x^2}{24}}{2x-\frac{3}{2}}, x=\frac{m_n}{m_1}\in [\frac{3}{2},2]$. We have $y'(x)=\frac{\frac{1}{12}x^2-\frac{1}{8}x-\frac{3}{4}}{(2x-\frac{3}{2})^2}$. Since $\frac{1}{12}x^2-\frac{1}{8}x-\frac{3}{4}, x\in [\frac{3}{2},2]$ is increasing in $x$ and $\frac{1}{12}x^2-\frac{1}{8}x-\frac{3}{4}=-\frac{2}{3}<0$ at $x=2$. Thus, $y'(x)<0$ for $x\in [\frac{3}{2},2]$. Thus, $y\geq y(2)\geq\frac{7}{15}\geq 0.466$ and the last inequality follows.

Overall, the lemma holds. 
\end{proof}

\vspace{0.5cm}
\hspace{-0.4cm}\textbf{Proof of Lemma \ref{lem-better-solution}}

To prove the lemma, we will first prove Lemma \ref{lem-feasible-solution}. To prove the lemma, we will first prove Claim \ref{claim-feasible} as follows.
Recall that $\Delta_{i,j}$ is defined as the maximum variability degree $\Delta$ of type-$m_i$ customer whose whole usage is included within the contract range of option $j$.
\begin{claim}\label{claim-feasible}
At solution $\Phi'$, consider any type-$m_i$ of customer and two contract options $i$ and $i',i'\neq i$ with $m_i\in [m_{i'}(1-\delta_{i'}),m_{i'}(1+\delta_{i'})]$. Then we have if $\delta_i<1$, then $\Delta_{i,i'}<\delta_i$.
\end{claim}
\begin{proof}
There are two cases. In the first case, we consider $\frac{m_n}{m_1}<\frac{3}{2}$. By Lemma \ref{pro-opt-P2}, we have $\delta_{i}=\frac{m_n}{m_{i}}-\frac{1}{2}<1$. There are two subcases for $i'$. If $i'<i$, then $\Delta_{i,i'}=\frac{m_{i'}(1+\delta_{i'})}{m_i}-1$. Thus, since $\delta_{i'}=\frac{m_n}{m_{i'}}-\frac{1}{2}$ by Lemma \ref{pro-opt-P2}, we have $\Delta_{i,i'}-\delta_i=\frac{m_{i'}-m_i}{2m_i}< 0$ since $m_{i'}< m_i$. Otherwise, if $i'>i$, we have $\Delta_{i,i'}=1-\frac{m_{i'}(1-\delta_{i'})}{m_i}$. Thus, since $\delta_{i'}=\frac{m_n}{m_{i'}}-\frac{1}{2}$ by Lemma \ref{pro-opt-P2}, we have $\Delta_{i,i'}-\delta_i=\frac{3(m_i-m_{i'})}{2m_i}< 0$ since $m_i< m_{i'}$. Thus, we have $\Delta_{i,i'}<\delta_i$ for any $i'\neq i$.

In the second case, we consider $\frac{m_n}{m_{j}}\geq\frac{3}{2}$ and $\frac{m_n}{m_{j+1}}<\frac{3}{2}$, $j=1,\ldots, n-1$. Firstly, we consider $i$ with $i\leq j$. By Lemma \ref{pro-opt-P2}, we have $\delta_i=1$. Next, we consider $i$ with $i\geq j+1$. We have $\delta_i=\frac{m_n}{m_i}-\frac{1}{2}<1$ by Lemma \ref{pro-opt-P2}. There are three subcases for $i'$. In the first subcase that $i'\leq j$, we have $\delta_{i'}=1$, $i'<i$ and $\Delta_{i,i'}=\frac{m_{i'}(1+\delta_{i'})}{m_i}-1$. Thus, we have $\Delta_{i,i'}-\delta_i=\frac{4m_{i'}-m_i-2m_n}{2m_i}$. Since $m_{i'}\leq \frac{2}{3}m_n$ and $m_i>\frac{2}{3}m_n$, we have $4m_{i'}-m_i-2m_n<0$. Thus, we have $\Delta_{i,i'}<\delta_i$. In the second subcase that $j+1\leq i'<i$, we have $\delta_{i'}=\frac{m_n}{m_{i'}}-\frac{1}{2}$, and $\Delta_{i,i'}=\frac{m_{i'}(1+\delta_{i'})}{m_i}-1$. Thus, since $\delta_i=\frac{m_n}{m_i}-\frac{1}{2}$, we have $\Delta_{i,i'}-\delta_i=\frac{m_{i'}-m_i}{2m_i}<0$ since $m_{i'}<m_i$. In the third subcase that $i'>i$, we have $\Delta_{i,i'}=1-\frac{m_{i'}(1-\delta_{i'})}{m_i}$. Thus, since $\delta_{i'}=\frac{m_n}{m_{i'}}-\frac{1}{2}$ and $\delta_{i}=\frac{m_n}{m_{i}}-\frac{1}{2}$ by Lemma \ref{pro-opt-P2}, we have $\Delta_{i,i'}-\delta_i=\frac{3(m_i-m_{i'})}{2m_i}< 0$ since $m_i< m_{i'}$.

Overall, we have if $\delta_i<1$, then $\Delta_{i,i'}<\delta_i$ and the lemma follows.

\end{proof}
This claim is important because it implies that at a solution $\Phi''$ (around $\Phi'$), $\E[C_i(m_i,\Delta)]\leq \E[C_{i'}(m_i,\Delta)]$ for all $\Delta\in [0,\Delta_{th,i}]$. Without the property (i.e., assume $\Delta_{i,i'}>\delta_i$), at the solution $\Phi''$ (around $\Phi'$), type-$m_i$ of customer's cost when choosing option $i'$ is lower than that choosing its right option $i$ for some $\Delta\in [\delta_i,\Delta_{th,i}]$, which breaks the IC constraint (\ref{eqn-mechanism-design}). For example, Figure \ref{non-feasible-contract} shows the cost function $\E[C_i(m_i,\Delta)]$ and $\E[C_{i'}(m_i,\Delta)]$ at the solution near $\Phi'$ assuming $\Delta_{i,i'}>\delta_i$. We see $\E[C_{i'}(m_i,\Delta)]< \E[C_i(m_i,\Delta)]$ at $\Delta=\delta_i+\epsilon'<\Delta_{th,i}$. This is because $\E[C_i(m_i,\Delta)]$ is larger than $m_ip$ at $\Delta=\delta_i+\epsilon'<\Delta_{i,i'}$ ($\epsilon'$ is a very small positive real number) while $\E[C_{i'}(m_i,\Delta)]$ is still $m_ip$ at the point since $\Delta_{i,i'}>\delta_i+\epsilon'$.  
\begin{figure}[]
\centering
\includegraphics[height=6cm, width=8cm]{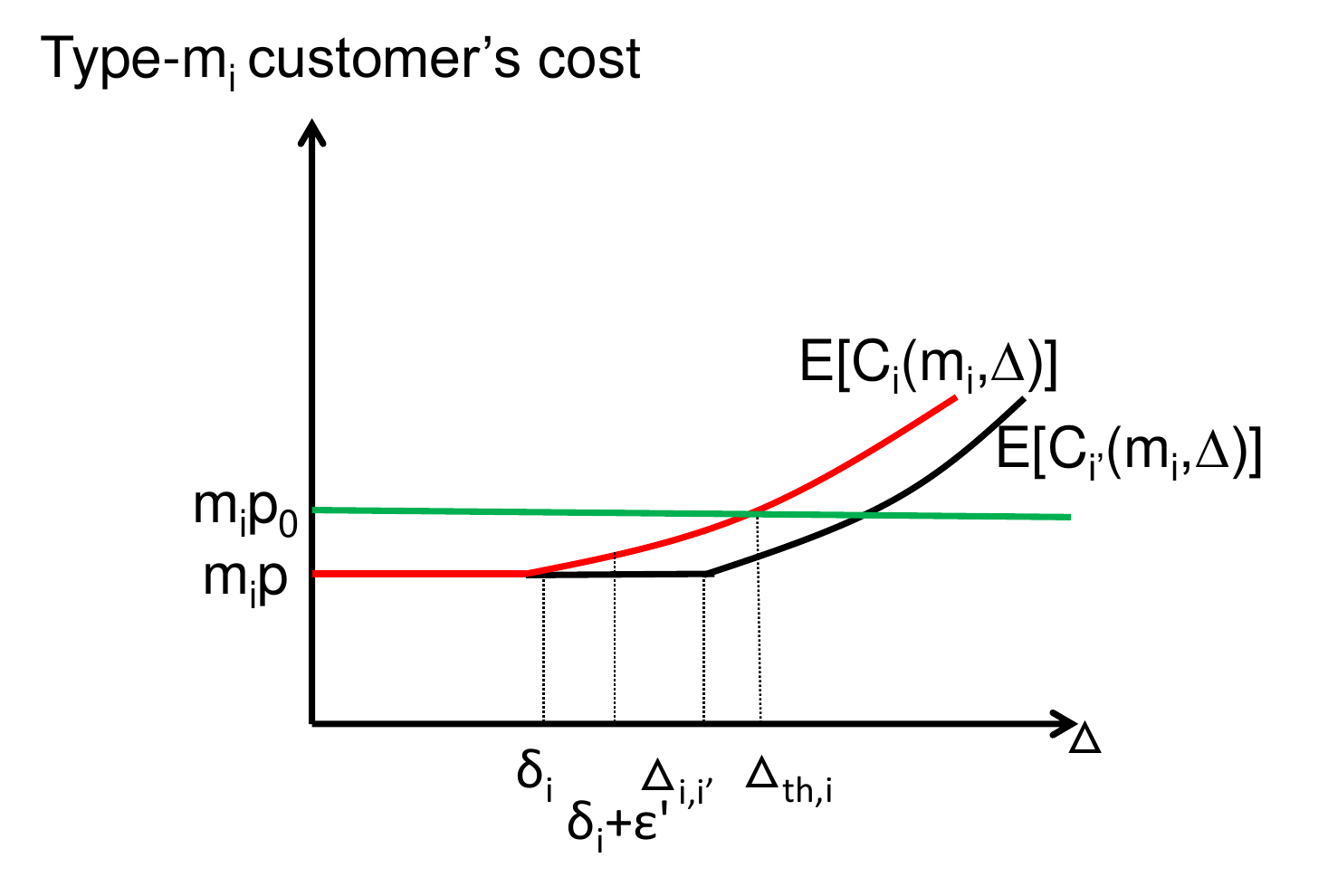}
\caption{Illustration of the consequence without the property at $\Phi'$ in Claim \ref{claim-feasible}}
\label{non-feasible-contract}
\end{figure}

By Claim \ref{claim-monotonicity} and Claim \ref{claim-feasible}, we show there exists an $\hat{\epsilon}>0$ such that for all $0<\epsilon<\hat{\epsilon}$, the solutions near $\Phi'$ are feasible.
\begin{lemma}\label{lem-feasible-solution}
There exist an $\hat{\epsilon}>0$ such that for all $0<\epsilon<\hat{\epsilon}$, the solutions $\Phi''$ (around $\Phi'$) are feasible.
\end{lemma}
\begin{proof}
To find a feasible solution $\Phi''$ (around $\Phi'$), we need to make sure that $\E[C_i(m_i,\Delta)]\leq \E[C_{i'}(m_i,\Delta)]$ for any $\Delta\in [0,\Delta_{th,i}]$ at the solution. We will prove that there exist a range $[0,\epsilon_{th,i,i'}], \epsilon_{th,i,i'}>0$ of $\epsilon$ (i.e. $\epsilon\in [0,\epsilon_{th,i,i'}]$) such that $\E[C_i(m_i,\Delta)]\leq \E[C_{i'}(m_i,\Delta)]$ under the solution near $\Phi'$ (each contract option's price is $p_0-\epsilon$). There are two cases. In the first case that $\delta_i<1$ at solution $\Phi'$, we first prove $\E[C_{i'}(m_i,\Delta)](\Delta=\delta_i, \epsilon=0)>m_ip_0$. To do this, there are two subcases depending on whether $m_i\in [m_{i'}(1-\delta_{i'}), m_{i'}(1+\delta_{i'})]$ or not. In the first subcase that $m_i\notin [m_{i'}(1-\delta_{i'}), m_{i'}(1+\delta_{i'})]$, we have $\E[C_{i'}(m_i,\Delta)](\Delta=\delta_i,\epsilon=0)>m_ip_0$ by Claim \ref{claim-monotonicity}. In the second subcase that $m_i\in [m_{i'}(1-\delta_{i'}), m_{i'}(1+\delta_{i'})]$, we have $\Delta_{i,i'}<\delta_i$ by Claim \ref{claim-feasible}. We can prove $\E[C_{i'}(m_i,\Delta)](\Delta=\delta_i, \epsilon=0)>m_ip_0$ as illustrated in Figure \ref{fig-feasible-contract-illustration}. We have $\E[C_{i'}(m_i,\Delta)](\Delta=\Delta_{i,i'}, \epsilon=0)=m_ip_0$. Thus, since $\Delta_{i,i'}<\delta_i$ by Claim \ref{claim-feasible} and $\E[C_{i'}(m_i,\Delta)]$ is strictly increasing in $\Delta\in [\Delta_{i,i'}, \delta_i]$ by Claim \ref{claim-monotonicity}, we have $\E[C_{i'}(m_i,\Delta)](\Delta=\delta_i, \epsilon=0)>m_ip_0$. Thus, since $\E[C_{i'}(m_i,\Delta)](\Delta=\delta_i,\epsilon\geq 0)$ is continuous as a function of $\epsilon$, there exists a range $[0,\epsilon_{th,i,i'}], \epsilon_{th,i,i'}>0$ of $\epsilon$ such that $\E[C_{i'}(m_i,\Delta)](\Delta=\delta_i,\epsilon\in [0,\epsilon_{th,i,i'}])> m_ip_0$ as illustrated in Figure \ref{fig-feasible-contract-illustration}. Under such a $\epsilon\in [0,\epsilon_{th,i,i'}]$, we will show that $\E[C_i(m_i,\Delta)]\leq \E[C_{i'}(m_i,\Delta)]$ for any $\Delta\in [0,\Delta_{th,i}]$ as illustrated in Figure \ref{fig-feasible-contract-illustration}. Since at the solution near $\Phi'$, $\Delta_{th,i}=\min(1,\frac{k\delta_i+2(p_0-p)+\sqrt{(k\delta_i+2(p_0-p))^2-k^2\delta^2_i}}{k})$ and $p=p_0-\epsilon$, we have $\Delta_{th,i}\geq\delta_i$. For $\Delta\in [0,\delta_i]$, we have $\E[C_{i'}(m_i,\Delta)]\geq m_ip$ by Claim \ref{claim-monotonicity} and $\E[C_i(m_i,\Delta)]=m_ip$, and therefore $\E[C_{i'}(m_i,\Delta)]\geq \E[C_i(m_i,\Delta)]$. If $\Delta_{th,i}>\delta_i$, for $\Delta\in (\delta_i,\Delta_{th,i}]$, since $\Delta_{i,i'}<\delta_i$ by Claim \ref{claim-feasible}, $\E[C_{i'}(m_i,\Delta)]$ is strictly increasing in $\Delta\in [\delta_i, \Delta_{th,i}]$ by Claim \ref{claim-monotonicity} and $\epsilon\in [0, \epsilon_{th,i,i'}]$, we have $\E[C_{i'}(m_i,\Delta)]> \E[C_{i'}(m_i,\Delta)](\Delta=\delta_i, \epsilon\in [0, \epsilon_{th,i,i'}])\geq m_ip_0$. Thus, since $\E[C_i(m_i,\Delta)]\leq m_ip_0$ for $\Delta\in [\delta_i,\Delta_{th,i}]$ by the definition of $\Delta_{th,i}$, we have $\E[C_{i'}(m_i,\Delta)]> \E[C_i(m_i,\Delta)]$ for $\Delta\in (\delta_i,\Delta_{th,i}]$. In the second case that $\delta_i=1$ at solution $\Phi'$, we have $\Delta_{th,i}=\delta_i=1$. For any solution near $\Phi'$, we have $\E[C_i(m_i,\Delta)]=m_ip$ for $\Delta\in [0,\Delta_{th,i}]$. Thus, since $\E[C_{i'}(m_i,\Delta)]\geq m_ip$ for $\Delta\in [0,\Delta_{th,i}]$ at the solution by Claim \ref{claim-monotonicity}, we have $\E[C_{i'}(m_i,\Delta)]\geq \E[C_i(m_i,\Delta)]$. 

For each pair of $i,i'\in [1,n], i'\neq i$, there exists a range $[0,\epsilon_{th,i,i'}]$ of $\epsilon$ such that $\E[C_i(m_i,\Delta)]\leq \E[C_{i'}(m_i,\Delta)]$ under the solutions $\Phi''$ (around $\Phi'$). Let $\hat{\epsilon}=\min\limits_{i,i'\in[1,n]}\epsilon_{th,i,i'}$. $\forall \epsilon\in [0,\hat{\epsilon}]$, we have $\E[C_{i'}(m_i,\Delta)]\geq \E[C_i(m_i,\Delta)]$ for any $\Delta\in [0,\Delta_{th,i}]$. Thus, the lemma holds. 
\end{proof}
\begin{figure}[]
\vspace{-0.1in}
\centering
\includegraphics[height=6cm, width=8cm]{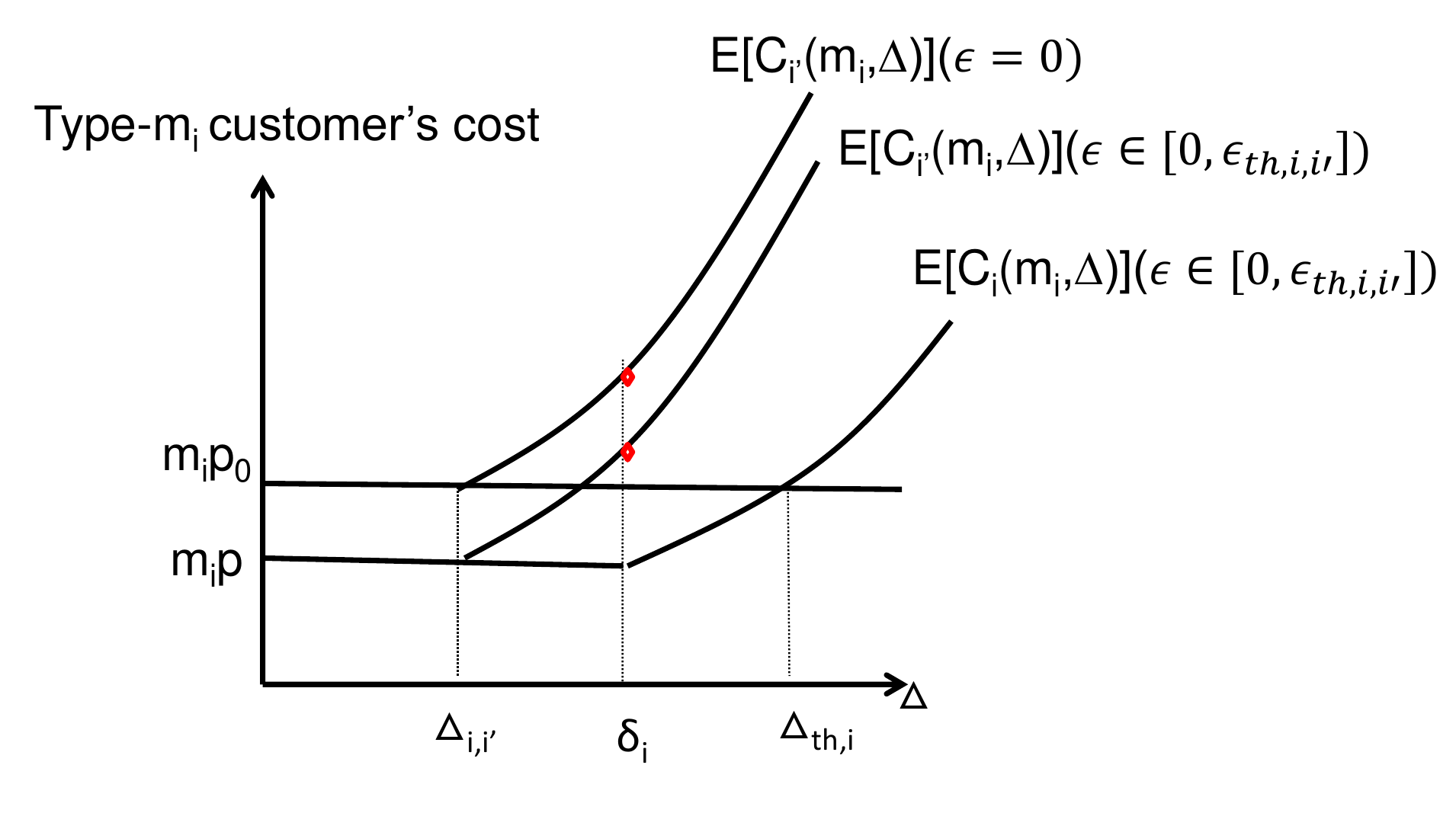}
\caption{Illustration of the intuition of proof in Lemma \ref{lem-feasible-solution} }
\label{fig-feasible-contract-illustration}
\end{figure}
Next, we will prove Lemma \ref{lem-better-solution}.

\begin{proof}
As shown in Lemma \ref{lem-feasible-solution}, there exists an $\hat{\epsilon}>0$ such that for all $0<\epsilon<\hat{\epsilon}$, the solutions $\Phi''$ (around $\Phi'$) are feasible. Under these feasible solutions, we have 
\begin{eqnarray}
\Delta_{th,i}&=&\min(1,\frac{k\delta_i+2(p_0-p)+\sqrt{(k\delta_i+2(p_0-p))^2-k^2\delta^2_i}}{k})\nonumber\\
             &=&\min(1,\delta_i+\frac{2\epsilon}{k}+\frac{2\sqrt{\epsilon^2+k\delta_i\epsilon}}{k})
\end{eqnarray}
where $\delta_i=\min(1,\frac{m_n}{m_i}-\frac{1}{2})$.
If $\delta_i=1$, we have $\Delta_{th,i}=1$. Otherwise if $0\leq \delta_i<1$, we have 
\begin{equation}\label{eqn-Delta-th-i-epsilon}
\Delta_{th,i}=
\begin{cases}
\delta_i+\frac{2\epsilon}{k}+\frac{2\sqrt{\epsilon^2+k\delta_i\epsilon}}{k},0\leq\epsilon\leq\frac{k(1-\delta_i)^2}{4}\\
1, \epsilon>\frac{k(1-\delta_i)^2}{4}
\end{cases}
\end{equation}
For any $0\leq\delta_i\leq 1$, we have $\Delta_{th,i}$ as a function of $\epsilon$ is continuous at $\epsilon=0$.

To compute the supplier's profit at solution (each option's price is $p_0-\epsilon$, $\epsilon\in [0, \epsilon_{th}]$) near $\Phi'$, we need to compute the supplier's revenue, and energy generation cost, capacity cost, respectively. To to do this, we first analyze the option choice of each type-$m_i$ customer with $\Delta$. For each type-$m_i$ customer with $\Delta\leq\delta_i$, it may choose option $i$ or other option $i'$. If it chooses other option $i'$, it must have $m_i\in [m_{i'}(1-\delta_{i'}),m_{i'}(1+\delta_{i'})]$ and $\Delta\leq\Delta_{i,i'}$. To show this, we have $\E[C_i(m_i,\Delta)]=m_ip$. If $m_i\notin [m_{i'}(1-\delta_{i'}),m_{i'}(1+\delta_{i'})]$, we have $\E[C_{i'}(m_i,\Delta)]>m_ip$ by Claim \ref{claim-monotonicity} and $\E[C_{i'}(m_i,\Delta)]>\E[C_i(m_i,\Delta)]$, and thus the customer will not choose option $i'$. If $m_i\in [m_{i'}(1-\delta_{i'}),m_{i'}(1+\delta_{i'})], \Delta>\Delta_{i,i'}$, we have $\E[C_{i'}(m_i,\Delta)]>m_ip$ by Claim \ref{claim-monotonicity}, and $\E[C_{i'}(m_i,\Delta)]>\E[C_i(m_i,\Delta)]$, and thus it will not choose option $i'$. Thus, we have if type-$m_i$ customer with $\Delta\in [0,\delta_i]$ chooses other option $i'$, it must have $m_i\in [m_{i'}(1-\delta_{i'}),m_{i'}(1+\delta_{i'})]$ and $\Delta\leq\Delta_{i,i'}$. If $\Delta_{th,i}>\delta_i$, for each type-$m_i$ customer with $\Delta\in (\delta_i,\Delta_{th,i}]$, it chooses option $i$. This is because by Lemma \ref{lem-feasible-solution}, we have $\E[C_{i'}(m_i,\Delta\in (\delta_i,\Delta_{th,i}],\epsilon\in [0,\epsilon_{th}])]>m_ip_0$ and $\E[C_i(m_i,\Delta\in (\delta_i,\Delta_{th,i}],\epsilon\in [0,\epsilon_{th}])]\leq m_ip_0$, and thus $\E[C_{i'}(m_i,\Delta\in (\delta_i,\Delta_{th,i}],\epsilon\in [0,\epsilon_{th}])]>\E[C_i(m_i,\Delta\in (\delta_i,\Delta_{th,i}],\epsilon\in [0,\epsilon_{th}])]$. Thus, it must choose option $i$. For each type-$m_i$ customer with $\Delta>\Delta_{th,i}$, it chooses baseline pricing scheme.

We now compute the supplier's revenue. For type-$m_i$ of customer with $\Delta\in [0,\delta_i]$, it may choose option $i$ and other options. If it choose other option $i'$, its whole demand range must be within option $i'$ contract range, and thus it contributes $m_ip_i$ amount of payment to the supplier. Thus, the supplier's total revenue from type-$m_i$ of customers is $\int_{0}^{\Delta_{th,i}}f(\Delta)m_ip_id\Delta+\int_{\Delta_{th,i}}^{1}f(\Delta)m_ip_0d\Delta=m_ip_i\Delta_{th,i}+m_ip_0(1-\Delta_{th,i})$.  Next, we compute the supplier's energy generation cost. For type-$m_i$ of customer with $\Delta\in [0,\delta_i]$, its mean demand is $m_i$ and thus the supplier's total energy generation cost from type-$m_i$ of customers is $c_0m_i$. 

We now analyze the estimated capacity of a type-$m_i$ customer considering it may choose other types of options. At the solution $\Phi''$ (around $\Phi'$), let $\bar{C}_i$ be the expected estimated capacity of type-$m_i$ of customers.
Similarly with the analysis of the estimated capacity of $m_1$ type customers by the supplier, we can derive each $C_i$. Note that at the solution near $\Phi'$, we can derive $m_1(1+\delta_1)<m_2(1+\delta_2)<\ldots<m_n(1+\delta_n)$ and thus type-$m_i$ customers does not choose option $i', (i'<i)$, since we assume the customer chooses the contract option that results in the lowest profit to the supplier in the worst case as mentioned in Section \ref{sec-challenge-II}. We also have $m_1(1-\delta_1)\leq m_2(1-\delta_2)\leq\ldots\leq m_n(1-\delta_n)$. To compute $\bar{C}_i$, there are three cases.

In the first case, consider $\frac{m_n}{m_1}\leq\frac{3}{2}$. We have $\Delta_{i,i'}=1-\frac{m_{i'}(1-\delta_{i'})}{m_i}$, $i'=i,\ldots, n$. Thus, since $\delta_{i'}=\frac{m_n}{m_{i'}}-\frac{1}{2}, i'=i,\ldots, n$, we have $\delta_i=\Delta_{i,i}>\Delta_{i,i+1}>\ldots>\Delta_{i,n}>0$. For a type-$m_i$ of customer with $\Delta\in [\Delta_{i,i'+1},\Delta_{i,i'}], i'=i,\ldots, n-1$, we have $\E[C_i(m_i,\Delta)]=\E[C_{i+1}(m_i,\Delta)]=\ldots=\E[C_{i'}(m_i,\Delta)]=m_ip<\E[C_{i''}(m_i,\Delta)], i''>i'$. The supplier can choose any option among options $i,\ldots, i'$. Since $m_i(1+\delta_i)<m_{i+1}(1+\delta_{i+1})<\ldots<m_{i'}(1+\delta_{i'})$, we assume type-$m_i$ of customer chooses contract option $i'$ since the estimated capacity of the customer is $m_{i'}(1+\delta_{i'})$ which is the worst for the supplier, the revenue it contributes to the supplier is $m_ip$ whatever option it chooses, the mean energy consumption it results in is $m_i$ whatever option it chooses, and therefore the profit it chooses option $i'$ is the worst. For a type-$m_i$ of customer with $\Delta\in [0,\Delta_{i,n}]$, we have $\E[C_i(m_i,\Delta)]=\E[C_{i+1}(m_i,\Delta)]=\ldots=\E[C_n(m_i,\Delta)]=m_ip$ and we assume it chooses option $n$ for the worst-case profit of the supplier. Besides, type-$m_i$ customer with $\Delta\in [\delta_i, \Delta_{th,i}]$ subscribes to option $i$. Moreover, type-$m_i$ customer with $\Delta\in [\Delta_{th,i}, 1]$ subscribes to the baseline pricing scheme. Thus, we have
\begin{eqnarray*}
\bar{C}_i&=&2m_n(1-\Delta_{th,i})+m_i(1+\delta_i)(\Delta_{th,i}-\delta_i)+\sum_{i'=i}^{n-1}m_{i'}(1+\delta_{i'})(\Delta_{i,i'}-\Delta_{i,i'+1})+m_n(1+\delta_n)\Delta_{i,n}\\
   &=&2m_n(1-\Delta_{th,i})+m_i(1+\delta_i)(\Delta_{th,i}-\Delta_{i,i+1})+\sum_{i'=i+1}^{n-1}m_{i'}(1+\delta_{i'})(\Delta_{i,i'}-\Delta_{i,i'+1})+m_n(1+\delta_n)\Delta_{i,n}\\
\end{eqnarray*} 
The last equality follows since $\Delta_{i,i}=\delta_i$.

In the second case, consider $\frac{m_n}{m_i}\geq 2$. Assume $\frac{m_n}{m_j}>\frac{3}{2}$ and $\frac{m_n}{m_{j+1}}\leq\frac{3}{2}, j=i,\ldots, n-1$. We have $\delta_{i'}=1, i'=i,\ldots, j$, $\delta_{i'}=\frac{m_n}{m_{i'}}-\frac{1}{2}, i'=j+1,\ldots, n$ by Proposition \ref{pro-opt-P2}. We also have $m_i(1-\delta_i)\leq m_{i+1}(1-\delta_{i+1})\leq\ldots\leq m_n(1-\delta_n)$. Assume $m_i\geq m_{j'}(1-\delta_{j'})$ and $m_i<m_{j'+1}(1-\delta_{j'+1})$, $j'=j+1, \ldots, n-1$. We have $\delta_i=\Delta_{i,i}=\Delta_{i,i+1}=\ldots=\Delta_{i,j}=1>\ldots>\Delta_{i,j'}\geq 0$. Since $\E[C_{i'}(m_i,\Delta)]>m_ip, i'=j'+1,\ldots, n$ by Claim \ref{claim-monotonicity} and $\E[C_i(m_i,\Delta)]=m_ip$ for $\Delta\in [0,1]$, type-$m_i$ of customer will not choose options $j'+1,\ldots, n-1$. For a type-$m_i$ of customer with $\Delta\in [\Delta_{i,i'+1},\Delta_{i,i'}], i'=j+1,\ldots, j'-1$, we have $\E[C_i(m_i,\Delta)]=\E[C_{i+1}(m_i,\Delta)]=\ldots=\E[C_{i'}(m_i,\Delta)]<\E[C_{i''}(m_i,\Delta)], i''>i'$ and the customer may choose any contract option among options $i,\ldots, i'$. Since $m_i(1+\delta_i)<m_{i+1}(1+\delta_{i+1})<\ldots<m_{i'}(1+\delta_{i'})$, we assume type-$m_i$ of customer choose contract option $i'$ since the estimated capacity of the customer is $m_{i'}(1+\delta_{i'})$ which causes the worst-case profit for the supplier. For a type-$m_i$ of customer with $\Delta\in [0,\Delta_{i,j'}]$, we have $\E[C_i(m_i,\Delta)]=\E[C_{i+1}(m_i,\Delta)]=\ldots=\E[C_{j'}(m_i,\Delta)]=m_ip$ and we assume it chooses option $j'$ for the worst-case profit of the supplier. Besides, type-$m_i$ customer with $\Delta\in [\Delta_{i,j+1}, 1]$ subscribes to option $j$ for the worst case-profit of the supplier. 
\begin{eqnarray*}
\bar{C}_i&=&m_j(1+\delta_j)(1-\Delta_{i,j+1})+\sum_{i'=j+1}^{j'-1}m_{i'}(1+\delta_{i'})(\Delta_{i,{i'}}-\Delta_{i,i'+1})+m_{j'}(1+\delta_{j'})\Delta_{i,j'}\\
 &=&2m_n(1-\Delta_{th,i})+m_i(1+\delta_i)(\Delta_{th,i}-1)+m_j(1+\delta_j)(1-\Delta_{i,j+1})\\
 & &+\sum_{i'=j+1}^{j'-1}m_{i'}(1+\delta_{i'})(\Delta_{i,{i'}}-\Delta_{i,i'+1})+m_{j'}(1+\delta_{j'})\Delta_{i,j'}\\
\end{eqnarray*}
The second equality follows since $\Delta_{th,i}=\delta_i=1$. 

In the third case, consider $\frac{3}{2}<\frac{m_n}{m_i}\leq 2$. Assume $\frac{m_n}{m_j}>\frac{3}{2}$ and $\frac{m_n}{m_{j+1}}\leq\frac{3}{2}, j=i,\ldots, n-1$. We have $\delta_{i'}=1, i'=i,\ldots, j$, $\delta_{i'}=\frac{m_n}{m_{i'}}-\frac{1}{2}, i'=j+1,\ldots, n$ by Proposition \ref{pro-opt-P2}. We also have $\delta_i=\Delta_{i,i}=\Delta_{i,i+1}=\ldots=\Delta_{i,j}=1>\ldots>\Delta_{i,n}>0$. For a type-$m_i$ of customer with $\Delta\in [\Delta_{i,i'+1},\Delta_{i,i'}], i'=j+1,\ldots, n-1$, we have $\E[C_i(m_i,\Delta)]=\E[C_{i+1}(m_i,\Delta)]=\ldots=\E[C_{i'}(m_i,\Delta)]<\E[C_{i''}(m_i,\Delta)], i''>i'$ and the customer may choose any contract option among options $j+1,\ldots, i'$. Since $m_i(1+\delta_i)<m_{i+1}(1+\delta_{i+1})<\ldots<m_{i'}(1+\delta_{i'})$, we assume type-$m_i$ of customer choose contract option $i'$ since the estimated capacity of the customer is $m_{i'}(1+\delta_{i'})$ which causes the worst-case profit for the supplier. For a type-$m_i$ of customer with $\Delta\in [0,\Delta_{i,n}]$, we have $\E[C_i(m_i,\Delta)]=\E[C_{i+1}(m_i,\Delta)]=\ldots=\E[C_n(m_i,\Delta)]=m_ip$ and we assume it chooses option $n$ for the worst-case profit of the supplier. Besides, type-$m_i$ customer with $\Delta\in [\Delta_{i,j+1}, 1]$ subscribes to option $j$ for the worst-case profit of the supplier. Thus, we have
\begin{eqnarray*}
\bar{C}_i&=&m_j(1+\delta_j)(1-\Delta_{1,j+1})+\sum_{i'=j+1}^{n-1}m_{i'}(1+\delta_{i'})(\Delta_{i,i'}-\Delta_{i,i'+1})+m_n(1+\delta_n)\Delta_{i,n}\\
   &=&2m_n(1-\Delta_{th,i})+m_i(1+\delta_i)(\Delta_{th,i}-1)+m_j(1+\delta_j)(1-\Delta_{1,j+1})\\
	 & &+\sum_{i'=j+1}^{n-1}m_{i'}(1+\delta_{i'})(\Delta_{i,i'}-\Delta_{i,i'+1})+m_n(1+\delta_n)\Delta_{i,n}\\
\end{eqnarray*}
The second equality follows since $\Delta_{th,i}=\delta_i=1$. 
Overall, we can write 
\begin{equation}
\bar{C}_i=2m_n(1-\Delta_{th,i})+m_i(1+\delta_i)(\Delta_{th,i}-\hat{\Delta}_{th,i})+C'_i  
\end{equation}
where
\begin{equation}
C'_i=
\begin{cases}
\sum_{i'=i+1}^{n-1}m_{i'}(1+\delta_{i'})(\Delta_{i,i'}-\Delta_{i,i'+1})+m_n(1+\delta_n)\Delta_{i,n}, \frac{m_n}{m_i}\leq\frac{3}{2}\\
m_j(1+\delta_j)(1-\Delta_{i,j+1})+\sum_{i'=j+1}^{j'-1}m_{i'}(1+\delta_{i'})(\Delta_{i,{i'}}-\Delta_{i,i'+1})+m_{j'}(1+\delta_{j'})\Delta_{i,j'}, \frac{m_n}{m_i}\geq 2\\
m_j(1+\delta_j)(1-\Delta_{i,j+1})+\sum_{i'=j+1}^{n-1}m_{i'}(1+\delta_{i'})(\Delta_{i,i'}-\Delta_{i,i'+1})+m_n(1+\delta_n)\Delta_{i,n}, \frac{m_n}{m_i}\in [\frac{3}{2}, 2]
\end{cases}
\end{equation}
 and 
\begin{equation}
\hat{\Delta}_{th,i}=
\begin{cases}
\Delta_{i,i+1}, \frac{m_n}{m_i}\leq\frac{3}{2}\\
1,\frac{m_n}{m_i}>\frac{3}{2}
\end{cases}
\end{equation}
Note that $C'_i, \hat{\Delta}_{th,i}$ are independent of $\epsilon$ since $\Delta_{i,i'}$ is independent of $\epsilon$. $C'_i$ can be regarded as the supplier's estimated capacity of type-$m_i$ customers choosing other options, and $\hat{\Delta}_{th,i}$ can be regarded as the threshold of type-$m_i$ of customers which choose option $i$ or other options. In particular, if $\Delta\leq\hat{\Delta}_{th,i}$, the customer chooses other options, and if $\hat{\Delta}_{th,i}\leq\Delta\leq\Delta_{th,i}$, the customer chooses option $i$.

Thus, the supplier's worst-case profit under the solution (each contract option's price is $p_i=p=p_0-\epsilon$) is 
\begin{eqnarray*}
P(\Phi)&=&N(h(m_1)(m_1(p_1\Delta_{th,1}+p_0(1-\Delta_{th,1}))-\hat{c}(C'_1+m_1(1+\delta_1)(\Delta_{th,1}-\hat{\Delta}_{th,1})+2m_n(1-\Delta_{th,1}))-c_0m_1)\\
		& &\ldots\\
    & &+Nh(m_n)(m_n(p_n\Delta_{th,n}+p_0(1-\Delta_{th,n}))-\hat{c}(C'_n+m_n(1+\delta_n)(\Delta_{th,n}-\hat{\Delta}_{th,n})+2m_n(1-\Delta_{th,n}))-c_0m_n)\\
		&=&\sum_{i=1}^{n}Nh(m_i)\bigg(m_i\big((p_0-\epsilon)\Delta_{th,i}+p_0(1-\Delta_{th,i})\big)-\hat{c}\big(C'_i+m_i(1+\delta_i)(\Delta_{th,i}-\hat{\Delta}_{th,i})+2m_n(1-\Delta_{th,i})\big)-c_0m_i\bigg)\\
    &=&\sum_{i=1}^{n}Nh(m_i)\bigg(\Delta_{th,i}\big(-m_i\epsilon-\hat{c}m_i(1+\delta_i)+2m_n\hat{c}\big)+m_ip_0-\hat{c}\big(C'_i-m_i(1+\delta_i)\hat{\Delta}_{th,i}+2m_n\big)-c_0m_i\bigg)\\
\end{eqnarray*}
We have
\begin{equation}
\frac{\partial P(\Phi)}{\partial\epsilon}=\sum_{i=1}^{n} Nh(m_i)\big((-m_i\epsilon-\hat{c}m_i(1+\delta_i)+2m_n\hat{c})\frac{\partial\Delta_{th,i}}{\partial\epsilon}-m_i\Delta_{th,i}\big)
\end{equation}
We will prove $\frac{\partial P(\Phi)}{\partial \epsilon}>0$ at $\epsilon=0$. Note that $\Delta_{th,i}$ is continuous at $\epsilon=0$ for $i\in [1,n]$ and $P(\Phi)$ is continuous at $\epsilon=0$. For each $i\in [1,n-1]$, we have either $\delta_i=1$ or $\delta_i<1$. In the former case, we have $\frac{\partial \Delta_{th,i}}{\partial\epsilon}=0$ at $\epsilon=0$ and in the latter case we have $\frac{\partial \Delta_{th,i}}{\partial\epsilon}=\frac{2}{k}+\frac{2\epsilon+k\delta_i}{k\sqrt{\epsilon^2+k\epsilon\delta_i}}=+\infty$ at $\epsilon=0$ since $\delta_i=\min(\frac{m_n}{m_i}-\frac{1}{2},1)\geq\frac{1}{2}>0$ by Proposition \ref{pro-opt-P2}. Thus, we have $\frac{\partial \Delta_{th,i}}{\partial\epsilon}\geq 0$ for $i\in [1,n-1]$. For $i=n$, we have $\delta_i=\frac{1}{2}$. Thus, by Equation \ref{eqn-Delta-th-i-epsilon}, we have $\frac{\partial \Delta_{th,i}}{\partial\epsilon}=\frac{2}{k}+\frac{2\epsilon+\frac{1}{2}k}{k\sqrt{\epsilon^2+\frac{1}{2}k\epsilon}}=+\infty$ at $\epsilon=0$. Thus, since $-m_i\epsilon-\hat{c}m_i(1+\delta_i)+2m_n\hat{c}=-\hat{c}m_i(1+\delta_i)+2m_n\hat{c}>0$ at $\epsilon=0$, and $m_i\Delta_{th,i}=m_i\delta_i$ is a finite number, we have $\frac{\partial P(\Phi)}{\partial \epsilon}=+\infty>0$ at $\epsilon=0$ for each $i\in [1,n]$. Thus, since $P(\Phi)$ is a continuous function as $\epsilon\geq 0$, there must exist a range $[0,\epsilon_0], \epsilon_0\leq\hat{\epsilon}$ of $\epsilon$ such that the solution $\Phi''$ (around $\Phi'$) is both feasible and produces at least the amount of profit $\Phi'$ results in under the ``pessimistic" assumption. In other words, the lemma holds.
\end{proof}


\begin{thebibliography}{1}
\bibitem{Google-Data-Center} 
Google uses DeepMind AI to cut data center energy bills. \url{https://www.theverge.com/2016/7/21/12246258/google-deepmind-ai-data-center-cooling}, 2016.
\bibitem{IBM-Data-Center}
Data center power usage report. \url{http://t.cn/R9vLxPS}
\bibitem{Kuser}
M. Kuser. NYISO auction shows higher prices for NYC, Hudson Valley. \url{https://www.rtoinsider.com/nyiso-capacity-auction-hudson-valley-41352/}, 2017.
\bibitem{Capacity-market}
Capacity Markets. \url{https://business.directenergy.com/understanding-energy/managing-energy-costs/deregulation-and-energy-pricing/capacity-markets}
\bibitem{Mchich}
A. Mchich, O. Johnson. Introducing the NYISO electricity capacity market. https://www.cmegroup.com/education/articles-and-reports/introducing-the-nyiso-eletricity-capacity-market.html
\bibitem{Lu}
L. Lu, J. Tu, C. Chao, M. Chen, X. Lin. Online energy generation scheduling for microgrids with intermittent energy sources and co-Generation. \emph{SIGMETRICS}, 2013.
\bibitem{Roozbehani}
M. Roozbehani, M. Dahleh, and S. Mitter. Dynamic pricing and stabilization of supply and demand in modern electric power grids. \emph{Proceeding of SmartGridComm}, 2010.
\bibitem{Duan}
L. Duan, R. Zhang. Dynamic contract to regulate energy management in microgrids. \emph{IEEE International Conference on Smart Grid Communications}, 2013.
\bibitem{Wang}
Q. Wang, C. Zhang, Y. Ding, G. Xydis, J. Wang, J. Stergaard. Review of real-time electricity markets for integrating distributed energy resources and demand response. \emph{Applied Energy}, vol. 138, pp. 695-706, 2015. 
\bibitem{Zhao}
S. Zhao, X. Lin, M. Chen. Peak-minimizing online EV charging: price-of-uncertainty and algorithm robustification. \emph{IEEE INFOCOM}, 2015.
\bibitem{Hertzog}
C. Hertzog. Smart grid dictionary plus. \emph{Cengage Learning}, 2012.
\bibitem{LiW}
W. Li, L. Tesfatsion. A swing-contract market design for flexible service provision in electric power systems. \emph{The IMA Volumes in Mathematics and its Applications}, Vol. 162, Springer, 2018.
\bibitem{Fudenberg}
D. Fudenberg and J. Tirole. Game Theory. \emph{The MIT Press}, 1991.
\bibitem{Mas-Colell}
A. Mas-Colell, M. D. Whinston, J. R. Green. Microeconomic theory. \emph{Oxford university press}, 1995.
\bibitem{Gao}
Y. Gao, Y. Chen, C. Wang, K.J.R. Liu. A contract-based approach for ancillary services in V2G networks: optimality and learning. \emph{Proceedings IEEE INFOCOM}, 2013.
\bibitem{Smolaks}
M. Smolaks. Synergy: Number of hyperscale data centers reached 430 in 2018. \url{https://www.datacenterdynamics.com/news/synergy-number-hyperscale-data-centers-reached-430-2018/}, 2019.
\bibitem{Navabi}
S. Navabi,  A. Nayyar. Optimal auction design for flexible consumers. \emph{IEEE Transactions on Control of Network Systems}, Vol. 6, No. 1,  2019.
\bibitem{Peak-price}
Billing summary of conEdison. \url{https://www.coned.com/-/media/files/coned/documents/save-energy-money/using-private-generation/specs-and-tariffs/demandbilledbillandstatement.pdf?la=en}, 2017.
\bibitem{Energy-cost}
Average power plant operating expenses for major U.S. investor-owned electric utilities. \url{https://www.eia.gov/electricity/annual/html/epa_08_04.html}

\end{thebibliography}
\end{document}